\documentclass[11pt]{article}

\newif\ifcomm
\commfalse

\usepackage{diagbox}
\usepackage{soul} 
\usepackage[english]{babel}
\usepackage[utf8x]{inputenc}
\usepackage[T1]{fontenc}

\usepackage{comment}
\usepackage{amsmath, graphicx, latexsym, amssymb, amsthm, amsfonts}
\usepackage{hyperref}
\usepackage{natbib}
\usepackage[mathscr]{eucal}
\usepackage{graphics}
\usepackage{graphicx}
\usepackage[dvipsnames]{xcolor}
\usepackage{epsfig}
\usepackage{bbm}
\usepackage{xspace}
\usepackage{float}
\usepackage{caption}
\usepackage{subcaption}
\ifcomm
  
  \newcommand{\GM}[1]{{({\color{blue}Gal: #1})}}
\else
  
  \newcommand{\GM}[1]{{({\color{blue}Gal: #1})}}
\fi

\begin{document}

\newtheorem{theorem}{Theorem}[section]
\newtheorem{lemma}[theorem]{Lemma}
\newtheorem{corollary}[theorem]{Corollary}
\newtheorem{definition}[theorem]{Definition}
\newtheorem{example}[theorem]{Example}
\newtheorem{proposition}[theorem]{Proposition}
\newtheorem{condition}[theorem]{Condition}
\newtheorem{assumption}[theorem]{Assumption}
\newtheorem{conjecture}[theorem]{Conjecture}
\newtheorem{problem}[theorem]{Problem}
\newtheorem{remark}[theorem]{Remark}
\newtheorem{key observation}[theorem]{Key observation}
\newtheorem{property}[theorem]{Property}


\def\thelemma{\arabic{section}.\arabic{lemma}}
\def\thetheorem{\arabic{section}.\arabic{theorem}}
\def\thecorollary{\arabic{section}.\arabic{corollary}}
\def\thedefinition{\arabic{section}.\arabic{definition}}
\def\theexample{\arabic{section}.\arabic{example}}
\def\theproposition{\arabic{section}.\arabic{proposition}}
\def\thecondition{\arabic{section}.\arabic{condition}}
\def\theassumption{\arabic{section}.\arabic{assumption}}
\def\theconjecture{\arabic{section}.\arabic{conjecture}}
\def\theproblem{\arabic{section}.\arabic{problem}}
\def\theremark{\arabic{section}.\arabic{remark}}

\makeatletter
\newtheorem*{rep@theorem}{\rep@title}
\newcommand{\newreptheorem}[2]{%
\newenvironment{rep#1}[1]{%
 \def\rep@title{#2 \ref{##1}}%
 \begin{rep@theorem}}%
 {\end{rep@theorem}}}
\makeatother

\newreptheorem{theorem}{Theorem}
\newreptheorem{proposition}{Proposition}
\newreptheorem{lemma}{Lemma}
\newcommand{\la}{\lambda}
\newcommand{\eps}{\varepsilon}
\newcommand{\ph}{\varphi}
\newcommand{\vr}{\varrho}
\newcommand{\al}{\alpha}
\newcommand{\bet}{\beta}
\newcommand{\gam}{\gamma}
\newcommand{\kap}{\kappa}
\newcommand{\s}{\sigma}
\newcommand{\sig}{\sigma}
\newcommand{\del}{\delta}
\newcommand{\om}{\omega}
\newcommand{\Gam}{\mathnormal{\Gamma}}
\newcommand{\Del}{\mathnormal{\Delta}}
\newcommand{\Th}{\mathnormal{\Theta}}
\newcommand{\La}{\mathnormal{\Lambda}}
\newcommand{\X}{\mathnormal{\Xi}}
\newcommand{\PI}{\mathnormal{\Pi}}
\newcommand{\Sig}{\mathnormal{\Sigma}}
\newcommand{\Ups}{\mathnormal{\Upsilon}}
\newcommand{\Ph}{\mathnormal{\Phi}}
\newcommand{\Ps}{\mathnormal{\Psi}}
\newcommand{\Om}{\mathnormal{\Omega}}

\newcommand{\C}{{\mathbb C}}
\newcommand{\D}{{\mathbb D}}
\newcommand{\M}{{\mathbb M}}
\newcommand{\N}{{\mathbb N}}
\newcommand{\Q}{{\mathbb Q}}
\newcommand{\R}{{\mathbb R}}
\newcommand{\U}{{\mathbb U}}
\newcommand{\T}{{\mathbb T}}
\newcommand{\Z}{{\mathbb Z}}

\newcommand{\EE}{{\mathbb E}}
\newcommand{\E}{{\mathbb E}}
\newcommand{\FF}{{\mathbb F}}
\newcommand{\I}{{\mathbb I}}
\newcommand{\PP}{{\mathbb P}}
\newcommand{\ONE}{\boldsymbol{1}}

\newcommand{\calA}{{\cal A}}
\newcommand{\calB}{{\cal B}}
\newcommand{\calC}{{\cal C}}
\newcommand{\calD}{{\cal D}}
\newcommand{\calE}{{\cal E}}
\newcommand{\calF}{{\cal F}}
\newcommand{\calG}{{\cal G}}
\newcommand{\calH}{{\cal H}}
\newcommand{\calI}{{\cal I}}
\newcommand{\calJ}{{\cal J}}
\newcommand{\calK}{{\cal K}}
\newcommand{\calL}{{\cal L}}
\newcommand{\calM}{{\cal M}}
\newcommand{\calN}{{\cal N}}
\newcommand{\calP}{{\cal P}}
\newcommand{\calQ}{{\cal Q}}
\newcommand{\calR}{{\cal R}}
\newcommand{\calS}{{\cal S}}
\newcommand{\calT}{{\cal T}}
\newcommand{\calU}{{\cal U}}
\newcommand{\calV}{{\cal V}}
\newcommand{\calX}{{\cal X}}
\newcommand{\calY}{{\cal Y}}

\newcommand{\bB}{{\mathbf B}}
\newcommand{\bI}{{\mathbf I}}
\newcommand{\bJ}{{\mathbf J}}
\newcommand{\bK}{{\mathbf K}}
\newcommand{\bY}{{\mathbf Y}}
\newcommand{\bX}{{\mathbf X}}
\newcommand{\bTh}{{\mathbf \Theta}}

\newcommand{\scrA}{\mathscr{A}}
\newcommand{\scrM}{\mathscr{M}}
\newcommand{\scrS}{\mathscr{S}}

\newcommand{\frA}{\mathfrak{A}}
\newcommand{\frM}{\mathfrak{M}}
\newcommand{\II}{\mathfrak{I}}

\newcommand{\frS}{\mathfrak{S}}

\renewcommand{\proof}{\noindent{\bf Proof.\ }}

\newcommand{\lan}{\langle}
\newcommand{\ran}{\rangle}
\newcommand{\uu}{\underline}
\newcommand{\oo}{\overline}
\newcommand{\skp}{\vspace{\baselineskip}}
\newcommand{\supp}{{\rm supp}}
\newcommand{\diag}{{\rm diag}}
\newcommand{\trace}{{\rm trace}}
\newcommand{\w}{\wedge}
\newcommand{\lt}{\left}
\newcommand{\rt}{\right}
\newcommand{\pl}{\partial}
\newcommand{\abs}[1]{\lvert#1\rvert}
\newcommand{\norm}[1]{\lVert#1\rVert}
\newcommand{\mean}[1]{\langle#1\rangle}
\newcommand{\To}{\Rightarrow}
\newcommand{\til}{\widetilde}
\newcommand{\wh}{\widehat}
\newcommand{\dist}{{\rm dist}}
\newcommand{\grad}{\nabla}
\newcommand{\iy}{\infty}

\newcommand{\be}{\begin{equation}}
\newcommand{\ee}{\end{equation}}


\newcommand{\para}[1]{\left( #1 \right)}        
\newcommand{\brac}[1]{\left\{ #1 \right\}}
\newcommand{\set}[1]{\left\{#1\right\}}         
\newcommand{\sbrac}[1]{\left[ #1 \right]}
\newcommand{\floor}[1]{\left\lfloor #1 \right\rfloor}
\newcommand{\ceil}[1]{\left\lceil #1 \right\rceil}

\newcommand{\tab}{\hspace*{0.3in}}
\newcommand{\Tab}{\hspace*{1.0in}}
\newcommand{\no}{\nonumber}
\newcommand{\noi}{\noindent}
\newcommand{\txt}{\textrm}
\newcommand{\ds}{\displaystyle}
\newcommand{\RR}{\mathbb{R}}
\newcommand{\vf}{\varphi}
\definecolor{co}{rgb}{0.8,0,0.8}
\definecolor{gr}{gray}{0.5}
\newcommand{\gr}{\color{gr}}
\newcommand{\vp}{\varepsilon}

\newcommand{\IA}{\text{\it IA}}

\newcommand{\eqtri}{\stackrel{\triangle}{=}}

\newcommand{\twopartdef}[4]
{
	\left\{
		\begin{array}{ll}
			#1 & \mbox{if } #2 \\
			#3 & \mbox{if } #4
		\end{array}
	\right.
}
\catcode`\@=11

\catcode`\@=12

\providecommand{\vs}{vs. }
\providecommand{\ie}{{i.e.,} }
\providecommand{\eg}{{e.g.,} }
\providecommand{\cf}{{cf.} }
\providecommand{\etc}{{etc.} }
\providecommand{\iid}{{i.i.d.\ }}

\newcommand{\rev}[1]{{\color{blue}#1}}

\newcommand\blfootnote[1]{%
  \begingroup
  \renewcommand\thefootnote{}\footnote{#1}%
  \addtocounter{footnote}{-1}%
  \endgroup
}

\title{Load Balancing Using Sparse Communication\blfootnote{
Manuscript version:  May 2024. Previously circulated under the title ``CARE: Resource Allocation Using Sparse Communication.'' Xu Kuang published under a different full name in earlier versions of this manuscript. Please use ``G.~Mendelson and X.~Kuang" when citing this paper.}} 

\author{Gal Mendelson \\ Technion \and  Xu Kuang \\ Stanford University}
\date{\vspace{-5ex}}
\maketitle

\begin{abstract}
Load balancing across parallel servers is an important class of congestion control problems that arises in service systems. An effective load balancer relies heavily on accurate, real-time congestion information to make routing decisions. However, obtaining such information can impose significant communication overheads, especially in demanding applications like those found in modern data centers.

We introduce a framework for communication-aware load balancing and design new load balancing algorithms that perform exceptionally well even in scenarios with sparse communication patterns. Central to our approach is state approximation, where the load balancer first estimates server states through a communication protocol. Subsequently, it utilizes these approximate states within a load balancing algorithm to determine routing decisions. 

We demonstrate that by using a novel communication protocol, one can achieve accurate  queue length approximation with sparse communication: for a maximal approximation error of $x$, the communication frequency only needs to be $O(1/x^2)$. 
We further show, via a diffusion analysis, that a constant maximal approximation error is sufficient for achieving asymptotically optimal performance. Taken together, these results therefore demonstrate that highly performant load balancing is possible with very little communication.  Through simulations, we observe that the proposed designs match or surpass the performance of state-of-the-art load balancing algorithms while drastically reducing communication rates by up to 90\%. 
\end{abstract}


\section{Introduction}

Load balancing across parallel servers is a fundamental resource allocation problem. It arises naturally when a decision maker has to dynamically route incoming demand  to be served at one of several, congestion-prone processing resources, and as such has found applications in health care, supply chains, and call centers. One prominent example of load balancing is in managing modern computer clusters and data centers, where a collection of computer servers process a stream of incoming jobs and performance critically depends on how well the servers are utilized. 

The system is illustrated in Figure \ref{fig:load_balancing_figure}. A load balancer allocates each incoming job to a resource or server, and the jobs then wait in a local queue until service is rendered by the said server. It is well understood that the key to a load balancer is the ability to steer incoming jobs away from already-congested servers. For this reason, it is paramount that the load balancer have access to each server's congestion states information, i.e., their queue length, in real time. Unfortunately, having perfect state information essentially amounts to the load balancer knowing exactly when and at which server every single job completes their service, which could incur a punishingly costly burden on the underlying communication infrastructure especially for large systems.

\begin{figure}[H]
\centering
\includegraphics[width=0.5\linewidth]{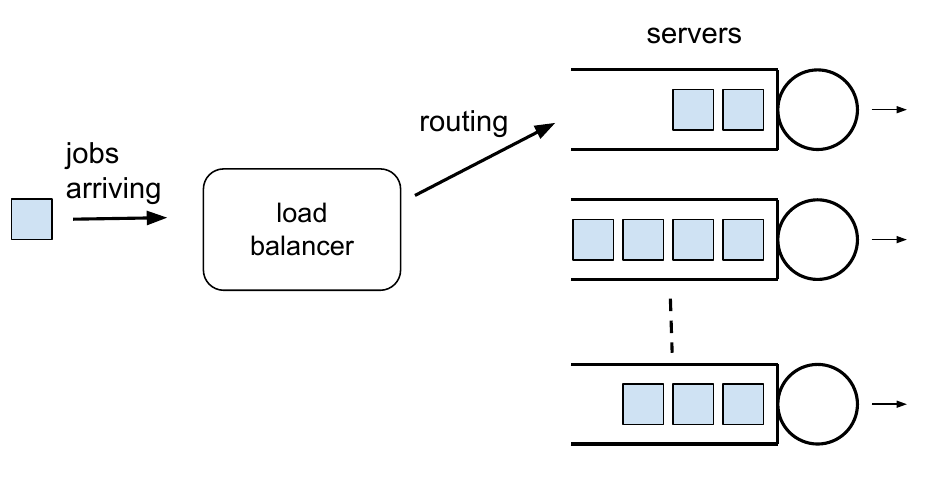}
\caption{A load balancer routes incoming jobs to servers working in parallel.}
\label{fig:load_balancing_figure}
\end{figure}

The goal of the present paper is to understand how to design performant  load balancing systems under a limited communication budget. We are not first to recognize the potential for optimizing information and communication in load balancing. A growing body of literature has identified performant load balancing algorithms that use less communication:  they query only a small number of servers each time a job needs to be routed \citep{vvedenskaya1996queueing, mitzenmacher2001power}, inform the load balancer only when a server becomes empty \citep{lu2011join}, or sample the state information of the servers on an intermittent basis \citep{tsitsiklis2013power}. 

However, while existing designs use less communication than would a naive approach, the resulting communication frequency remains substantial and not sparse. For instance, all three examples given above require communication frequencies that are comparable to exchanging one message for every arriving job, which adds up quickly in large systems (see Section \ref{sec: related work} for a detailed discussion). For context, it is known that a one-message-per-job communication rate is sufficient for maintaining perfect  state information \citep{lu2011join}.  Therefore, the ranges of communication frequency  achieved by existing algorithms do not mark a significant reduction from a full-information setup.  The question remains: 
\begin{center}
    Is performant load balancing possible under sparse communication, and if so, how to design it? 
\end{center}

\paragraph{State Approximation as Intermediary} In this paper, we formulate a model for load balancing with sparse communication. Using it, we will both provide a largely affirmative answer to the above question, and offer concrete design guidelines on how to achieve it.  Core to our approach is the use of state approximation: the load balancer first estimates the state of the queue lengths using a communication protocol,  and subsequently feeds the approximated state into a load balancing algorithm to generate a routing decision. 

Using state approximation explicitly in the design  has several benefits. State approximation is a natural abstraction for aggregate information learned over a sequence of sparse communications. It also allows for the freedom for the designer to choose the load balancing algorithm independently from the communication protocol, leading to greater flexibility in practical implementations. Indeed, one would expect a performant full-information algorithm would also perform reasonably well under  good, if not perfect, state approximation, as is evidenced by recent  empirical successes in algorithms that combine conventional load balancing with state approximation \citep{van2019hyper, vargaftik2020lsq}.  

Perhaps equally important, state approximation is a useful conceptual intermediary that breaks a complex design challenge into more manageable sub-tasks which we can now tackle individually: 
\begin{enumerate}
    \item How to use sparse communication to achieve reasonably accurate state approximations? 
    \item How to leverage reasonably accurate state approximations for good performance?
\end{enumerate}

\paragraph{Preview of Main Contributions} To answer these questions, we propose a model for designing and  analyzing communication-constrained load balancing systems, consisting of modules that represent {Communication, Approximation, Resource allocation, and Environment}, respectively. We will refer to this model by its acronym, {CARE} (Figure \ref{fig:CARE}). The communication and approximation components specify how the servers communicate with the load balancer, and how the load balancer converts the messages into an approximation of the server states. The resource allocation component captures how the load balancer  maps the state approximation into routing actions. Finally, the environment describes the underpinning stochastic mechanisms that drive the arrivals and services. The system operator is entitled to make various design choices in each of the first three components, with only the environment being taken as given. 

Using this model, we make the following contributions: 

\begin{enumerate}
    \item \emph{Approximation with sparse communication.} We propose and analyze three classes of communication-approximation schemes to  study the relationship between communication frequency and the resulting approximation error. A main takeaway is that we can indeed achieve high-quality state approximation under sparse communication, but it would require thoughtfully chosen communication and approximation schemes. We show that, by using a novel server-side-adaptive communication protocol, the load balancer can obtain a queue-length approximation with maximum error $y>0$ (for all time and across all servers), while using a communication frequency that is at most a factor $1/(y^2+y)$ of that required for maintaining perfect state information. To put this in perspective, it implies that we can reduce communication by $83\%$ compared to that of the full-information setting, while ensuring an approximation error that never exceeds $2$. If we are to tolerate an approximation error of $3$, then we can further reduce communication to that of $91\%$ of the full-information setting.

    Our designs build on a novel \emph{information asymmetry} in load balancing which we exploit: while the load balancer does not know the servers' queue lengths exactly, the servers do know exactly the load balancer's approximation about their own state. Using this insight, we design communication protocols where servers initiate communication when they perceive that the load balancer's approximation error concerning their own state crosses a certain threshold. Figure \ref{fig:first_paper_com_fig} offers an example illustrating the resulting drastic reduction in communication frequency.  
    
    \begin{figure}[H]
    \centering
    \includegraphics[width=0.6\linewidth]{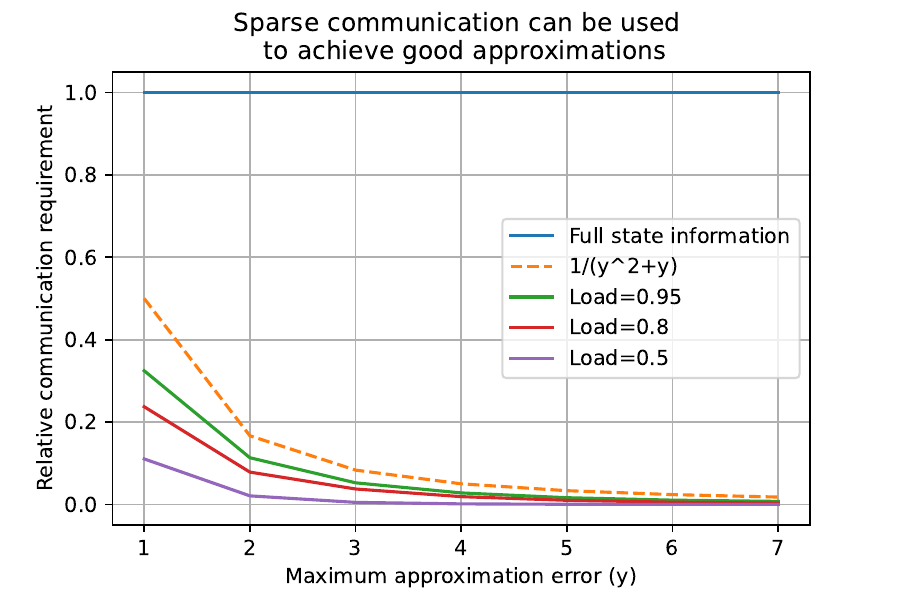}
    \caption{The relative communication requirements for our proposed algorithm in simulation experiments with 30 servers and various loads as a function of the maximum approximation error, $y$. The measured amount of communication is much less than what is required for full state information and the theoretical upper bound of $1/(y^2+y)$.}
    \label{fig:first_paper_com_fig}
    \end{figure}
     \item \emph{Performance characterization under state approximation.} We characterize when using approximate state information can lead to near-optimal performance. We focus on the natural load balancing algorithm where jobs join the server with the smallest approximate queue length. Because system dynamics under adaptive load balancing and approximation schemes are generally not amenable to exact analysis, we develop an asymptotic diffusion-scaled analysis to characterize performance, which also allows for time-varying arrival rates. Our main theorem in this part shows that the load balancer achieves asymptotically optimal performance whenever the approximation error is moderately bounded, which includes as an important special case when the approximation error is bounded by a constant $y$, as discussed previously. 
     
    The proof of the theorem is carried out in two steps. In the first step, we provide a sufficient condition on the approximation errors under which the queue lengths are completely balanced in diffusion scale. Following \cite{atar2019subdiffusive}, we refer to this phenomenon as Sub-Diffusivity of the Deviation Process (SDDP) of the queue lengths. In the second step, we prove that if SDDP of the queue lengths holds, then the workload in the system is minimized, when compared to the workload under any other algorithm (not necessarily approximation based), or even to the workload in a single server queue with a combined rate of the servers. 
    
    \item \emph{Design recommendations.} We conduct extensive simulations to draw performance comparisons between different communication-approximation schemes. The simulation results both corroborate our theoretical findings, and allow us to offer concrete design recommendations under various operating conditions. In particular, we find that using our proposed approximation algorithm and server-side-adaptive communication protocol results in significant performance gains with sparse communication. Figure \ref{fig:second_paper_com_fig} shows the performance of our proposed algorithm for different communication constraints, relative to that required for full state information, compared to the well known Join-the-Shortest-Queue, Power-of-Choice (SQ(2)) and Round-Robin load balancing algorithms. As we will soon discuss, the fundamental communication requirements of JSQ and SQ(2) are at least 1 message per job. Our proposed algorithm outperforms SQ(2) using around $0.1$ messages per job. Remarkably, it also offers a significant improvement over Round Robin while using less than $0.01$ messages per job.  We also demonstrate that the idle-signal based policies Join-the-Idle-Queue and Persistent-Idle cannot be considered as using sparse communication in our setting except for extremely high loads (e.g., above $0.99$), in which case the algorithms we propose perform better.

    \begin{figure}[H]
    \centering
    \includegraphics[width=0.6\linewidth]{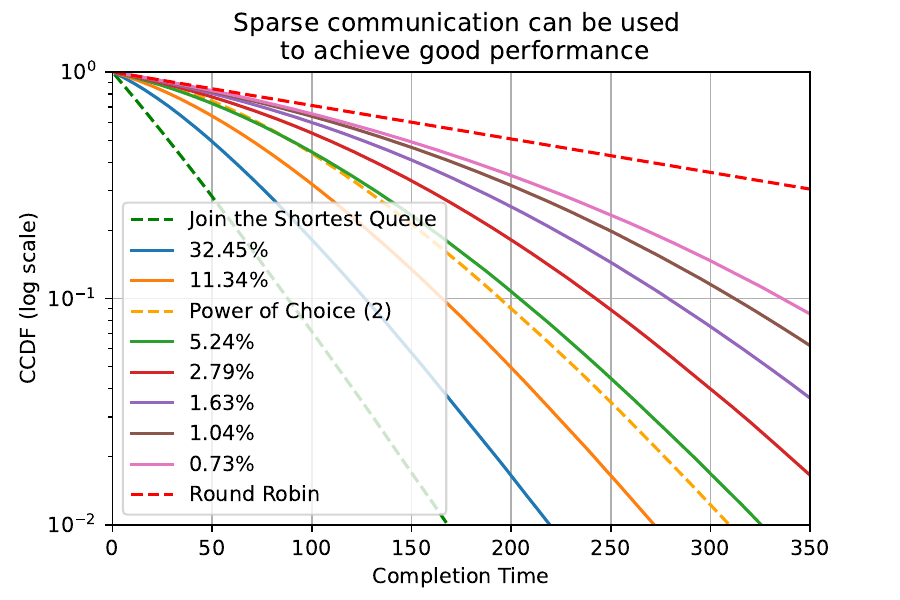}
    \caption{Measured CCDF of job completion times of our proposed algorithm, respecting different relative communication constraints, compared to that of JSQ, SQ(2) and Round Robin, in a simulation experiment with 30 servers and 95$\%$ load. Our algorithm outperforms SQ(2) while using around 10$\%$ of the communication, and offers substantial performance gains over Round Robin even when using less than $1\%$ of the communication.}
    \label{fig:second_paper_com_fig}
    \end{figure}
    
\end{enumerate}

\paragraph{Organization} The rest of the paper is organized as follows. In Section \ref{sec: CARE} we describe the CARE model and present our main results in more detail. In Section \ref{sec: related work} we discuss related work. Sections \ref{sec:The Approximation Component}, \ref{sec:The Communication Component} and \ref{sec: The Communication-Approximation Link} are devoted to the communication-approximation link, where we use the CARE model to study the problem of generating reasonable approximations using sparse communication. We motivate and present the different communication patterns and approximation algorithms and provide theoretical guarantees. In Section \ref{sec:approximation-performance link} we study the approximation-performance link by analyzing a diffusion scale model with a time varying arrival rate. In Section \ref{sec: connecting comm to per} we prove that our sparse communication algorithms lead to reasonable state approximations, which in turn lead to asymptotically optimal performance. Finally, in Section \ref{sec:simulations} we provide simulation results that support our theoretical findings and shed light on which algorithms perform best given a constrained communication budget.

\paragraph{Notation} We denote $a\wedge b=\min\{a,b\}$, $a\vee b=\max\{a,b\}$, $[K]=\{1,...,K\}$, $\mathbbm{N}=\{1,2,\ldots\}$, $\mathbbm{Z}_+=\mathbbm{N}\cup \{0\}$ and  $\mathbbm{R}_+=[0,\infty)$. Denote by $\mathcal{C}$ and $\mathcal{D}$ the classes of continuous, respectively, right continuous with left limits (RCLL), functions mapping $\mathbbm{R}_+$ to $\mathbbm{R}_+$. Denote $\mathcal{D}_0=\{f\in\mathcal{D}:f(0)=0\}$. For $f \in \mathcal{D}$, and a time interval $J=(t_1,t_2]$, denote $f(J]=f(t_1,t_2]=f(t_2)-f(t_1)$. For $x\in \mathbbm{R}^k$ ($k$ is a positive integer), $\norm{x}$ denotes the $\ell_1$ norm. For $f:\mathbbm{R}_+ \rightarrow \mathbbm{R}^k$, denote $\norm{f}_T=\sup_{t\in[0,T]}\norm{f(t)}$, and for $\theta>0$,
\begin{equation*}
    \omega_T(f,\theta)=\sup_{0\leq s<u\leq s+\theta\leq T}\norm{f(u)-f(s)}.
\end{equation*}

A sequence of processes $X_n$ with sample paths in
$\mathcal{D}$ is said to be {\it $\mathcal{C}$-tight} if it is tight and every subsequential
limit has, with probability 1, sample paths in $\mathcal{C}$. We use shorthand notation for integration as follows: $\II f(t)=\int_0^t f(u)du$.

We use the following asymptotic notation. For  deterministic positive sequences $f_n$ and $g_n$ we say that $f_n$ is \textit{of order} $o(g_n)$, and write $f_n\in o(g_n)$, if $\limsup_{n \rightarrow \infty}f_n/g_n=0$. We write $f_n\in O(g_n)$ if $\limsup_{n \rightarrow \infty}f_n/g_n<\infty$, $f_n\in \omega(g_n)$ if $\liminf_{n \rightarrow \infty}f_n/g_n=\infty$, $f_n\in \Omega(g_n)$ if $\liminf_{n \rightarrow \infty}f_n/g_n>0$ and $f_n\in \Theta(g_n)$ if $0<\liminf_{n \rightarrow \infty}f_n/g_n\leq \limsup_{n \rightarrow \infty}f_n/g_n<\infty$. For a non negative stochastic sequence $X_n$, we say that $X_n$ is \textit{of order} $o(g_n)$, and write $X_n \in o(g_n)$, if $X_n/g_n\rightarrow 0$ in probability, and $X_n \in \omega(g_n)$, if $X_n/g_n\rightarrow \infty$ in probability.

\section{The CARE Model and Summary of Contributions}\label{sec: CARE}
\subsection{The CARE Model}\label{subsec:the_care_model}

We present CARE, a model to study and design communication constrained resource allocation systems, as depicted in Figure \ref{fig:CARE}. 

\begin{figure}[H]
\centering
\includegraphics[width=0.9\linewidth]{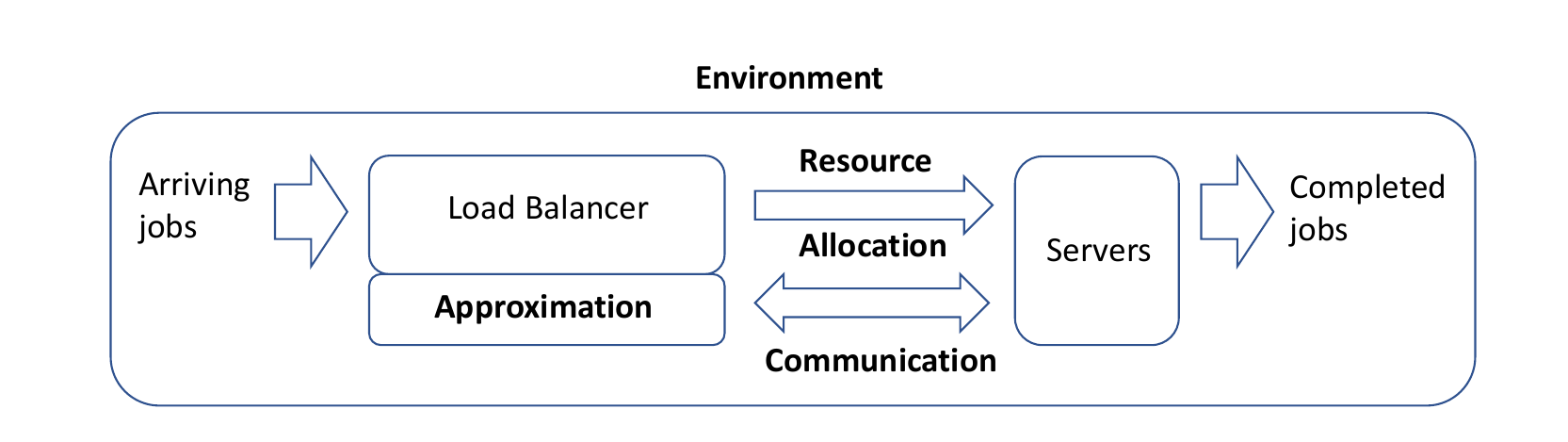}
\caption{The CARE model in the context of load balancing with parallel servers.}
\label{fig:CARE}
\end{figure}

\subsubsection{ The Environment Component}

The environment describes the underlying stochastic mechanism of the system. The environment, depicted in Figure \ref{fig:load_balancing_figure}, evolves in continuous time $t \geq 0$, and consists of a single load balancer and $K$ servers labeled $\{1,\ldots,K\}$. Jobs arrive to the load balancer according to some stochastic counting process, denoted by $A(t)$, which counts the number of job arrivals during $[0,t]$. When a job arrives the load balancer immediately sends it to one of the servers for processing.  Each server has a dedicated infinite capacity buffer in which a queue can form. Servers serve jobs in order of arrival (FIFO) and do not idle when there are jobs in their queue. 

Next, we describe how the queue length variables can be decomposed, a representation that will be used frequently in this paper. For $t_1 <t_2$, denote by $A_i(t_1,t_2]$ the number of jobs that arrived to server $i$ during $(t_1,t_2]$. Denote by $D_i(t_1,t_2]$ the number of jobs server $i$ has completed during $(t_1,t_2]$, and by $D(t)$ the total number of completed jobs in the system during $[0,t]$, namely:
\begin{equation*}
    D(t):=\sum_i D_i(0,t].
\end{equation*}
We assume that no jobs arrived or were completed at time $t=0$. Denote by $Q_i(t)$ the queue length at server $i$ at time $t$ which includes jobs in process if any. We assume the system starts operating at time $t=0$ with initial number of jobs in each queue given by $Q_i(0)$. Thus, the queue lengths satisfy the following basic flow equation:

\begin{equation}\label{eq:ql flow}
    Q_i(t)=Q_i(0)+A_i(0,t]-D_i(0,t].
\end{equation}

A few results will require the following additional assumptions on the model:

\begin{assumption}\label{ass2}
The arrival process $A(t)$ is a renewal process with constant rate $\lambda>0$.  
\end{assumption}

\begin{assumption}\label{ass1}
The service time requirements of jobs processed by each server $i \in \{1,\ldots,K\}$ are i.i.d.~random variables (RVs) with mean $1/\mu_i$ for some $\mu_i>0$.  
\end{assumption}

Assumption \ref{ass1} is incorporated into the model as follows. Fix $K$ infinite sequences $\{b_1^j\}_{j=1}^\infty,\ldots,\{b_K^j\}_{j=1}^\infty$  of i.i.d.~RVs with mean $1$. We assume that $b_i^j/\mu_i$ is the service time requirement of the $j$th job that server $i$ processes. Accordingly, define the potential service process as

\begin{equation}\label{eq:def of potential service}
    S_i(t):=\sup\{l\in \mathbbm{N}: \sum_{j=1}^l b_i^j \leq \mu_i t\}.
\end{equation}
Therefore, $S_i(t)$ is a renewal process that counts the number of jobs that server $i$ would have completed if it were busy $t$ units of time. Denote by $I_i(t)$ the cumulative amount of time server $i$ has been idle during $[0,t]$. We can now write the departure process as:

\begin{equation}\label{eq:departure potential service}
    D_i(0,t]=S_i(t-I_i(t)),
\end{equation}
and, since $S_i(t)$ is a non decreasing process, we have
\begin{equation}\label{eq:departure potential service bound}
    D_i(0,t]=S_i(t-I_i(t))\leq S_i(t).
\end{equation}

\subsubsection{ Communication Component}

At each point in time, each server is allowed to communicate with the load balancer by sending it a message. The content of a message sent by server $i$ at time $t$ is just its current queue length $Q_i(t)$. The \textit{communication pattern} dictates at which times the servers send messages to the load balancer. We assume that messages are sent and received instantly without delay. 

For each server $i$, denote the (possibly random) times after time zero at which server $i$ sends a message to the load balancer by $t_i^1,t_i^2,\ldots$, such that $t_i^m$ is the time at which the $m$th message is sent from server $i$ to the load balancer. Set $t_i^0=0$ for all $i$. Denote by $M_i(t)$ the number of messages server $i$ sent during $[0,t]$, which is given by:
\begin{equation*}
    M_i(t)=\sup\{m \in \mathbbm{N}: t_i^m \leq t\}.
\end{equation*}
Denote by $M(t)$ the total number of messages sent by the servers during $[0,t]$:
\begin{equation*}
    M(t):=\sum_i M_i(t).
\end{equation*}
It will be useful to write the flow equation \eqref{eq:ql flow} in a different way, using the messaging times $t_i^m$. For any messaging time $t_i^m<t$ we can apply \eqref{eq:ql flow} twice (for $t$ and then $t_m^i$) to obtain
\begin{align}\label{eq:actual ql}
    Q_i(t)&=Q_i(0)+A_i(0,t]-D_i(0,t]=Q_i(0)+A_i(0,t]-D_i(0,t]+Q_i(t_i^m)-Q_i(t_i^m)\cr
    &=Q_i(0)+A_i(0,t]-D_i(0,t]+Q_i(t_i^m)-\big(Q_i(0)+A_i(0,t_i^m]-D_i(0,t_i^m]\big)\cr
    &=Q_i(t_i^m)+A_i(t_i^m,t]-D_i(t_i^m,t],
\end{align}
where the last equality is due to the facts that $A_i(0,t]=A_i(0,t_i^m]+A_i(t_i^m,t]$ and 
$D_i(0,t]=D_i(0,t_i^m]+D_i(t_i^m,t]$.

\subsubsection{ Approximation Component}
The load balancer keeps an array of queue length state approximations, $\tilde{Q}(t)$, where  $\tilde{Q}_i(t)$ denotes the approximation for $Q_i(t)$ at time $t$. The \textit{approximation algorithm} is the rule by which the load balancer updates the approximated state using all the information available at its disposal, which comprises the messages it has received and its own routing decisions up until that time. We assume that all of the approximations initially coincide with the real state, namely $\tilde{Q}_i(0)=Q_i(0)$ for all $i$.

\subsubsection{ Resource Allocation Component}
The resource allocation component specifies how the load balancer assigns incoming jobs to servers. The \textit{load balancing algorithm} is the rule by which the load balancer makes these assignments using as input the approximated state, $\tilde{Q}(t)$. We will focus on the load balancing algorithm where a job is always sent to the server with the smallest approximate queue length, with ties broken randomly. We term this the Join the Shortest Approximated Queue (JSAQ) algorithm. The JSAQ algorithm is a natural candidate as the direct analog to the JSQ algorithm in the presence of state approximation. Indeed, as our main results demonstrate, JSAQ can achieve near optimal performance that rivals JSQ whenever the maximal approximation error is moderately desirable.  

\subsubsection{Metrics for Communication, Approximation and Performance}
\label{subsec:comm metric}

Next, we describe the metrics we use to characterize the amount of communication, approximation error and system performance. We will then use them to investigate the relationship between these three core elements of our model.

\paragraph{Communication Metric} We consider two metrics for measuring the amount of communication: (1) the number of messages per job departure, i.e., the relationship between $M(t)$ and $D(t)$, and (2) the long term average rate messages are sent, given by the quantity $M(t)/t$ for large $t$. Both metrics will be useful in comparing different approaches. In particular, it is known that one message per departure is sufficient to provide exact state information \citep{lu2011join}, so we are looking for combinations of communication patterns and approximation algorithms that use substantially less than that.

\paragraph{Approximation Metric} The metric we use to measure approximation error is the maximal absolute difference between any actual queue and its corresponding approximation at any point in time.  Define the approximation error of server $i$ at time $t$ by:
\begin{equation}\label{eq:ae def}
    AE_i(t):=|Q_i(t)-\tilde{Q}_i(t)|,
\end{equation}
and the maximal approximation error at time $t$ by:
\begin{equation*}
    AQ(t):=\max_i AE_i(t). 
\end{equation*}
For $y\in\mathbbm{N}$, we say that an approach has a \textit{maximal approximation error of $y$} if $\sup_{t \geq 0}AQ(t)\leq y$, i.e., the maximal absolute approximation error in the system is always less than or equal to $y$. Note that since we assumed that messages are sent instantly, the approximation errors at messaging times are zero.

\paragraph{Performance Metric} The primary metric for our theoretical analysis is the total workload in the system, defined as the sum of the remaining service duration requirements of the jobs in the system, including the ones in service. For simulations, we use a more direct, though less analytically tractable, metric of job completion times, i.e., the distribution of the time jobs spend in the system from the time they arrive up to the time a server finished their processing. The two metrics are related, in that the workload a job sees in front of it is the time it would take it to reach a server for processing. Since the waiting time of an incoming job equals exactly the workload in the corresponding server, minimizing the workload is closely tied to minimizing the time jobs spend in the system and is thus a meaningful measure of performance.  

\subsection{Summary of Main Results}

\paragraph{New Approximation Algorithms and Communication Patterns} We design a host of novel approximation algorithms and communication patterns, which, when used in unison, are shown to achieve impressive performance. On the approximation algorithm front, we propose a new approximation technique we call the \textit{queue length emulation approach}. Under this approach, after receiving a message from a server, the load balancer emulates the queue length evolution until the next message. This can be done naively, by only increasing the emulated queue after arrivals (which we term the \textit{basic} approach and is equivalent to the special case considered in \cite{van2019hyper} and \cite{vargaftik2020lsq}), or assigning the Mean job Service time Requirement to emulated jobs, an approach we term \textit{MSR}. We show that the latter is a key enabler in substantially reducing communication while achieving excellent performance. We also propose that in some situations, it is beneficial to truncate the departure process of the emulation after a predetermined number of emulated departures. We refer to this approach as \textit{MSR-x}, where the departure process is truncated after $x-1$ emulated departures.

We propose two new communication patterns, Departure Triggered (DT-$x$) and Error Triggered (ET-$x$), and compare their performance to that of an existing design in the literature, which we call Rate Triggered (RT-$r$) (used in \cite{van2019hyper} and \cite{vargaftik2020lsq}). Under RT-$r$, each server $i$ sends messages at a constant rate $r_i$. Under DT-$x$, servers send a message after $x$ of their job departures. Under ET-$x$, the most clever of the three, servers do the same exact emulation the load balancer does between messages and therefore can keep track of the approximation error. The servers send a message only when the approximation error reaches $x$. We show that combinations of DT-$x$ or ET-$x$ with the approximation algorithms described above are able to achieve surprisingly low maximal approximation errors and good performance using very little communication. Specifically, we find that using ET-$x$ with MSR dominates all other algorithms when the communication constraint is substantial.  

\paragraph{The Communication-Approximation Link} 
We provide a set of theoretical results which demonstrate that, using our protocols, we can achieve high-quality, if imperfect, state approximations even under highly sparse communication. The first result holds in striking generality, with essentially no assumptions on  the arrival process $A(t)$, job service time requirement distribution, server service discipline (not necessarily FIFO) or load balancing algorithm (not necessarily routing to the server with the smallest approximated queue).

\begin{theorem}\label{thm:summary 1}
For every $x \in \{2,3,4,\ldots\}$ there exists a combination of an approximation algorithm and communication pattern under which $AQ(t)\leq x-1$ and $M(t)\leq \frac{1}{x}D(t)$, for all $t\geq 0$. Specifically, this holds for DT-$x$ or ET-$x$ with the basic or MSR-x approximation algorithms. 
\end{theorem}
In words, the approaches we develop in this paper are able to achieve approximations that are never more than $x-1$ from the real queue lengths, using a fraction of $1/x$ of the communication needed for exact state information. E.g, a maximal approximation error of $2$ can be achieved with $33\%$ of the communication.

Our next two results establish that, under some conditions, we can further improve the guarantees in Theorem \ref{thm:summary 1} and achieve far sparser communication with the same maximal approximation error. Here we assume a FIFO service discipline and exponential service times. However, we still do not assume anything on the arrival process and the load balancing algorithm.

\begin{theorem}\label{thm:summary 2}
Consider a time interval $I_i$ between consecutive messages sent by server $i$ and denote its duration by $\tau_i$. Assume a FIFO service discipline, and that the service time requirements of jobs processed by server $i$ are i.i.d.~Exp($\mu_i$) random variables. Then, for every $x \in \{3,4,\ldots\}$, there exists a combination of an approximation algorithm and communication pattern under which $AQ(t)\leq x-1$ for all $t \in I_i$ and  $\mathbbm{E}[\tau_i]\geq (x/2-1)^2/\mu_i$. Specifically, this holds for ET-$x$ with MSR.
\end{theorem}

Theorem \ref{thm:summary 2} suggests that it is possible to achieve a messaging rate that decreases quadratically in the maximal approximation error. The third and final result concerning the communication-approximation link provides a stronger guarantee for small levels of approximation error, but at the price of assuming non-idleness.

\begin{theorem}\label{thm:summary 3}
In the setting and under the assumptions of Theorem \ref{thm:summary 2}, add the further assumption 
of an infinite backlog at the beginning of the interval $I_i$. Then, for every $x \in \{2,3,4,\ldots\}$, there exists a combination of an approximation algorithm and communication pattern under which $AQ(t)\leq x-1$ for all $t \in I_i$ and  $\mathbbm{E}[\tau_i]\geq x(x-1)/\mu_i$. Specifically, this holds for ET-$x$ with MSR.
\end{theorem}

Therefore, if servers are heavily loaded and service times are exponential, each server $i$ sends messages at an average rate of, approximately, at most $\frac{\mu_i}{x(x-1)}$. In the same setting, the average communication rate of server $i$ needed for exact state information is $\mu_i$ (corresponding to the long term departure rate). Thus, Theorem \ref{thm:summary 3} suggests that it is possible to achieve a maximal approximation error of $x-1$ with a fraction of $1/(x^2-x)$ of the communication which is significantly less than the $1/x$ fraction in Theorem \ref{thm:summary 1}. This is also supported by our simulation results in Section \ref{sec:simulations}, where the measured communication reduction is even more significant. We prove Theorems \ref{thm:summary 1}, \ref{thm:summary 2} and \ref{thm:summary 3} by construction, using the approaches we develop below. 


\paragraph{The Approximation-Performance Link} 
The next group of results characterizes the system performance for a given level of maximal approximation error. We consider a sequence of environments as defined in Section \ref{subsec:the_care_model} and indexed by $n$, the usual diffusion scaling parameter, operating during a fixed time interval $[0,T]$. All of the stochastic processes are defined analogously, only indexed by $n$. Jobs arrive according to a general modulated renewal process, such that the arrival rate $\lambda^n(t)$ may change over time. The potential service processes are general renewal processes such that the service rate of each server $i$ is $\mu_i^n$, i.e., the average duration of service for each job is $1/\mu_i^n$. In particular, we do not assume that the job service times are exponentially distributed. 

We assume that the arrival and service rates, $\lambda^n(t)$ and $\mu^n$, are of order $\Theta(n)$. Since we make no further assumptions on the rates, $\lambda^n(t)$ is allowed to transition between being below, equal to, or above $\sum_i \mu^n_i$, which translates into under-loaded, critically-loaded (also called heavy traffic) and over-loaded periods during $[0,T]$. Importantly, the rate at which events occur in the system is of order $\Theta(n)$, and, in general, queue lengths can be of order $\Theta(n)$. The latter will occur, for example, if the system is always over-loaded with $\lambda^n(t)-\sum_i \mu^n_i\geq cn$ for all $t>0$ and some constant $c>0$. 

First, we prove that under \textit{any} CARE system using JSAQ with a maximal approximation error that is of order $o(n^{1/2})$, the the system achieves asymptotic \emph{diffusion scale balance}, also known as  \textit{Sub-Diffusivity of the queue lengths Deviation Process} (SDDP). In particular, this refers to the phenomena that the maximal difference between any two actual queue lengths during $[0,T]$ is of order $o(n^{1/2})$.  In short, this result can be summarized as saying: 
\begin{equation}\label{intro eq:approx-ssc}
  \text{Maximal approximation error} \in o(n^{1/2})\Rightarrow \text{ SDDP}.
\end{equation}

Second, we prove that \textit{any} system design (including but not limited to CARE) which results in SDDP, is asymptotically optimal in the sense that it minimizes the \textit{nominal} workload in the system at any point in time up to a factor of order $o(n^{1/2})$. The nominal workload is defined as the time it would take a server with a processing rate of 1 to complete the work that is currently in the system. That is, 
\begin{equation}\label{intro eq:ssc-ao}
  \text{SDDP}\Rightarrow \text{Nominal workload asymptotic optimality} .
\end{equation}
We prove that this optimality guarantee holds not only compared to any other possible design, but also holds compared to a system where all of the servers are replaced with a single server with a combined processing rate of $\sum_i\mu^n_i$. 

In the special case of a homogeneous system where $\forall i, \mu^n_i=\mu^n$, we obtain that the \textit{actual} workload is minimized up to a factor of order $o(n^{1/2})$. The actual workload is tied to job completion times, since a server's actual workload equals the time an arriving job will wait to receive service. Thus, we show that

\begin{equation}\label{intro eq:ssc-ao}
  \text{SDDP+Homogeneous}\Rightarrow \text{Actual workload asymptotic optimality} .
\end{equation}


\paragraph{Connecting Communication to Performance} 

We prove that if the communication rate of the approaches we develop below is of order $\omega(n^{1/2})$, then the maximal approximation error is of order $o(n^{1/2})$:

\begin{equation}\label{intro eq:comm-approx}
    \text{Communication rate} \in \omega(n^{1/2}) \Rightarrow \text{Maximal approximation error} \in o(n^{1/2}).
\end{equation}

Thus, we can summarize the main theoretical contributions of this paper by combining \eqref{intro eq:comm-approx}, \eqref{intro eq:approx-ssc} and \eqref{intro eq:ssc-ao} and obtain:
\begin{equation*}
  \boxed{ \text{Comm.~rate } \omega(n^{1/2}) \Rightarrow \text{Approx.~quality } o(n^{1/2}) \Rightarrow \text{SDDP}\Rightarrow \text{Asymptotic optimality}} 
\end{equation*}
This series of results establishes that the approaches we consider achieve asymptotically optimal performance provided that their communication rate is of order $\omega(n^{1/2})$, which can be much slower than the rate at which events occur in the system. Specifically, it can be very sparse compared to the arrival or departure rates.


\paragraph{Design Recommendations} Finally, we combine the insights from our theoretical results and an extensive simulation study in order to recommend design choices using the CARE model. Specifically, we find that using the ET-$x$ communication pattern together with the MSR approximation approach leads to excellent performance when only sparse communication can be used. When substantial communication is allowed, DT-$x$ combined with MSR-$x$ performs very well. Finally, we find that RT-$r$ consistently delivers worse performance than the other approaches we propose. 

\section{Related Work}
\label{sec: related work}

\paragraph{Routing Policies} Load balancing across parallel queues has been an active area of research for more than a half-century. The Join-the-Shortest-Queue (JSQ) policy, whereby the load balancer routes each incoming job to the shortest of all queues, can be traced back to at least as early as the works of Haight \cite{haight1958two} and Kingman \cite{kingman1961two}. Since then, JSQ has been shown to perform optimally or near-optimally in a variety of systems and against various performance metrics \citep{ atar2019subdiffusive, ephremides1980simple, foschini1978basic, winston1977optimality,weber1978optimal}. 

Recognizing the difficulty in obtaining full state information to implement JSQ in modern large-scale systems, \cite{vvedenskaya1996queueing} and \cite{mitzenmacher2001power} propose a clever communication-light variant of JSQ, known as the Power-of-$d$-Choices or Shortest-Queue-$d$ (SQ($d$)). When a job arrives, the load balancer obtains the queue lengths at only $d$ out of a total of $K$ servers, and sends the job to the server with the shortest queue among them. As such, SQ($d$) can be thought of as a generalization of JSQ, and the two coincide when $d=K$. The idea of SQ($d$) is that by querying only $d$ servers, one would reduce the number of messages exchanged per arrival from order $K$ down to $d$. (As we will see below, this reduction is not entirely accurate.) Remarkably,  \cite{vvedenskaya1996queueing} and \cite{mitzenmacher2001power} show that even when $d=2$ the system's delay performance vastly outperforms random routing, and can in some cases rival that of JSQ. The high level idea of SQ($d$) is to use \textit{partial state} information to make routing decisions. However, unlike the state approximation considered here, this family of policies does not explicitly construct an estimate of the overall system state. 

The authors of \cite{shah2002use,mitzenmacher2002load} advocate for further improving SQ($d$) by \textit{memorizing} previously sampled queue lengths, leading to a variant of SQ($d$) called Power-of-Memory (PoM($m$,$n$)). The load balancer remembers the identity of the $m$ servers which had the shortest queues at the last sampling time, and when a job arrives samples them and an additional $n$ servers chosen at random. The idea is that between sampling times the state of the queues may not change significantly, so if a short queue is found, it is worthwhile to sample it again. A broader insight from this line of work is that when communication is limited, past observations are useful in constructing a more complete picture of the system state; the algorithms analyzed in this paper also make use of memory and past samples with a similar rationale. 

The authors of \cite{lu2011join} proposed a different partial state algorithm named Join-the-Idle-Queue (JIQ). The partial state available to the load balancer is which servers are idle (i.e., have no jobs to process). Then, an incoming job is assigned to an idle server. If there are no idle servers, the job is assigned at random. JIQ was extensively studied, e.g., \citep{mitzenmacher2016analyzing, foss2017large, stolyar2015pull}, and performs exceptionally well for low to moderate loads. The intuition behind it is that when the system is not very loaded, there are almost always idle servers. Thus, knowing which servers are idle is sufficient to successfully mimic the behaviour of JSQ. For moderate to high loads, servers are not idle often, and the performance of JIQ degrades to that of random assignment. 

The authors of \cite{atar2020persistent} propose a variant of JIQ called Persistent-Idle (PI), under which an incoming job is assigned to an idle server if there is one, and otherwise it is sent to the last server a job was sent to. Thus, PI uses the same partial information as JIQ, but also leverages the load balancer's memory to avoid the random assignment of JIQ in the case that there are no idle servers.

\paragraph{Communication Rates} Let us now look at the amount of communication that is required to implement the aforementioned policies. Until recently, the communication amount required to implement JSQ was considered to be $2K$ messages per arriving job: i.e., when a job arrives the load balancer sends a message to each server requesting its current queue length information, and each server responds with a message. Using the same type of implementation, SQ($d$) requires $2d$ messages per job arrival, which can be much less than that of JSQ, especially if the number of servers is large. Similarly, PoM requires $2(m+n)$ messages per job arrival.

However, \cite{lu2011join} points out that the fundamental communication requirement for implementing JSQ is actually much lower than previously thought. They make the insightful observation that because the load balancer knows to which servers jobs were sent, all it takes for the load balancer to reconstruct the exact queue lengths is for the servers to notify the load balancer every time a departure occurs. Using this implementation, we will need only 1 message per every departure, which is equivalent to 1 message per arrival assuming the system is stable. Compared to this more efficient JSQ implementation, the messaging rate of $2d$ and $2(m+n)$ per arrival of SQ($d$) and PoM($m$,$n$) are in fact much worse, and it is unclear whether there exists an implementation of SQ($d$) or PoM($m$,$n$) that can outperform the JSQ benchmark. 

As for algorithms such as JIQ and PI which use idleness information, the servers can send a message for every departure after which they have no more jobs to process. Then the load balancer knows they are idle and stores that information in memory. When the load balancer sends a job to an idle server, it knows it is no longer idle and updates its memory. Thus, for the load balancer to always know which servers are idle, the amount of required communication is less than or equal to 1 message per departure, since not every departure leaves the server idle. 

However, our simulation results indicate that up to moderate to high loads, e.g., $70\%$, the amount of required communication of JIQ and PI is roughly the same as that of JSQ, namely, 1 message per departure. This is because when the system is not heavily loaded, there are almost always idle servers and most departures leave the servers idle. As the load increases, the amount of required communication decreases, but it requires extremely high loads, e.g., $98\%$, for it to drop below half of what JSQ requires, i.e., below 1 message per 2 departures. Thus, JIQ and PI cannot be regarded as using sparse communication in our setting, apart from situations where the system is almost critically loaded.

Table \ref{tab:comm} summarizes the various communications rates. 

\begin{figure}[H]
\centering
\begin{tabular}{ |l | l | }
 \hline
 Algorithm & Communication \\ 
 \hline
 Join-Shortest-Queue (JSQ) & $1$ (D), $K$ (A)  \\  
 Shortest-Queue-$d$ (SQ($d$)) &  $1$ (D), $2d$ (A) \\
 Power-of-Memory-$(m,n)$ (PoM($m,n$))  & $1$ (D), $2(m+n)$ (A) \\
 Joint-the-Idle-Queue (JIQ)  & <1 \\
 Persistent-Idle (PI) & <1 \\
 \hline
 ET-x with MSR, or DT-x, general job size (Thm.~\ref{thm:summary 1}) & $1/(x-1)$\\
 ET-x with MSR, exponential job size (Thm.~\ref{thm:summary 2}) & $1/{(x/2-1)^2}$\\
 ET-x with MSR, exponential job size, heavy load (Thm.~\ref{thm:summary 3}) & $1/(x^2-x)$\\
 \hline 
 \end{tabular}
 \caption{Comparison of the communications rates of different load balancing architectures in a system with $K$ servers. Communication rates are measured in number of messages per arrival. (D) denotes the rate under the implementation where servers notify the load balancer upon job departures, while (A) denotes the implementation where the load balancer queries servers upon each arrival. The parameter $x$ controls the quality of approximation, and can be roughly interpreted as the maximum error tolerance in the queue length approximations. }
 \label{tab:comm}
\end{figure}

\paragraph{State Approximation} The idea of using state approximation in routing has appeared in \cite{van2019hyper} and \cite{vargaftik2020lsq}. 
The authors of \cite{van2019hyper} use approximate states in the following manner. Servers send periodic updates to the load balancer. The time between updates can be constant or exponentially distributed. When the load balancer receives a message from a server it updates it to the value it received. When the load balancer sends a job to a server, it increases the corresponding approximation by 1. The load balancing algorithm is chosen to be Join the Shortest Queue (JSQ) which we term Join-the-Shortest-Approximated-Queue (JSAQ). Under JSAQ, the load balancer assigns an incoming job the the server with the smallest approximated queue. 

In terms of performance,  \cite{van2019hyper} assume a Poisson arrival process with constant rate $\lambda$ and exponentially distributed job service times. They use a fluid model to analyze the system in a large system limit, i.e., as the number of servers tends to infinity. Their main result is that if the communication rate is close to the arrival rate $\lambda$ then the mean waiting time of jobs in steady state tends to zero. The intuition behind this result is that as the number of servers grows, if the communication rate is close to the rate of arrivals then there will probably be idle servers when a job arrives which the load balancer is aware of, in which case the job is sent to an idle server and its waiting time is zero. As such, \cite{van2019hyper} focus more on proposing using approximations for routing in a specific manner, and less on exploring the different trade-offs in the design of load balancing systems that use state approximation. In particular, the main performance guarantee holds when the communication rate is close to the arrival rate in the system. But, as was observed by \cite{lu2011join} and discussed above, this is precisely the communication rate which is sufficient to obtain {exact} queue length information, and hence is sufficient for implementing the optimal JSQ. Thus, a "restricted" communication rate must be at most some small fraction of the arrival rate to justify using approximations.

\cite{vargaftik2020lsq} suggest using JSAQ for a different, discrete time setting, where jobs arrive and are processed in batches and there are multiple load balancers that must make routing decisions at the same time. In this setting even JSQ is not optimal, since it suffers from a {herding effect} (also referred to as {incast}). This corresponds to the possibility of several load balancers sending their batches of jobs at the same time to the same server with the shortest queue, possibly overloading it. Using approximate states is suggested as a way to mitigate the herding effect while simultaneously using a reduced amount of communication. The main result of the paper is that if the average difference between the approximations and the actual queue lengths is bounded, then the system is strongly stable. While stability is an important first order concern, it does not provide insight on performance since algorithms that perform poorly can also be stable (e.g., assigning jobs to servers randomly). 

Finally, \cite{zhou2021asymptotically} also consider the discrete time multiple dispatcher setting and study a large class of load balancing algorithms which contains JSAQ. They provide sufficient conditions for the system to be stable, as well as heavy traffic mean delay optimal in steady state. Here, as well as in \cite{vargaftik2020lsq}, the communication patterns and the approximation methods are suggested in an ad hoc manner, so as to fulfill the sufficient conditions discussed in these papers.   

In summary, while prior work has showcased encouraging results using approximate states, we have yet to arrive at a satisfactory understanding as to the true potential and limitation of load balancing using approximate state information. 

\paragraph{Formal Models of Communication and Memory} Our work is related in spirit to a growing body of literature that proposes formal models for quantifying the value of information and communication in resource allocation systems (cf.~\cite{walton2021learning} for a survey.)  \cite{xu2020information} proposes the Stochastic Processing under Imperfect Information (SPII) model to study the impact of communication and memory on the capacity region of a stochastic processing network. For various levels of communication constraints and memory capacities, they derive bounds on the achievable capacity region and provide the corresponding optimal scheduling policies and communication protocols. While \cite{xu2020information}, like the present paper, also studies the joint problem of communication and resource allocation, it focuses solely on system stability and does not address delay performance.  

More relevant for the context of load balancing is \cite{gamarnik2018delay}, where the authors propose a model of communication to obtain insights on the trade-off between memory, communication and performance for a variant of JIQ. Two important features of the policy they consider is that the load balancer's memory and the rate at which servers may inform the load balancer they are idle may be constrained. They study the steady state waiting time of jobs in this system using a fluid limit approach where the number of servers grows to infinity, the arrival process is assumed Poisson and the service times are exponentially distributed. In many modern communication systems, memory is not an issue if the task only requires storing identities of servers at the load balancer. Thus, for our purposes, the most relevant result of \cite{gamarnik2018delay} is that if the arrival rate per server is $0<\lambda<1$ and the average communication rate per server is $\alpha\lambda$, where $0<\alpha<1$, then the asymptotic delay is upper bounded by $1/\alpha\lambda-1$. This is better than Random Routing, where the delay per server is $\lambda/(1-\lambda)$, whenever $\lambda>1/(1+\alpha)$. Thus, if the communication is substantially constrained, i.e., $\alpha$ is small, a significant improvement over Random Routing is proven only for high loads. In addition, it is well known that the Round Robin policy, where jobs are assigned to servers in a cyclic fashion, and therefore uses no communication, performs much better than Random Routing (see \cite{liu1994optimality} and references therein). The simulations we conducted concur with this, and that is why we compare our proposed algorithms to Round Robin. The question of whether variants of JIQ using sparse communication can be used to achieve good performance is therefore left open.  

\paragraph{Methodology on Sub-diffusive Analysis} Related to our model and performance analysis, the authors of \cite{atar2019subdiffusive} study the diffusion scale performance of several load balancing algorithms, including, for example, SQ($d$), in a time varying setting. The analysis in \cite{atar2019subdiffusive} relies on carefully tailored event constructions whose probability can be estimated using the fact that the load balancer uses actual state information to make routing decisions. In this paper, we show how to extend the analysis to work with approximation based load balancing. The authors also prove a weaker version of \eqref{intro eq:ssc-ao} where the asymptotic optimality guarantee is with respect to all load balancing algorithms for $K$ servers. We strengthen their result by using new proof techniques to prove that (1) the performance of a combined rate single server queue serves as a lower bound for the performance of all load balancing algorithms for $K$ servers, and (2) that optimality holds with respect to the performance of this single server queue.


\section{The Approximation Component}\label{sec:The Approximation Component}
\subsection{Overview}
In this section we present the approximation component of the CARE model. We start by observing that because a message always contains the exact state of the server at that time, the error of any approximation algorithm is determined by its ability to estimate the departure process in between two adjacent messages. We consider a first-order approximation algorithm we term \textit{the basic approximation algorithm} in which the load balancer approximates the queue lengths based only on message updates and arrivals, without estimating any departures.

We then introduce a new, general, approximation approach we call \textit{the general queue length emulation approach}. Under this scheme, the load balancer emulates single server queues between messaging times while assigning estimated service requirement times to initial and arriving jobs. We then consider the natural emulation approach where the load balancer assigns each job in the emulation its a priori known Mean Service Requirement and term the approach MSR. Finally, for reasons we detail in later sections, we consider the MSR-x approach where the departure process in the emulation is truncated after $x-1$  departures.

\subsection{Approximation and Departure Estimation}
In order for the load balancer to approximate the queue lengths well, it must use the information available to it in a smart way. There are two pieces of information that are available to the load balancer at any point in time. The first is the state of each queue at the time of the last received message from the corresponding server, i.e., $Q_i(t_i^m)$ if $t \in [t_i^m,t_i^{m+1})$. The second is the arrival process to each server (since the arrivals are jobs that the load balancer sent to the servers, it can keep track of them). This is $A_i(t_i^m,t]$ for $t \in [t_i^m,t_i^{m+1})$. Considering the evolution of the actual queue length in \eqref{eq:actual ql} and the information available to the load balancer, it is natural to write $\tilde{Q}_i(t)$ under \textit{any} approximation algorithm in the following way:

\begin{equation}\label{eq:ap ql flow}
    \tilde{Q}_i(t)=Q_i(t_i^m)+A_i(t_i^m,t]-\tilde{D}_i(t_i^m,t], \quad t \in [t_i^m,t_i^{m+1}),
\end{equation}
where $\tilde{D}_i(t_i^m,t]$ encodes the approximation algorithm.

Recall from \eqref{eq:ae def} the definition of the approximation error of the queue length at server $i$ at time $t$,
\begin{equation*}
    AE_i(t):=| Q_i(t)- \tilde{Q}_i(t) |.
\end{equation*}
Then, by \eqref{eq:actual ql}, \eqref{eq:ae def} and \eqref{eq:ap ql flow} we have

\begin{equation}\label{eq:ae departure diff}
    AE_i(t)=| D_i(t_i^m,t]- \tilde{D}_i(t_i^m,t] |, \quad t\in [t_i^m,t_i^{m+1}).
\end{equation}
The implication of \eqref{eq:ae departure diff} is the following:
\begin{key observation}\label{ko1}
The evolution of the approximation error $AE_i(t)$ as a function of time for \textbf{any} approximation algorithm is determined solely by how the load balancer estimates the departure process between messages. 
\end{key observation}

\subsection{The Basic Approach}
A natural first algorithm to consider is the simple rule where the load balancer's estimate of the number of departures is always zero. We term this algorithm as the \textit{basic} approximation algorithm and define it next.

\begin{definition}[The basic approximation algorithm]\label{def:basic}
For every server $i$, the approximation of the queue in server $i$ at time $t\in [t_i^m,t_i^{m+1})$ is given by:
\begin{equation}\label{eq: def of B}
    \tilde{Q}^B_i(t)=Q_i(t_i^m)+A_i(t_i^m,t] \quad t\in [t_i^m,t_i^{m+1}),
\end{equation}
where we use the superscript $B$ to distinguish the resulting approximation of the basic algorithm from other approaches.

\end{definition}

The approximation error of the basic algorithm is simply the number of departures that have occurred since the last messaging time:

\begin{proposition}
We have:
\begin{equation}\label{eq:error of basic}
    AE^B_i(t)=D_i(t_i^m,t], \quad t\in [t_i^m,t_i^{m+1}).
\end{equation}

\end{proposition}

\proof
Fix $t\in [t_i^m,t_i^{m+1})$. Using \eqref{eq:actual ql}, \eqref{eq:ae def} and \eqref{eq: def of B} we have:
\begin{align*}
    AE^B_i(t)=|\tilde{Q}^B_i(t)-Q_i(t)|&=|Q_i(t_i^m)+A_i(t_i^m,t]-\big( Q_i(t_i^m)+A_i(t_i^m,t]-D_i(t_i^m,t]\big)|\cr
    &=|D_i(t_i^m,t]|=D_i(t_i^m,t],
\end{align*}
where the last equality is due to the fact that the departure process is non negative. \qed \\

The basic algorithm can be implemented in a  very simple way: The load balancer keeps an array in which it stores its approximations. Whenever it receives a message from a server, it overwrites the corresponding array entry with the new state information. Whenever the load balancer sends a job to server $i$, it increases the corresponding array entry by 1.

If the load balancer does not have any knowledge on the service time distribution of jobs, then the basic approximation algorithm is a conservative, robust choice. The approximation cannot under-estimate the actual queue length and the approximation error grows exactly as the number of departures since the last message. Using this approach, the load balancer cannot send a job to a long queue based on the misconception that it is shorter than it really is, which is desirable. However, the load balancer might not be aware of actual short queues if it substantially over-estimates them. 

\subsection{The General Queue Length Emulation Approach}
Let us now consider more sophisticated load balancer approximation algorithms. Revisiting \eqref{eq:ae departure diff} and Observation \ref{ko1}, it is clear that if the load balancer is able to estimate actual departure times effectively, the approximation error can be substantially reduced. To this end, we propose the \textit{general queue length emulation approach}. Under this approach, the load balancer emulates for each server a single server queue between messaging times. For $t\in [t_i^m,t_i^{m+1})$, the queue length approximation $\tilde{Q}_i(t)$ is just the value of the emulation that started at time $t_i^m$. A detailed definition is as follows.

\begin{definition}[The general queue length emulation approach]\label{def:qls}
For every server $i$ and every time $t_i^m=0,1,2,\ldots$, the load balancer starts an emulation of a single server FIFO queue as follows:
 \begin{itemize}
     \item The initial condition is $Q_i(t_i^m)$.
     \item The load balancer assigns estimated service time requirements to each of the initial $Q_i(t_i^m)$ jobs.
     \item Every arrival to the actual queue is also an arrival to the emulated queue. For each such arrival, the load balancer assigns it an estimated service time requirement.
 \end{itemize}
 For $t\in [t_i^m,t_i^{m+1})$, the value of $\tilde{Q}_i(t)$ is given by the current queue length value of the emulation that started at time $t_i^m$.
\end{definition}


\begin{remark}
Note that the basic approach is just a special case of the general emulation approach, where the load balancer estimates all service times as $\infty$.
\end{remark}

\begin{remark}
If the load balancer knows the actual service requirement of every job, it can keep track of the exact queue lengths without any communication from the servers. Definition \ref{def:qls} provides the mechanism by which this can be achieved.
\end{remark}

We started with trying to effectively estimate the departure process in order to reduce the approximation error in \eqref{eq:ae departure diff}. But, looking at Definition \ref{def:qls}, we observe the following:

\begin{key observation}\label{ko2}
Under the general queue length emulation approach, the general problem of estimating the actual departure process reduces to the problem of estimating job service time requirements.
\end{key observation}

Observation \ref{ko2} motivates us to consider the following specific queue length emulation approach we term the \textit{Mean Service Requirement} (MSR) approximation algorithm. Assuming that the load balancer knows the average service time requirements of jobs, it simply assigns it to each job during its emulations. 

\begin{definition}[The MSR approximation algorithm]\label{def:MSR}
Suppose that the average service time requirement of each arriving job to server $i$ is equal to $1/\mu_i$ for some $\mu_i>0$. The load balancer uses the queue length emulation approach as described in Definition \ref{def:qls}, where for each job it must assign a service time requirement to, it assigns $1/\mu_i$, depending on the server $i$ it was routed to. 
\end{definition}

We use a superscript $M$ to distinguish the MSR algorithm form the basic algorithm. The queue length approximations under MSR are:

\begin{equation*}
    \tilde{Q}^M_i(t)=Q_i(t_i^m)+A_i(t_i^m,t]-\tilde{D}^M_i(t_i^m,t], \quad t \in [t_i^m,t_i^{m+1}),
\end{equation*}
where $\tilde{D}^M_i(t_i^m,t]$ is the number of departures from the emulated queue where each job in the emulation has a deterministic service time requirement equal to $1/\mu_i$. By \eqref{eq:ae departure diff} the approximation error of MSR is:

\begin{equation}\label{eq:ae MSR}
    AE^M_i(t)=| D_i(t_i^m,t]- \tilde{D}^M_i(t_i^m,t] |, \quad t\in [t_i^m,t_i^{m+1}).
\end{equation}
   
While the MSR approach provides a natural way of estimating the departure process, the uncertainty in the error \eqref{eq:ae MSR} grows with time and the number of jobs in the emulation procedure. As we discuss in Section \ref{subsec: DT-x}, it is useful to truncate the departure process in the emulation before the number of departures reaches a predetermined number $x \in \mathbbm{N}$. 

\begin{definition}[The MSR-x approximation algorithm]\label{def:MSR-x}
Let $x \in \mathbbm{N}$ and suppose that the average service time requirement of each arriving job to server $i$ is equal to $1/\mu_i$ for some $\mu_i>0$. The load balancer uses the queue length emulation approach as described in Definition \ref{def:qls}, where for each server $i$, for the first $x-1$ jobs it must assign a service time requirement to, it assigns $1/\mu_i$, and for every job after that it assigns the value $\infty$. 
\end{definition}

We use an additional superscript $x$ to distinguish the MSR-x algorithm form the MSR and the basic algorithms. The queue length approximations under MSR-x are:
\begin{align*}
    \tilde{Q}^{M,x}_i(t)&=Q_i(t_i^m)+A_i(t_i^m,t]-\tilde{D}^{M,x}_i(t_i^m,t]\cr
    &=Q_i(t_i^m)+A_i(t_i^m,t]-\tilde{D}^M_i(t_i^m,t]\wedge (x-1), \quad t \in [t_i^m,t_i^{m+1}),
\end{align*}
and the approximation error is:
\begin{align*}
    AE^{M,x}_i(t)&=| D_i(t_i^m,t]- \tilde{D}^{M,x}_i(t_i^m,t] |\cr
    &=| D_i(t_i^m,t]- \tilde{D}^M_i(t_i^m,t]\wedge (x-1) |, \quad t\in [t_i^m,t_i^{m+1}).
\end{align*}

We proceed with the next component of the CARE model, namely the communication component.
\section{The Communication Component}\label{sec:The Communication Component}

\subsection{Overview}

Our observations above shed light on the design of the communication pattern. A natural approach is to limit the time between messages so as to somewhat control how (stochastically) large the approximation error gets. We term this \textit{the Rate-Triggered} (RT) communication pattern. A second approach is to control the error by limiting the number of actual departures between messages, an approach we call \textit{the Departure-Triggered} (DT) communication pattern. A third, more sophisticated method, is to allow the servers to conduct the exact same queue length emulation the load balancer conducts, and have servers send a message whenever the approximation error reaches some threshold. We term this \textit{the Error-Triggered} (ET) communication pattern. In this section we present the three approaches. We discuss the communication frequency and resulting maximal approximation error of the three approaches in more detail in the next sections. 

\subsection{Rate-Triggered Communication}
Under this approach each server sends messages to the load balancer at some constant rate $r_i>0$, regardless of its state. Generally speaking, the larger $r_i$ is, the smaller the approximation error is.

\begin{definition}[The RT-$r$ communication pattern]\label{def:rt}
Fix $r=(r_1,\ldots,r_K)$. Each server $i$ sends a message to the load balancer every $1/r_i$ units of time.
\end{definition}

RT has two main advantages. First, it is very simple and easy to implement. It does not require the servers to keep track of their state or emulate queue lengths. Second, given a constraint on the communication frequency, e.g., every server can send messages to the load balancer at a rate of at most $r_{max}$, it is straightforward to use RT with rate $r_{max}$, thus guaranteeing the constraint is respected. 

The price however, is that the communication is not used to prevent the approximation error from substantially increasing. It is used to limit the time during which it can increase. We shall see in later sections that the adaptive approaches DT and ET offer lower maximal approximation error guarantees and better performance.   

\subsection{Departure-Triggered Communication}

Another reasonable communication pattern design is that each server sends a message every $x$ of its departures. We call this the Departure Triggered communication pattern and denote it by DT-$x$.  

\begin{definition}[The DT-$x$ communication pattern]\label{def:dt}
Each server sends a message to the load balancer immediately after every $x$ of its departures.
\end{definition}

DT-$x$ achieves two goals simultaneously. First, the value of $x$ determines the number of messages sent by the servers and allows us to compare it against the communication required for exact state information. For example, if $x=3$, then the number of messages is $1/3$ messages per departure in the system, which is $1/3$  of the communication required for exact state information. Second, it allows us to guarantee that the approximation error is never more than $x$ (when combined with the basic or MSR-x approximation algorithms). We discuss these points in further detail in the next two sections. Finally, DT-$x$ has the advantage that it does not depend on the approximation algorithm. Thus, if one wishes to change the way the load balancer approximates queues, no changes are needed on the servers' side.

\subsection{Error-Triggered Communication}

We conclude this section by presenting a more sophisticated communication pattern which targets the approximation error itself. Generally, if we want to make sure that the approximation error never reaches some value $x$, then it is necessary and sufficient for the servers to send a message only when the approximation error reaches $x$. This can only be achieved if the servers can keep track of the approximation error as a function of time.

The approximation error depends on two things. The first is the actual queue length evolution $Q_i(t)$, which is known to the server. The second, is the load balancer approximation $\tilde{Q}_i(t)$. Thus, if a server has access to $\tilde{Q}_i(t)$, it can calculate the approximation error at each point in time and send a message whenever it reaches $x$.

We therefore propose the approach where the servers mimic the approximation algorithm to calculate $\tilde{Q}_i(t)$. For example, the servers can do the same single server queue emulation the load balancer does between messages. 

\begin{definition}[The ET-$x$ communication pattern]\label{def:et}
Every server sends a message to the load balancer whenever the approximation error reaches $x$.
\end{definition}

We will show that ET-$x$ has the same maximal approximation error guarantee as DT-$x$, but with, possibly, substantially less communication. The price however is that ET-$x$ requires a more complicated calculation by the servers.

Now that we have described our proposed approximation and communication algorithms, we are ready to discuss and provide results on the communication-approximation link.

\section{The Communication-Approximation Link}\label{sec: The Communication-Approximation Link}

In this section we discuss the communication requirement and the maximal approximation error of the different approaches we considered above. We first characterize the communication requirement for obtaining exact state information, which we will use as the baseline for comparison. Then, we consider the RT-$r$, DT-$x$ and ET-$x$ communication patterns separately, after which we reiterate and prove Theorems \ref{thm:summary 1}, \ref{thm:summary 2} and \ref{thm:summary 3}.
\subsection{Exact State Information Baseline}
As was described in Section \ref{subsec:comm metric}, the communication is measured by either (1) the number of messages per job departure, i.e., the relationship between $M(t)$ and $D(t)$, and/or (2) the long term average rate messages are sent, given by the quantity $M(t)/t$ for large $t$. The following provides the baseline for the communication requirements of the different algorithms we propose.

\begin{proposition}
Assume that the service time requirements of jobs are not deterministic and are not known to the load balancer. Then the load balancer can know the exact state information \textit{at any point in time} if and only if each server sends a message after every job departure. Namely, 
\begin{equation*}
    \text{Baseline}: \quad M(t)=D(t)\quad \forall t \ge 0
\end{equation*}
is necessary and sufficient to obtain exact state information.
\end{proposition}

\begin{proof}
Suppose each server sends a message after each of its departures. This allows the load balancer to track the departure process at each server. Since the load balancer also knows the arrival process to each server, by \eqref{eq:ql flow}, it can construct the exact queue length at any point in time. For the opposite direction, assume by contradiction that there was a job departure at server $i$ which did not trigger a message. If just before this departure the load balancer did not know the exact queue length at server $i$ we are done. Otherwise, since service times are assumed stochastic and unknown to the load balancer, it cannot predict this departure with probability 1, meaning there is a positive probability that the load balancer does not know the exact state after the departure.
\end{proof}

For long-run average results, we need more assumptions.
\begin{proposition}\label{prop:long exact arrival}
Assuming renewal arrivals (Assumption \ref{ass2}),

\begin{align*}
    \text{Baseline}: \quad \frac{M(t)}{t}=\frac{\sum_{i=1}^K M_i(t)}{t}\leq \frac{A(t)}{t} \rightarrow \lambda \quad \text{almost surely as } t \rightarrow \infty.
\end{align*}
\end{proposition}

\begin{proposition}\label{prop:long exact}
Assuming i.i.d. services at rate $\mu_i$ at server i (Assumption \ref{ass1}),

\begin{align*}
    \text{Baseline}: \quad \frac{M_i(t)}{t}\leq \frac{S_i(t)}{t} \rightarrow \mu_i \quad \text{almost surely as } t \rightarrow \infty.
\end{align*}
\end{proposition}

Propositions \ref{prop:long exact arrival} and \ref{prop:long exact} follow immediately from the facts that $M(t)=D(t)\leq A(t)\wedge \sum_iS_i(t)$, and that $A(t)$ and $S_i(t)$ are renewal processes.
We conclude, under Assumptions \ref{ass2} and  \ref{ass1}, that the long term average rate of messages required for exact state information is roughly given by $\lambda \wedge \sum_i \mu_i$. 
\subsection{RT-r}

The communication requirement of RT-$r$ is deterministic, such that each server sends messages at a constant rate $r_i$. The maximal approximation error however, is harder to characterize and depends on the approximation algorithm. We proceed with an informal discussion.

If the basic approximation algorithm is used (see Definition \ref{def:basic}), then the approximation error is given in \eqref{eq:error of basic}

\begin{equation*}
    \text{RT-$r$}: \quad AE^B_i(t)=D_i(t_i^m,t], \quad t\in [t_i^m,t_i^m+1/r_i),
\end{equation*}
namely, the number of departures since the last message. This quantity depends on how many jobs were initially in the queue, how many jobs arrived to it and what their service time requirements were. If we assume Assumption \ref{ass1}, such that jobs have i.i.d.~service times requirements with mean $1/\mu_i$, then the number of departures can roughly grow at most linearly at a rate of $\mu_i$.

Since the departure process is non decreasing, the maximum error is attained right before the next messaging time, namely after $1/r_i$ units of time have past since the last message. So, given a communication rate $r_i$, the approximation error should be, on average, always below $\mu_i/r_i$. 

Now, if the MSR (Definition \ref{def:MSR}) or MSR-x (Definition \ref{def:MSR-x}) are used as the approximation algorithm the analysis is even more difficult. The MSR approach offers a reasonable way to estimate the departure process. The approximation error is given in \eqref{eq:ae MSR}:

\begin{equation}\label{eq:ae MSR RT}
    \text{RT-$r$}: \quad AE^M_i(t)=| D_i(t_i^m,t]- \tilde{D}^M_i(t_i^m,t] |, \quad t\in [t_i^m,t_i^{m}+1/r_i).
\end{equation}

Intuitively, the difference in \eqref{eq:ae MSR RT} is given by, roughly, the actual departure process minus its mean, and behaves similarly to a martingale. Therefore, drawing intuition from the $L^p$ maximal inequality (Theorem 4.4.4. in \cite{durrett2019probability}), as time $t$ progresses, this difference grows more slowly, as $2\sigma_i \sqrt{t}$, where $\sigma_i$ is the standard deviation related to the variability in service time requirements. Thus, under RT-$r$, the approximation error during a time interval of duration $1/r_i$ should not substantially exceed $2\sigma_i/\sqrt{r_i}$. This type of relation is somewhat expected in the sense that higher variability of service times should force the rate to be larger to keep a certain level of maximal approximation error.

The key takeaway is that while the above discussion sheds some light on what approximation errors to expect given a deterministic communication rate, one cannot guarantee a certain maximal approximation error when using RT-$r$. 


\subsection{DT-x}\label{subsec: DT-x}
Under DT-$x$, each server sends a message after every $x$ of its departures. Thus, the number of messages each server sends during a time interval depends on several components of the system, namely the arrival process, the load balancing and approximation algorithms and the service time requirements of jobs. For example, as opposed to RT-$r$, if the system is empty and there are no arrivals, the servers will not send any messages.

It is therefore natural to measure the communication frequency of DT-$x$ as the number of messages per departure in the system. This is also how we characterized the communication requirement for exact state information: one message per departure. We can conclude that the communication requirement of DT-$x$ is $1/x$ messages per departure:

\begin{proposition}\label{prop:dt comm guarnatee}
Under DT-$x$ with any approximation algorithm, we have:
\begin{equation*}
    \text{DT-$x$}: \quad M_i(t)=\floor{D_i(t)/x}\leq D_i(t)/x, \quad \forall t\geq 0.
\end{equation*}
\end{proposition}

\proof 
By definition of DT-$x$. \qed 

\vspace{10pt}

We can go further if we assume Assumption  \ref{ass2} and/or \ref{ass1}. We now give an upper bound on the long term rate servers communicates with the load balancer.

\begin{proposition}\label{prop:long dt arrival}
Under DT-$x$, with any approximation algorithm, if Assumptions \ref{ass2}, respectively, \ref{ass1}, hold, then: 

\begin{align}\label{eq:dt comm ass2}
    \text{DT-$x$}: \quad \frac{M(t)}{t}=\frac{\sum_{i=1}^K M_i(t)}{t}\leq \frac{1}{x}\frac{A(t)}{t} \rightarrow \frac{\lambda}{x} \quad \text{almost surely as } t \rightarrow \infty.
\end{align}
and, respectively,
\begin{align}\label{eq:dt comm ass1}
    \text{DT-$x$}: \quad \frac{M_i(t)}{t}\leq \frac{1}{x}\frac{S_i(t)}{t} \rightarrow \frac{\mu_i}{x} \quad \text{almost surely as } t \rightarrow \infty.
\end{align}
\end{proposition}

\proof
First, we have:

\begin{align*}
    M_i(t)=\floor{D_i(t)/x}\leq D_i(t)/x \leq S_i(t)/x,
\end{align*}

where the last inequality is due to \eqref{eq:departure potential service bound}. This proves the first inequality in \eqref{eq:dt comm ass1}. Second, the process $S_i(t)$ (defined in \eqref{eq:def of potential service}) is a renewal process. Thus, it satisfies $S_i(t)/t \rightarrow \mu$ almost surely as $t \rightarrow \infty$ (Theorem 1 in section 10.2 of \cite{grimmett2020probability}) which completes the proof of \eqref{eq:dt comm ass1}. The proof of \eqref{eq:dt comm ass2} follows a similar argument after noting that $\sum M_i(t)\leq \sum D_i(t)/x\leq A(t)/x$.
\qed 

\vspace{10pt}

We now discuss the maximal approximation error under DT-$x$. If the basic approximation algorithm is used, the approximation error is just the number of departures since the last messaging time, and it must be less or equal to $x-1$ by definition:
\begin{equation}\label{eq:error of basic dt}
    \text{DT-$x$}: \quad AE^B_i(t)=D_i(t_i^m,t]\leq x-1, \quad t\in [t_i^m,t_i^{m+1}).
\end{equation}

The following example shows that if the MSR algorithm is used, the approximation error under DT-$x$ cannot be deterministically bounded:
\begin{example}\label{ex1}
Suppose that the queue length at server $i$ at time $t_i^m$ is very large. Suppose further that the service time requirement of the job that is in service immediately after time $t_i^m$ is also very large. Then, during a long period of time there will be no departures from the actual queue but there will be many emulated departures, resulting in a large approximation error.  
\end{example}

The situation described in Example \ref{ex1} should be very rare. However, from an engineering perspective, it is desirable to be able to guarantee that these types of occurrences cannot happen. This is why we proposed the MSR-x algorithm, which prevents the emulated departures from exceeding $x-1$, guaranteeing the approximation error is always less or equal to $x-1$:

\begin{proposition}\label{prop:dtx msrx ae}
\begin{equation}\label{eq:error of msrx dt}
    \text{DT-$x$}: \quad AE^{M,x}_i(t)\leq x-1, \quad \forall t>0.
\end{equation}
\end{proposition}

\proof
Fix $t \in [t_i^m,t_i^{m+1})$. By \eqref{eq:ae MSR} we have

\begin{align*}
    AE^{M,x}_i(t)=| D_i(t_i^m,t]- \tilde{D}^{M,x}_i(t_i^m,t] |\leq D_i(t_i^m,t]\vee \tilde{D}^{M,x}_i(t_i^m,t]\leq(x-1)\vee(x-1)=x-1
\end{align*}
\qed

\subsection{ET-x}

Under ET-$x$, each server sends a message whenever the approximation error reaches $x$. The error is zero at a messaging time and can only increase or decrease by 1 due to actual or emulated departures. Thus, the approximation error of ET-$x$ is never more than $x-1$ under any of the approximation algorithms we consider.

If the basic approximation algorithm is used, then the load balancer does not emulate any departures. Thus the approximation error reaches $x$ exactly after $x$ departures since the last messaging time. This means that ET-$x$ coincides with DT-$x$ in this case, and thus has the same communication requirement of $1/x$ messages per departure (and under Assumptions \ref{ass2} and \ref{ass1} satisfies \eqref{eq:dt comm ass2} and \eqref{eq:dt comm ass1}, respectively).

If the MSR-x algorithm is used, then the number of emulated departures cannot exceed $x-1$. Therefore, the approximation error cannot reach $x$ because of an emulated departure, only because of an actual departure. But, it will take at least $x$ departures for the approximation error to reach $x$. We conclude that the communication frequency of ET-$x$ with MSR-x is at most (and could be substantially less) 1 message every $x$ departures. Under Assumptions \ref{ass2} and \ref{ass1}, \eqref{eq:dt comm ass2} and \eqref{eq:dt comm ass1} hold for ET-$x$ with MSR-x as well.

As an intermediary summary, we have the following results for ET-$x$:

\begin{proposition}\label{prop:et gurantees}
Under ET-$x$, we have that $AQ=x-1$ for any approximation algorithm. If the basic or MSR-$x$ approximation algorithms are used, then:
\begin{equation*}
        \text{ET-$x$}: \quad M_i(t)=\floor{D_i(t)/x}\leq D_i(t)/x, \quad \forall t\geq 0.
\end{equation*}
If, in addition, Assumptions \ref{ass2}, respectively , \ref{ass1}, hold, then:      
\begin{align*}
    \text{ET-$x$}: \quad \frac{M(t)}{t}=\frac{\sum_{i=1}^K M_i(t)}{t}\leq \frac{1}{x}\frac{A(t)}{t} \rightarrow \frac{\lambda}{x} \quad \text{almost surely as } t \rightarrow \infty.
\end{align*}
\begin{align*}
    \text{ET-$x$}: \quad \frac{M_i(t)}{t}\leq \frac{1}{x}\frac{S_i(t)}{t} \rightarrow \frac{\mu_i}{x} \quad \text{almost surely as } t \rightarrow \infty.
\end{align*}

\end{proposition}

Finally, we turn to discuss ET-$x$ with MSR. The communication requirement in this case is the most difficult to characterize. In the unconstrained case of MSR, it could be that messages are sent because of emulated departures. In fact, Example \ref{ex1} shows that there is no deterministic limit on the number of messages that are sent per departure. Specifically, a job with a large service time requirement can trigger many messages due to emulated departures. But, as our analysis and simulations will show, this is very rarely the case. 

We begin with proving a rough upper bound on the long-run average communication rate of ET-$x$ with MSR, and defer the more refined analysis, showing a quadratic decrease of the communication rate as $x$ increases, to the next section.

\begin{proposition}\label{prop:et msr gurantees}
Under ET-$x$ with MSR, we have: 
\begin{equation}\label{eq:et msr gurantee}
        \text{ET-$x$}: \quad M_i(t)=\floor{D_i(t)/x}+\floor{\mu_i t/x}\leq D_i(t)/x+\mu_i t/x, \quad \forall t\geq 0.
\end{equation}
If, in addition, Assumptions \ref{ass2}, respectively , \ref{ass1}, hold, then:      
\begin{align}\label{eq:et msr rough 1}
    \text{ET-$x$}: \quad \frac{M(t)}{t}=\frac{\sum_{i=1}^K M_i(t)}{t}\leq \frac{1}{x}\frac{A(t)}{t} +\frac{\sum_i \mu_i}{x}\rightarrow \frac{\lambda+\sum_i \mu_i}{x} \quad \text{almost surely as } t \rightarrow \infty.
\end{align}
\begin{align}\label{eq:et msr rough 2}
    \text{ET-$x$}: \quad \frac{M_i(t)}{t}\leq \frac{1}{x}\frac{S_i(t)}{t}+\frac{ \mu_i}{x} \rightarrow \frac{2\mu_i}{x} \quad \text{almost surely as } t \rightarrow \infty.
\end{align}

\end{proposition}

These results establish that the long term communication rate of ET-$x$ with MSR is roughly, at most $(\lambda \wedge \sum_i \mu_i+\sum_i \mu_i)/x$. This upper bound on the total rate of communication is larger by $\frac{\sum \mu_i}{x}$ than the rate we proved for the other approaches using DT-$x$ or ET-$x$ with the basic or MSR-$x$ approximation algorithms. However, it can still be substantially lower than the baseline rate for large $x$.  

\vspace{10pt}

\begin{proof}
Equations \eqref{eq:et msr rough 1} and \eqref{eq:et msr rough 2} will follow after we prove \eqref{eq:et msr gurantee} and use the facts that $D_i(t)\leq S_i(t)$ and that $D(t)\leq A(t)$.  For a server $i$, consider the consecutive, disjoint, time intervals $[0,t^1_i],(t^1_i,t^2_i],\ldots,(t^{M_i(t)-1}_i,t^{M_i(t)}_i],(t^{M_i(t)}_i,t]$ whose union is $[0,t]$. By the definition of ET-$x$ with MSR, at the end of each of these intervals, except the last, there must have been an actual departure, or an emulated one. In the first case, the number of actual departures since the previous message must be at least $x$. Similarly, in the latter case, the number of emulated departures since the last message must be at least $x$ as well. Thus, the number of intervals that ended with an actual departure cannot exceed $\floor{D_i(t)/x}$, and the number of intervals that ended with an emulated departure cannot exceed $\floor{\mu_i t/x}$. Since the total number of intervals equals $M_i(t)$, \eqref{eq:et msr gurantee} follows.
\qed 
\end{proof}

\subsection{Main Communication-Approximation Link Results}

We are now ready to prove Theorems \ref{thm:summary 1}, \ref{thm:summary 2} and \ref{thm:summary 3}. We reiterate them for readability.

\begin{reptheorem}{thm:summary 1}
For every $x \in \{2,3,4,\ldots\}$ there exists a combination of an approximation algorithm and communication pattern under which $AQ(t)\leq x-1$ and $M(t)\leq \frac{1}{x}D(t)$, for all $t\geq 0$. Specifically, this holds for DT-$x$ or ET-$x$ with the basic or MSR-x approximation algorithms. 
\end{reptheorem}

\begin{proof}
Fix $x \in \{2,3,4,\ldots\}$. The result for ET-$x$ follows from Proposition \ref{prop:et gurantees}. Next, consider DT-$x$. If the basic or MSR-$x$ approximation algorithms are used, then by \eqref{eq:error of basic dt} and \eqref{eq:error of msrx dt}, the maximal approximation error is equal to $x-1$ as desired. By Proposition \ref{prop:dt comm guarnatee}, $M_i(t)\leq D_i(t)/x$, which concludes the proof.
\qed 
\end{proof}

\vspace{10pt}

To prove Theorems \ref{thm:summary 2} and \ref{thm:summary 3}, we will need the following results:

\begin{theorem}\label{thm: poisson}
Fix $y \in \mathbbm{N}$. Let $N(t)$ be a Poisson process with rate $\mu$. Define:

\begin{equation}\label{eq:def of tau exit time}
    \tau:=\inf\{t\ge 0 : |N(t)-\floor{\mu t}|=y\},
\end{equation}
and
\begin{equation}\label{eq:def of sigma exit time}
    \sigma:=\inf\{t\ge 0 : |N(t)-\mu t|\geq y\}.
\end{equation}

Then:
\begin{equation}\label{eq:result on avg tau exit time}
    \mathbbm{E}[\tau] \geq (y^2-y)/\mu,
\end{equation}
and
\begin{equation}\label{eq:result on avg sigma exit time}
    \mathbbm{E}[\sigma] \geq y^2/\mu.
\end{equation}
\end{theorem}

\begin{proof}
Denote by $\mathcal{F}^N$ the natural filtration of $N(t)$. Clearly, $\tau$ and $\sigma$ are stopping times with respect to $\mathcal{F}^N$. Define:
\begin{align*}
    &Y(t)=N(t)-\mu t \cr
    &Z(t)=Y^2(t)-\mu t.
\end{align*}
We will need the following Lemma:
\begin{lemma}\label{lem:ui}
We have:
\begin{enumerate}
    \item $Y(t)$ and $Z(t)$ are martingales  with respect to $\mathcal{F}^N$.
    \item $Y(t\wedge \tau)$, $Y(t\wedge \sigma)$, $Z(t\wedge \tau)$  and $Z(t\wedge \sigma)$ are uniformly integrable martingales with respect to $\mathcal{F}^N$.
\end{enumerate}
\end{lemma}

\begin{proof}
The fact that $Y(t)$ and $Z(t)$ are martingales with respect to $\mathcal{F}^N$ is well known (e.g., the example in Chapter 3, page 50 of \cite{le2016brownian}). By Corollary 3.24 in page 61 of \cite{le2016brownian}, $Y(t\wedge \tau)$, $Y(t\wedge \sigma)$, $Z(t\wedge \sigma)$ and $Z(t\wedge \tau)$ are also martingales with respect to $\mathcal{F}^N$. We are left with verifying uniform integrability. A process $X(t)$ is uniformly integrable if:
\begin{equation*}
    \lim_{a \rightarrow \infty}\sup_{t \geq 0} \mathbbm{E}[|X(t)|\mathbbm{1}_{\{|X(t)| \geq a\}}]=0.
\end{equation*}
If $X(t)$ is bounded, namely, there exists $c>0$ such that $X(t)\le c$ for all $t \ge 0$, then it is uniformly integrable:
\begin{equation*}
    0 \leq \mathbbm{E}[|X(t)|\mathbbm{1}_{\{|X(t)| \geq a\}}]\leq c\mathbbm{P}(|X(t)| \geq a)\leq c \mathbbm{P}(c \geq a)=0, \quad \forall a>c.
\end{equation*}
We thus start with proving that $Y(t\wedge \tau)$ and $Y(t\wedge \sigma)$ are bounded. We have:

\begin{equation}\label{eq:UI 1}
    |Y(t)|=|N(t)-\mu t|=|N(t)-\mu t+\floor{\mu t}-\floor{\mu t}|\leq |N(t)-\floor{\mu t}|+1.
\end{equation}
By the definition of $\tau$ in \eqref{eq:def of tau exit time}, we have:
\begin{equation}\label{eq:UI 2}
    |N(t\wedge \tau)-\floor{\mu (t\wedge \tau)}|\leq y, \quad \forall t\geq 0.
\end{equation}
Combining \eqref{eq:UI 1} and \eqref{eq:UI 2} yields:
\begin{equation*}
    |Y(t\wedge \tau)|\leq y+1.
\end{equation*}
By the definition of $\sigma$ in \eqref{eq:def of sigma exit time}, we have:
\begin{equation}\label{eq:UI 3}
    |N(t\wedge \sigma)-\mu (t\wedge \sigma)|\leq y, \quad \forall 0\leq t < \sigma.
\end{equation}
For $t\geq \sigma$, notice that at time $\sigma$, either $\mu ( \sigma)-N(\sigma)=y$, or $N(\sigma)-\mu ( \sigma)\geq y$. In the latter case, we must have $N(\sigma)-\mu ( \sigma^-)< y$ and $N(\sigma)-N(\sigma^-)=1$. Thus $|N(\sigma)-\mu ( \sigma)|\leq  y+1$. Combined with \eqref{eq:UI 3}, we obtain:

\begin{equation}\label{eq:UI 4}
    |N(t\wedge \sigma)-\mu (t\wedge \sigma)|\leq y+1, \quad \forall t \geq 0.
\end{equation}
Combining \eqref{eq:UI 1} and \eqref{eq:UI 4} yields:
\begin{equation}\label{eq:UI 5}
    |Y(t\wedge \sigma)|\leq y+1.
\end{equation}
Therefore, $Y(t\wedge \tau)$ and $Y(t\wedge \sigma)$ are uniformly integrable.

We proceed with proving that $Z(t\wedge \sigma)$ is uniformly integrable. The proof that $Z(t\wedge \tau)$ is uniformly integrable follows the same arguments and is hence omitted. Since $Z(t\wedge \sigma)$ is a martingale, we have:

\begin{equation*}
    0=\mathbbm{E}[Z(0\wedge \sigma)]=\mathbbm{E}[Z(t\wedge \sigma)]=\mathbbm{E}[Y^2(t\wedge \sigma)-\mu (t\wedge \sigma)]\leq (y+1)^2-\mu\mathbbm{E}[t\wedge \sigma],
\end{equation*}
where the last inequality is due to \eqref{eq:UI 5}, implying that
\begin{equation*}
    \mathbbm{E}[t\wedge \sigma]\leq (y+1)^2/\mu.
\end{equation*}
Since $0\leq t\wedge \sigma \uparrow \sigma$ almost surely as $t \rightarrow \infty$, by the Monotone Convergence Theorem (Theorem 1.6.6. in \cite{durrett2019probability}), we have $\mathbbm{E}[t\wedge \sigma] \uparrow \mathbbm{E}[ \sigma]$ as $t \rightarrow \infty$, yielding
\begin{equation}\label{eq:ui sigma avg is bounded}
    \mathbbm{E}[ \sigma]\leq (y+1)^2/\mu.
\end{equation}
Now,
\begin{equation*}
    |Z(t\wedge \sigma)|=|Y^2(t\wedge \sigma)-\mu (t\wedge \sigma)|\leq (y+1)^2+\mu \sigma,
\end{equation*}
and therefore:
\begin{align*}
    0\leq \mathbbm{E}[|Z(t\wedge \sigma)|\mathbbm{1}_{\{|Z(t\wedge \sigma)| \geq a\}}]\leq \mathbbm{E}[((y+1)^2+\mu \sigma)\mathbbm{1}_{\{(y+1)^2+\mu \sigma \geq a\}}]\xrightarrow{a \rightarrow \infty} 0,
\end{align*}

where the convergence is due to \eqref{eq:ui sigma avg is bounded}, which implies that $\mathbbm{E}[(y+1)^2+\mu \sigma]< \infty$ (see Exercise 5.6.5. and the discussion on page 396 in \cite{grimmett2020probability}).
\qed 
\end{proof}

\vspace{10pt}

We now continue with proving \eqref{eq:result on avg tau exit time} and \eqref{eq:result on avg sigma exit time} of Theorem \ref{thm: poisson}. By Lemma \ref{lem:ui}, $Z(t\wedge \sigma)$ is a uniformly integrable martingale. Thus we can use the Optional Stopping Theorem (Theorem 3.22 in \cite{le2016brownian}) and obtain:
\begin{equation*}
    0=\mathbbm{E}[Z(0\wedge \sigma)]=\mathbbm{E}[Z( \sigma)]=\mathbbm{E}[Y^2(\sigma)-\mu \sigma]\geq y^2-\mu\mathbbm{E}[\sigma],
\end{equation*}
which proves \eqref{eq:result on avg sigma exit time}. To prove \eqref{eq:result on avg tau exit time}, by similar arguments that led to \eqref{eq:ui sigma avg is bounded}, we can obtain that $\mathbbm{E}[\tau]<\infty$ and thus $\tau< \infty$ almost surely. Define the complementing events:
\begin{align*}
    &A:=\{N(\tau)-\floor{\mu \tau}=y\}\cr
    &A^c:=\{N(\tau)-\floor{\mu \tau}=-y\}.
\end{align*}
Under the event $A$, by the definition of $\tau$ in \eqref{eq:def of tau exit time}, the process $N(t)$ must jump up by 1 at time $\tau$ and $\floor{\mu t}$ does not. Under the event $A^c$, the process $\floor{\mu t}$ must jump up by 1 at time $\tau$ and $N(t)$ does not. In particular, under $A^c$, we have $\floor{\mu \tau}=\mu \tau$. By using the Optional Stopping Theorem for $Y(t\wedge \tau)$ we obtain:
\begin{align*}
    0&=\mathbbm{E}[Y(0\wedge \tau)]=\mathbbm{E}[Y( \tau)]=\mathbbm{E}[N(\tau)-\mu \tau]\cr
    &= \mathbbm{E}[N(\tau)-\mu \tau | A]\mathbbm{P}(A)+\mathbbm{E}[N(\tau)-\mu \tau | A^c]\mathbbm{P}(A^c)\cr
    &=(y-\mathbbm{E}[\mu \tau -\floor{\mu \tau }|A])\mathbbm{P}(A)-y(1-\mathbbm{P}(A)).
\end{align*}
After rearranging and using the fact that  $0 \leq \mathbbm{E}[\mu \tau -\floor{\mu \tau }|A] \leq 1$, we obtain:
\begin{equation}\label{eq:bound on alpha}
    \frac{1}{2} \leq \mathbbm{P}(A) \leq \frac{y}{2y-1}.
\end{equation}
Now,  using the Optional Stopping Theorem for $Z(t\wedge \tau)$ we obtain:
\begin{align*}
    0=\mathbbm{E}[Z(0\wedge \tau)]=\mathbbm{E}[Z( \tau)]=\mathbbm{E}[(N(\tau)-\mu \tau)^2-\mu \tau],
\end{align*}
Implying that 
\begin{align*}
    \mu \mathbbm{E}[\tau]=\mathbbm{E}[(N(\tau)-\mu \tau)^2].
\end{align*}
Finally, we have:
\begin{align*}
    \mu \mathbbm{E}[\tau]=\mathbbm{E}[(N(\tau)-\mu \tau)^2]&=\mathbbm{E}[(N(\tau)-\mu \tau)^2 | A]\mathbbm{P}(A)+\mathbbm{E}[(N(\tau)-\mu \tau)^2 | A^c]\mathbbm{P}(A^c)\cr
    &\geq (y-1)^2\mathbbm{P}(A)+y^2(1-\mathbbm{P}(A))\cr
    &=y^2-(2y-1)\mathbbm{P}(A)\geq y^2-y,
\end{align*}

where the last inequality is due to \eqref{eq:bound on alpha}, which concludes the proof of \eqref{eq:result on avg tau exit time}.
\qed \\
\end{proof}

We are now ready to prove Theorems \ref{thm:summary 2} and \ref{thm:summary 3}.

\begin{reptheorem}{thm:summary 2}
Consider a time interval $I_i$ between consecutive messages sent by server $i$ and denote its duration by $\tau_i$. Assume a FIFO service discipline, and that the service time requirements of jobs processed by server $i$ are i.i.d.~Exp($\mu_i$) random variables. Then, for every $x \in \{3,4,\ldots\}$, there exists a combination of an approximation algorithm and communication pattern under which $AQ(t)\leq x-1$ for all $t \in I_i$ and  $\mathbbm{E}[\tau_i]\geq (x/2-1)^2/\mu_i$. Specifically, this holds for ET-$x$ with MSR.
\end{reptheorem}

\begin{proof}
Suppose that at time $t_0$ a message was sent by server $i$ to the load balancer, updating it on its queue length $Q_i(t_0)$. Denote by $\mathcal{F}_t$ the filtration which contains all of the information on what occurred in the system up to and including time $t$. In particular, if the message at time $t_0$ was triggered by an emulated departure, then $\mathcal{F}_{t_0}$ includes the age of the current actual job server $i$ is processing at time $t_0$. Conditioned on $\mathcal{F}_{t_0}$, and regardless of whether the message at time $t$ was triggered by an actual or emulated departure, by the memory-less property of the exponential distribution, the distribution of the actual queue length after time $t_0$ is determined by $Q_i(t_0)$, the arrivals after time $t_0$, and a "fresh" Poisson potential service process. 

Thus, for simplicity and with a slight abuse of notation, we set $t_0=0$ and write $A_i(t)$ for the number of arrivals to server $i$, $B_i(t)$ and $I_i(t)$ for the busy and idle times and $S_i(\mu_iB_i(t))$ for the number of departures during $(t_0,t_0+t]$, where $S_i(t)$ is a Poisson($1$) potential service process. We also write $Q_i(t)$ and $\tilde{Q}_i(t)$ instead of $Q_i(t_0+t)$ and $\tilde{Q}_i(t_0+t)$ and write $\tilde{B}_i(t)$ and $\tilde{I}_i(t)$ for the busy and idle times of the emulated queue. 

Using this simpler notation, the actual and emulated queue lengths are given by:
\begin{align*}
    Q_i(t)&=Q_i(0)+A_i(t)-S_i(\mu_iB_i(t))=Q_i(0)+A_i(t)-S_i(\mu_iB_i(t))+\mu_iB_i(t)-\mu_iB_i(t)\cr
    &=Q_i(0)+A_i(t)-S_i(\mu_iB_i(t))+\mu_iB_i(t)-\mu_it+\mu_iI_i(t)=X_i(t)+\mu_iI_i(t),
\end{align*}

where 
\begin{equation}\label{eq: def of x}
    X_i(t):=Q_i(0)+A_i(t)-S_i(\mu_iB_i(t))+\mu_iB_i(t)-\mu_it,
\end{equation}
and
\begin{align*}
    \tilde{Q}_i(t)&=Q_i(0)+A_i(t)-\floor{\mu_i\tilde{B}_i(t)}=Q_i(0)+A_i(t)-\floor{\mu_i\tilde{B}_i(t)}+\mu_i\tilde{B}_i(t)-\mu_i\tilde{B}_i(t)\cr
    &=Q_i(0)+A_i(t)-\floor{\mu_i\tilde{B}_i(t)}+\mu_i\tilde{B}_i(t)-\mu_it+\mu_i\tilde{I}_i(t)=\tilde{X}_i(t)+\mu_i\tilde{I}_i(t),
\end{align*}
where 
\begin{equation}\label{eq: def of tilde x}
    \tilde{X}_i(t):=Q_i(0)+A_i(t)-\floor{\mu_i\tilde{B}_i(t)}+\mu_i\tilde{B}_i(t)-\mu_it.
\end{equation}

We now leverage the Lipschitz continuity property of the one dimensional Skorohod reflection mapping to obtain an upper bound on the approximation error $|Q_i(t)-\tilde{Q}_i(t)|$.

For readability, we first present the one dimensional Skorohod problem.  

\begin{definition}[One dimensional Skorohod problem.] Let $x \in {\mathcal{D}}_{0}$. A pair $(z,y)\in {\mathcal{D}}_{0} \times {\mathcal{D}}_{0}$ is a solution of the one dimensional Skorohod problem for $x$ if the following hold:
\begin{enumerate}
    \item $z(t)=x(t)+y(t)$,  $t \geq 0$,
    \item $z(t) \geq 0$, $t \geq 0$,
    \item $y$ satisfies:
    \begin{enumerate}
        \item $y(0)=0$,
        \item $y$ is non decreasing,
        \item $\int_o^\infty z(t)dy(t)=0$ ($y$ does not increase in an interval $[s,t]$ whenever $z$ is positive in the interval).
    \end{enumerate}
\end{enumerate}

\end{definition}

By Theorem 6.1 in \cite{chen2001fundamentals}, given $x\in {\mathcal{D}}_{0}$, there exists a unique solution $(z,y)\in {\mathcal{D}}_{0} \times {\mathcal{D}}_{0}$ for the Skorohod problem. Moreover, the one dimensional Skorohod reflection mapping is Lipschitz continuous with respect to the uniform norm over finite time intervals. Sepcifically, if $(\tilde{z},\tilde{y})\in {\mathcal{D}}_{0} \times {\mathcal{D}}_{0}$ is the unique solution of the Skorohod problem for some $\tilde{x}\in {\mathcal{D}}_{0} $, then:

\begin{equation}\label{eq:lip con}
    \sup_{s \in [0,t]}|z(s)-\tilde{z}(s)|\leq 2\sup_{s \in [0,t]}|x(s)-\tilde{x}(s)|, \quad \forall t \geq 0.
\end{equation}

It is straightforward to check that the pairs $(Q_i(t),\mu_iI_i(t))$ and $(\tilde{Q}_i(t),\mu_i\tilde{I}_i(t))$ are the unique solutions for the Skorohod problem for $X_i(t)$ and $\tilde{X}_i(t)$, respectively. Therefore, using \eqref{eq:lip con}, we have:

\begin{equation}\label{eq:skorohod q bound}
    |Q_i(t)-\tilde{Q}_i(t)|\leq \sup_{s \in [0,t]}|Q_i(s)-\tilde{Q}_i(s)|\leq 2\sup_{s \in [0,t]}|X_i(s)-\tilde{X}_i(s)|, \quad \forall t \geq 0.
\end{equation}

Now, by the definitions of $X_i(t)$ and $\tilde{X}_i(t)$ in \eqref{eq: def of x} and \eqref{eq: def of tilde x}, we have: 
\begin{equation*}
    |X_i(s)-\tilde{X}_i(s)|\leq |S_i(\mu_iB_i(s))-\mu_iB_i(s)|+1, \quad \forall s\geq 0.
\end{equation*}
Therefore
\begin{align}\label{eq:skorohod x bound}
    2\sup_{s \in [0,t]}|X_i(s)-\tilde{X}_i(s)|&\leq 2\sup_{s \in [0,t]}|S_i(\mu_iB_i(s))-\mu_iB_i(s)|+2\cr
    & \leq 2\sup_{s \in [0,t]}|S_i(\mu_is)-\mu_is|+2, \quad \forall t\geq 0,
\end{align}
where the last inequality is due to the fact that $B_i(s) \leq s$. Combining \eqref{eq:skorohod q bound} and \eqref{eq:skorohod x bound} we obtain:
\begin{equation}\label{eq:queue length bound skorohod}
    |Q_i(t)-\tilde{Q}_i(t)|\leq 2\sup_{s \in [0,t]}|S_i(\mu_is)-\mu_is|+2, \quad \forall t\geq 0.
\end{equation}

Fix $x \in \{3,4,\ldots\}$. Define:

$$\tau_i:=\inf\{t \geq 0: |Q_i(t)-\tilde{Q}_i(t)|=x\},$$
and
$$\sigma_i:=\inf\{t \geq 0: |S_i(\mu_it)-\mu_it|\geq x/2-1\}.$$

By \eqref{eq:queue length bound skorohod}, we have:
\begin{equation*}
    x=|Q_i(\tau_i)-\tilde{Q}_i(\tau_i)|\leq 2\sup_{s \in [0,\tau_i]}|S_i(\mu_is)-\mu_is|+2,
\end{equation*}
namely,
\begin{equation*}
    \sup_{s \in [0,\tau_i]}|S_i(\mu_is)-\mu_is|\geq x/2-1,
\end{equation*}
which implies that $\tau_i \geq \sigma_i$ almost surely. Therefore:
$$\mathbbm{E}[\tau_i]\geq \mathbbm{E}[\sigma_i] \geq (x/2-1)^2/\mu_i,$$
where the last inequality is due to \eqref{eq:result on avg sigma exit time} in Theorem \ref{thm: poisson}, which concludes the proof.
\qed 
\end{proof}

\vspace{10pt}

\begin{reptheorem}{thm:summary 3}
In the setting and under the assumptions of Theorem \ref{thm:summary 2}, add the further assumption 
of an infinite backlog at the beginning of the interval $I_i$. Then, for every $x \in \{2,3,4,\ldots\}$, there exists a combination of an approximation algorithm and communication pattern under which $AQ(t)\leq x-1$ for all $t \in I_i$ and  $\mathbbm{E}[\tau_i]\geq x(x-1)/\mu_i$. Specifically, this holds for ET-$x$ with MSR.
\end{reptheorem}

\begin{proof}
The assumption of an infinite backlog at the beginning of the interval enables our analysis, since it implies that the idling times of the actual and emulated queues until the next message are zero. While this is an artificial assumption, it does shed light on the rate of communication we can expect when servers do not idle often. 

Considering that the idling times are zero, and adopting the simplifying notation and arguments used in the proof of Theorem \ref{thm:summary 2}, we have:
$$|Q_i(t)-\tilde{Q}_i(t)|=|S_i(\mu_it)-\floor{\mu_i t}|, \quad \forall 0 \leq t \leq \tau_i,$$
where time zero is treated as the beginning of the interval, $\tau_i$ its duration, and $S_i$ is a rate 1 Poisson process. Thus, by \eqref{eq:result on avg tau exit time} in Theorem \ref{thm: poisson}, we have:
$$\mathbbm{E}[\tau_i]\geq (x^2-x)/\mu_i,$$
which concludes the proof.
\qed 
\end{proof}
\section{The Approximation-Performance Link}\label{sec:approximation-performance link}
This section is devoted to studying the link between achieving reasonable approximations and the resulting performance. We should emphasize that while an asymptotic analysis will be used in this section to provide sharp theoretical insights, the  focus of our paper is to develop practical approximation and communication schemes that work well in non-asymptotic systems. In other words, the asymptotic limits should be thought of as providing additional sources of insight and assurance for the algorithms we propose, rather than as an accurate portrayal of the real environment used to drive algorithmic design.

We begin by describing the diffusion scaled system model followed by the main results and their proofs.

\subsection{Diffusion Scale Model}\label{subsec:sys model 2}
 
A sequence of systems indexed by $n \in \N$ is defined on a probability space $(\Om,\calF,\PP)$
as follows.
There is a fixed number $K \in \N$ of servers, labeled by $[K]:=\{1,\ldots,K\}$,
each with an infinite size buffer in which a queue can form. Each server is non-idling and offers service on a first-come-first-served basis. 
Jobs arrive sequentially to a single load balancer. The $k$th job to arrive (after time zero)
is referred to as {\it job $k$}. For simplicity, we assume that the system starts empty. 

When a job arrives to the load balancer, it must immediately route it to one of the servers. 
The job arrival processes is modeled by a modulated renewal process of rate $\la^n(\cdot)$,
where $\la^n$ is a deterministic, Borel measurable, locally integrable function $\R_+\to\R_+$.
To this end, a rate-1 renewal process $A$ is given, and
the arrival counting process $A^n$ is defined via the relation
\begin{equation} \label{eq:arrival}
A^n(t)=A\Big(\int_0^t\la^n(s)ds\Big), \quad t\ge0.
\end{equation}

For $i\in\ [K]$, let $\{T_i(l) : l \in \N\}$
be a sequence of strictly positive \iid RVs with mean $1$ and variance
$0<\sigma^2<\iy$.  Let constants $\bar{\mu}_1,\ldots,\bar{\mu}_K$ be given. Each server
works at a constant rate $\mu^n_i:=\bar{\mu}_in$, such that  it takes
\begin{equation*}
    T^n_i(k):= T_i(k)/\mu^n_i
\end{equation*} 
units of time for server $i$ to process the $k$th job
to arrive to it. Let $\{S_i\}$ be independent rate-1 renewal processes with inter-event times given by $\{T_i(k)\}$, namely 
\begin{equation} \label{eq:nominal pot serv proc}
S_i(t)=\sup \Big\{ l\geq 0: \sum_{k=1}^l T_i(k)\leq t \Big \},
\end{equation}
and define the potential service processes by
\begin{equation}\label{eq: pot serv proc}
    S_i^n(t)=S_i(\mu_i^n t),
\end{equation}
so that
\begin{equation*}
    S_i^n(t)=S_i(\mu_i^n t)=\sup \Big\{ l\geq 0: \sum_{k=1}^l T_i(k)\leq \mu_i^n t \Big \}=\sup \Big\{ l\geq 0: \sum_{k=1}^l T_i^n(k)\leq t \Big \}.
\end{equation*}
Namely,
$S^n_i(t)$ is the number of job departures from queue $i$
by the time the corresponding server has been busy for $t$ units of time. The $K+1$ processes $\{T_i\}$ and $A$ are assumed to be mutually independent.

Denote by $I^n_i(t)$ the cumulative idle time of server $i$ at time $t$. Next, $A^n_i$ and $D^n_i$ are counting processes for arrivals into buffer $i$,
and departures from buffer $i$, respectively. Let $Q^n_i(t)$ denote the queue length of the $i$th queue in the $n$th
system at time $t$ (this includes the job being processed at that time, if there is one), and denote $Q^n=(Q^n_1,\ldots,Q^n_N)$.
The relations between the processes $A^n_i$, $S^n_i$, $D^n_i$, $I^n_i$ and $Q^n_i$ are expressed by the following equations:
\begin{align}
&D^n_i(t)=S^n_i(t-I^n_i(t)) \label{eq:D vs S}\\
&Q^n_i(t)=A^n_i(t)-D^n_i(t),\label{eq: balance}
\end{align}
and the non-idling property 
\begin{equation*}
  \label{05}
  \int_{[0,\iy)}Q^n_i(t)dI^n_i(t)=0.
\end{equation*}

At each point in time, each server is allowed to communicate with the load balancer in the form of sending it a message. The content of a message sent by server $i$ at time $t$ is just its current queue length $Q^n_i(t)$. The \textit{communication pattern} dictates at which times the servers send messages to the load balancer. The load balancer stores in its memory an array of queue length approximations. Denote by $\tilde{Q}^n_{i}(t)$ the approximation of the queue in buffer $i$ at time $t$ in the $n$th system. Write $\tilde{Q}^n=(\tilde{Q}^n_{1},\ldots , \tilde{Q}^n_{K})$ for the queue length approximation array. The load balancer uses an \textit{approximation algorithm} to update the approximation array based on the information that is available to it (queue lengths from previous messages, job arrival and assignment decisions). We assume that the system starts empty and that the initial approximations are all zero, i.e., $Q^n_i(0)=\tilde{Q}^n_i(0)=0$, for all $i\in[K]$.

The load balancer uses the JSAQ load balancing algorithm, namely, when job $k$ arrives to the load balancer, it is routed to the buffer $j_k^n \in [K]$ with the smallest queue length approximation at the moment of its arrival. The tie breaking rule is immaterial to our results, but for concreteness, we assume that
ties are broken by
prioritizing buffer $i$ over buffer $j$ whenever $1\le i<j\le K$.
Thus, if $\tau_k^n$ denotes the time of arrival of the $k$th job in the $n$th system, then
the job is routed to the buffer
\begin{equation*}
j^n_k=\min\{i\in [K]:\tilde{Q}^n_{i}(\tau_k^n-)\le \tilde{Q}^n_{j}(\tau_k^n-)\text{ for all } j \in [K]\},
\end{equation*}
and the arrival processes $A^n_i$ satisfy
\begin{equation*}
A^n_i(t)=\sum_{k=1}^{A^n(t)}\mathbbm{1}_{\{j^n_k=i\}}.
\end{equation*}

\subsection{Main Approximation-Performance Link Results} \label{subsec:main results}

In this section we state and discuss the main results on the approximation-performance link. We first define what we refer to as admissible approximations followed by the definition of SDDP. 
Recall our notation $\R_+=[0,\iy)$, and for $f:\R_+\to\R^k$, $\|f\|_T=\sup_{t\in[0,T]}\|f(t)\|$.

\begin{definition}[Admissible Approximations]
We say that the approximations are admissible if:
\begin{equation}\label{property:aoa}
    \frac{1}{n^{1/2}}\max_{ i \in[K]}\left\|\tilde{Q}^n_{i}-Q^n_{i}\right\|_T \xrightarrow{n \rightarrow \infty}0 \text{ in probability.}
\end{equation}

\end{definition}

In Section \ref{sec: connecting comm to per}, we will show that our proposed approximation algorithms with suitably chosen parameters are admissible.

\begin{definition}[Sub-Diffusivity of the Deviation Process (SDDP)]
We say that the queue lengths undergo SDDP if:
\begin{equation}\label{property:ssc}
    \frac{1}{n^{1/2}}\max_{ i,j \in[K]}\left\|Q^n_{i}-Q^n_{j}\right\|_T \xrightarrow{n \rightarrow \infty}0 \text{ in probability.}
\end{equation}

\end{definition}

\subsubsection{Maximal Approximation Error and SDDP}

In what follows, we provide a sufficient condition on the maximal approximation error, such that if it holds, the queues undergo SDDP. Essentially, SDDP holds if the approximations are admissible and the input rate is large enough to overcome the service rate differences, thus pushing the queue lengths towards equalization. To state our condition, we look at the arrival and service rates during a time interval $[0,T]$. To this end,
let scaled versions of $\la^n$ and $\mu^n_i$ be given by
\begin{equation}\label{eq:scaled_rates}
\bar\la^n(t)=n^{-1}\la^n(t),
\qquad
\bar\mu_i=n^{-1}\mu^n_i,\qquad t\ge0,
\end{equation}
and denote
\begin{align*}
  \bar\la_{min}(T)=\inf_n\inf_{t\in[0,T]}\bar\la^n(t),
\qquad
  \bar\la_{max}(T)=\sup_n\sup_{t\in[0,T]}\bar\la^n(t),
  \qquad
  \bar\mu_{min}=\min_i\bar\mu_i.
\end{align*}

We now state the SDDP result.

\begin{theorem}[SDDP]\label{thm:SDDP}
Fix $T>0$. Assume that: \\
(i) $\bar\la_{max}(T)<\iy$. \\
(ii) $\bar\la_{min}(T)>\sum_{i\in[K]}\bar\mu_i-K\bar\mu_{min}.$\\
(iii) the approximations are admissible, namely \eqref{property:aoa} holds: 
$$\frac{1}{n^{1/2}}\max_{ i \in[K]}\left\|\tilde{Q}^n_{i}-Q^n_{i}\right\|_T \xrightarrow{n \rightarrow \infty} 0 \text{ in probability.}$$
Then the queue lengths undergo SDDP and \eqref{property:ssc} holds: 
$$ \frac{1}{n^{1/2}}\max_{ i,j \in[K]}\left\|Q^n_{i}-Q^n_{j}\right\|_T \xrightarrow{n \rightarrow \infty} 0 \text{ in probability.} $$
\end{theorem}

\begin{remark}
At a high level, the asymptotic regime used in the above theorem can be understood as one where the amount of resources per server and total demand tend to infinity, while the number of servers stays fixed. This is in part based on the observation that for many ``typical'' applications such as web servers for a medium-sized e-commerce site, the number of servers needed tends to be moderate. An illustrative calculation is as follows: A web server with an 8-core CPU \citep{servermania_small_business_server} can handle up to about $250 \times 8 =2,000$ concurrent requests \citep{Britten2022}. According to a 2023 Hubspot report, more than 50\% of websites receive less than 100,000 visitors per month \citep{hubspot_website}. The highest percentile of web traffic described in this report is 2 million visits per month for the top  1\% of the business-to-consumer (B2C) sites. This averages to about 66,000 visits per day. Using this number as a crude peak-traffic estimate by assuming that all daily visitors occur during a short time window, we see that the number of servers needed to serve such concurrent traffic is about $66,000 / 2,000 = 33$ servers.

\end{remark}

\begin{remark}
Note that for a homogeneous system, i.e., all server rates are equal, condition (ii) corresponds to $\bar\la_{min}(T)>0$. Another interesting special case is where the arrival rate does not vary with time and a critical load condition is assumed. Thus the system is in the classic heavy traffic setting. Notice that condition (ii) holds for this case. In addition, it is well known \citep{atar2019subdiffusive} that in this case the (scaled) queue lengths converge weakly to the same stochastic process, a phenomenon called State-Space-Collapse. 
\end{remark}


\subsubsection{SDDP and Asymptotic Optimality}

Our next set of results relates the performance of any system which undergoes SDDP, to a single server queue with the combined processing rate of all of the servers. For ease of exposition, we introduce only the necessary notation and details to state our results, and defer the rest to the proofs section.

For $i \in [K]$, denote by $W^n_i(t)$ the \textit{nominal} workload at server $i$ at time $t$. This is defined as the time it would take a server working at a rate equal to 1 to complete all of the jobs that are currently in queue $i$ (including the one in service if there is any). Thus, the quantity $\sum_iW_i^n(t)$ equals the total nominal workload at time $t$. Now, consider a single server queue with processing rate $\sum_i \mu^n_i$ that is coupled to the multi-server system such that the arrival process and the service time requirements of jobs are identical. Denote the nominal workload in the single server system at time $t$ by ${\mathcal{W}}^n(t)$.

First, we present a general result which states that the nominal workload of \textit{any} $K$-server system must be at least that of the single server system.
\begin{theorem}[General Lower Bound]\label{thm:workload lower bound}
For any $K$-server load balancing system we have:

\begin{equation*}
    \sum_{i=1}^K W_i^n(t) \geq {\mathcal{W}}^n(t), \qquad t \geq 0.
\end{equation*}
\end{theorem}

We now present our main asymptotic optimality result\rev{s}.

\begin{theorem}[Asymptotic Optimality: Nominal Workload]\label{thm:as optimal}
Fix $T>0$. Assume that the service time requirements of all arriving jobs are uniformly upper bounded by some constant. If the queue lengths undergo SDDP, then:
\begin{equation*}
       \frac{1}{n^{1/2}}\left\|\sum_{i=1}^K W^n_i-{\mathcal{W}}^n\right\|_T\xrightarrow{n \rightarrow \infty}0 \text{ in probability.}
\end{equation*}
\end{theorem}

The importance of Theorem \ref{thm:workload lower bound} is that it establishes that there exist no algorithm which can achieve lower nominal workload than a single server system. Then, Theorem \ref{thm:as optimal} implies that if the queue lengths undergo SDDP, then the nominal workload is optimal up to an additive term of order $o(n^{1/2})$.

An immediate Corollary to Theorem \ref{thm:as optimal} is the asymptotic optimality with respect to the average \textit{actual} workload in \textit{homogeneous systems}. The actual workload at server $i$ is defined as the time it would take server $i$ to complete all of the jobs that are currently in its queue (including the one in service if there is any). Thus, it is given by $ W^n_i/\mu^n_i$. The average actual workload is given by $(1/K)\sum_{i=1}^K W^n_i/\mu^n_i$. The actual workload at the single server queue is given by ${\mathcal{W}}^n/\left(\sum_{i=1}^K\mu^n_i\right) $.

\begin{corollary}[Asymptotic Optimality: Average Actual Workload in Homogeneous Systems]\label{cor:ao_aw_h}
Fix $T>0$. Assume that the service time requirements of all arriving jobs are uniformly upper bounded by some constant. If the queue lengths undergo SDDP and $\mu^n_i=\mu^n, \forall i\in[K]$, then:
\begin{equation*}
       \frac{1}{n^{1/2}}\left\|\frac{1}{K}\sum_{i=1}^K \frac{ W^n_i}{\mu^n_i}-\frac{{\mathcal{W}}^n}{\sum_{i=1}^K\mu^n_i}\right\|_T\xrightarrow{n \rightarrow \infty}0 \text{ in probability.}
\end{equation*}
\end{corollary}

The proof of Corollary \ref{cor:ao_aw_h} follows from substituting $\mu^n_i=\mu^n$ and using Theorem \ref{thm:as optimal}. Thus, an arriving job to a homogeneous system would see, on average, approximately the same workload as it would see upon arriving to the single server system. Thus, it will experience similar queueing times and therefore will have similar completion times.

\begin{remark}
    For heterogeneous systems, we cannot expect JSQ, JSAQ, or any other policy for which SDDP holds to lead to asymptotic optimality with respect to the actual workload. If the queue lengths are balanced, then the actual workloads are imbalanced, with the level of imbalance depending on the different service rates. We refer the reader to \cite{atar2019replicate}, Section 4.2, for more details. However, our results do indicate that one can expect the performance of JSAQ to be similar to that of JSQ, even when the communication rate is much slower than the rate at which events occur in the system. One would expect that a different policy which balances the actual workloads (e.g., join the queue with the minimal $Q^n_i/\mu^n_i$) would lead to asymptotic optimality with respect to the actual workload. However, we were not able to prove this and we leave it as an open question.
\end{remark}

\subsection{Proofs}
In this section we reiterate each result from Section \ref{subsec:main results} for readability and provide its proof. We reiterate a few definitions from Section \ref{sec: CARE} which we will use in the proofs below. Recall that $t_i^m$ refers to the (possibly random) time the $m$th message was sent from server $i$ to the load balancer, and that $t_i^0$ is defined as zero for all $i$. For every $n$ these messaging times are generally different but we avoid adding $n$ to the notation for simplicity.

The approximations are given by:

\begin{equation}\label{eq:approximations diffusive section}
    \tilde{Q}^n_i(t)=Q^n_i(t_i^m)+A^n_i(t_i^m,t]-\tilde{D}^n_i(t_i^m,t], \quad t \in [t_i^m,t_i^{m+1}),
\end{equation}
where $Q^n_i(t_i^m)$ is the queue length of server $i$ at the most recent messaging time $t_i^m$, $A^n_i(t_i^m,t]$ is the number of jobs the load balancer sent server $i$ during $(t_i^m,t]$, and $\tilde{D}^n_i(t_i^m,t]$ is the estimated number of departures from the queue during $(t_i^m,t]$. The actual queue lengths can be written using the messaging times as follows (see \eqref{eq:actual ql}):
\begin{align}\label{eq:actual ql diffusion section}
    Q^n_i(t)=Q^n_i(t_i^m)+A^n_i(t_i^m,t]-D^n_i(t_i^m,t],
\end{align}
where $D^n_i(t_i^m,t]$ equals the actual number of departures from queue $i$ during $(t_i^m,t]$.

Finally, denote centered and scaled versions of the primitive processes $A$ and $S_i$ by:
\begin{equation} \label{eq: sclaed primitives}
  \hat A^n(t)=n^{-1/2}(A(nt)-nt),\quad \hat S^n_i(t)=n^{-1/2}(S_i^n(t)-\mu_i^nt).
\end{equation}
These processes jointly converge to mutually independent Brownian Motions with zero drift and
infinitesimal variance $1$ and $\bar{\mu}_i \sigma^2$, respectively (see Theorem 5.11 of \cite{chen2001fundamentals}). The processes $\hat{A}^n$ and $\hat S^n_j$ are $\calC$-tight, as processes that converge
to Brownian Motions.
Recalling that $\II$ denotes integration, by \eqref{eq:arrival}, the process $A^n$ is given as $A^n=A\circ\II\la^n$. Therefore
\begin{equation}\label{eq: scaled A}
n^{-1/2}(A^n(t)-\II\la^n(t))=n^{-1/2}(A(n\II\bar\la^n(t))-n\II\bar\la^n(t))
=(\hat A^n\circ\II\bar\la^n)(t),
\end{equation}
where $\bar\la^n(t)$ is defined in \eqref{eq:scaled_rates}.

\subsubsection{Maximal Approximation Error and SDDP}
We now prove Theorem \ref{thm:SDDP}, restated below. 

\begin{reptheorem}{thm:SDDP}[SDDP]
Fix $T>0$. Assume that: \\
(i) $\bar\la_{max}(T)<\iy$. \\
(ii) $\bar\la_{min}(T)>\sum_{i\in[K]}\bar\mu_i-K\bar\mu_{min}.$\\
(iii) the approximations are admissible, namely \eqref{property:aoa} holds: 
\begin{equation}
    \frac{1}{n^{1/2}}\max_{ i \in[K]}\left\|\tilde{Q}^n_{i}-Q^n_{i}\right\|_T \xrightarrow{n \rightarrow \infty} 0 \text{ in probability.}
    \label{eq:aoa_rep}
\end{equation}
Then the queue lengths undergo SDDP and \eqref{property:ssc} holds, 
\begin{equation}
    \frac{1}{n^{1/2}}\max_{ i,j \in[K]}\left\|Q^n_{i}-Q^n_{j}\right\|_T \xrightarrow{n \rightarrow \infty} 0 \text{ in probability.} 
    \label{eq:ssc_rep}
\end{equation}
\end{reptheorem}

\begin{proof} 

\emph{Proof Sketch}: Let us first outline the main idea of the proof, which can be roughly divided into three  parts. 
\begin{enumerate}
    \item We begin with a proof-by-contradiction argument. Suppose that SDDP does {not} occur, then it would necessarily imply that there exists a pair of queues $i, j$  such that $Q_i(\tau)$ is substantially larger than $Q_j(\tau)$ by $\Omega(n^{1/2})$ at some time $\tau$. We call this  a divergence event, and queue $i$ the superior queue. 
    \item The rest of the proof aims to show that a divergence event is very unlikely to occur. We make two observations. First, if $\tau$ is the first time when a divergence event occurs, then the superior queue $i$ must be amongst the longest at that time. Second, let $\sigma$ be the  last time before $\tau$ when queue $i$ received an arrival. Because the load balancer runs the JASQ policy, we know that queue $i$ must be amongst the shortest queues in the system at time $\sigma$ for it would have not received an arrival otherwise. 
    \item Combining both observations above, we see that during the span of the time interval $(\sigma, \tau]$, queue $i$  went from being amongst the shortest in the system to being amongst the longest, despite having only received one single arrival. We finish the proof by showing that, under diffusion scaling this is indeed extremely unlikely under the JASQ policy with a sufficiently accurate state approximation. This completes the proof. 
\end{enumerate}

We now present the  proof. Fix $T>0$ and $\eps>0$.  Define the following two events corresponding to the compliment of SDDP and the state approximations being sufficeintly accurate,  respectively: 
\begin{align*}
     A^n_\eps:=   &  \Big \{\max_{i,j\in [K] }\|Q^n_i-Q^n_j\|_T>\eps n^{1/2}\Big \}, \\
    B^n_\eps:=  &  \Big \{\max_{i \in[K]}\|\tilde{Q}^n_{i}-Q^n_{i}\|_T<\frac{\eps}{4(K-1)} n^{1/2}\Big \}.
\end{align*}
We obtain 
\begin{align*}
    \limsup_{n \to \infty} \mathbbm{P}(A^n_\eps)= & \limsup_{n \to \infty} \mathbbm{P}(A^n_\eps \cup B^n_\eps)-\mathbbm{P}(B^n_\eps)+\mathbbm{P}(A^n_\eps \cap B^n_\eps) \\
= & \limsup_{n \to \infty} \mathbbm{P}(A^n_\eps \cap B^n_\eps), 
\end{align*}
where the second equality follows from $\mathbbm{P}(B^n_\eps)\xrightarrow{n \to \infty}1$ by \eqref{eq:aoa_rep}, which also implies $\mathbbm{P}(A^n_\eps \cup B^n_\eps)\xrightarrow{n \to \infty}1$.  Thus, to prove the theorem,  it suffices to show that 
\begin{equation}
    \lim_{n \to \infty} \mathbbm{P}(A^n_\eps \cap B^n_\eps)=0. 
    \label{eq:pb_AB_to0}
\end{equation}

Define the stopping time representing the first time a divergence occurs: 
\begin{equation}
\tau:=\tau^n=\inf\Big\{t:\text{ there exist } i,j\in [K] \text{ such that }
Q^n_i(t)-Q^n_j(t)
>\eps n^{1/2}\Big\}.
\label{eq:tau_def}
\end{equation}
and the event that a divergence occurs within the time horizon: 
\begin{equation}
    C^n := \{\tau < T\}. 
\end{equation}
It follows by definition that $A^n_\eps \subset C^n$, and  $A^n_\eps \cap B^n_\eps \subset C^n \cap B^n_\eps$. Therefore, it suffices to prove that
\begin{equation}
    \lim_{n \to \infty} \mathbbm{P}(C^n\cap B^n_\eps)=0.
\end{equation}

For every sample path in the event $C^n\cap B^n_\eps$, there exist $i,j\in[K]$ such that
$Q^n_{i}(\tau)- Q^n_{j}(\tau)>\eps n^{1/2}$. 
Fix such $i$ and $j$, and we will refer to queue $i$ as the superior queue. 
Denote by $\sigma$ the last time before $\tau$ such that the superior queue receives a job: 
\begin{equation*}
    \sigma:=\sigma^n=\sup\Big\{t< \tau: A_i^n(t)-A_i^n(t^-)=1 \Big\}.
\end{equation*}
Time $\sigma$ is well defined since we assumed that the system starts empty, and by the definition of $\tau$ we have $Q^n_i(\tau)>0$, so $\tau>0$ and there must have been arrivals to buffer $i$ before time $\tau$. 

We will use the following short-hand notation: 
\begin{equation}
    Q^n_S(t) := \sum_{k \in [K]\setminus\{i\}}Q^n_k(t), 
    \label{eq:Q_Sdef}
\end{equation}
where $i$ is the index of the superior queue. For any interval $J=(t_1,t_2]$ and a function $f$, let $f(J]$ denote the difference in $f$ between the two end points of $J$:
\begin{equation}
    f(J] = f(t_2)-f(t_1). 
\end{equation}
The next lemma establishes some crucial properties of the superior queue.

\begin{lemma} 
\label{lem:superior_i_prop}
Let $i$ and $j$ be defined as in \eqref{eq:tau_def}, where $i$ is the superior queue. The following holds for all sufficiently large $n$ and sample paths in $C^n\cap B^n_\eps$: 
\begin{enumerate}
    \item The superior queue is amongst the largest queues at time $\tau$: $Q_i^n(\tau) = \max_{i'}Q^n_{i'}(\tau)$. 
    \item There is no arrival to queue $i$ at time $t=\tau$, and consequently, queue $i$ does not receive any jobs during $(\sigma,\tau]$. 
    \item The change in queue $i$  over $t \in (\sigma, \tau]$ is substantially larger than the average change in the rest of the queues during this period: 
    \begin{equation}
        (Q^n_i(\tau) - Q^n_i(\sigma))-\frac{1}{K-1} \sum_{i' \neq i} (Q^n_{i'}(\tau) - Q^n_{i'}(\sigma)) \geq  \frac{\eps}{2}n^{1/2}\frac{1}{K-1}, 
    \end{equation}
    or, using the short-hand notation, 
    \begin{equation}
        Q^n_i(J]-\frac{1}{K-1}Q^n_S(J] \geq  \frac{\eps}{2}n^{1/2} \frac{1}{K-1}, 
        \label{eq:final inequality}
    \end{equation}
    where $J = (\sigma, \tau]$. 
\end{enumerate}
\end{lemma}

\emph{Proof.} \emph{Claim 1}. The claim follows directly from the definition of $\tau$.  Since $\tau$ is the first time such that there is a difference of $\epsilon n^{1/2}$ between any two queues, $Q^n_i$ and $Q_j^n$ must be the largest and smallest queue lengths, respectively, across all queues in the system at time $\tau$. 

\emph{Claim 2}. We argue that for all sufficiently large $n$, as a result of restricting to the event $B^n_\epsilon$ (i.e., accurate approximation), buffer $i$ cannot receive a new job at time $\tau$.  Suppose that $n$ is large enough such that $\epsilon n^{1/2}>6$.  Since the queue length processes can only jump by 1 at any given time, we have that just before time $\tau$, namely at $\tau^-$, queues $i$ and $j$ cannot be closer than $\epsilon n^{1/2}-2$ (the $-2$ factor corresponds to the event of having an arrival to buffer $i$ and departure from buffer $j$ at $\tau$ simultaneously). 

Since the sample path is contained in $B^n_{\epsilon}$, the approximations cannot be deviate more than $(\epsilon n^{1/2})/4$ from the actual queue lengths. We have that:  
\begin{equation*}
    \tilde{Q}^n_i(\tau^-)-\tilde{Q}^n_j(\tau^-)\geq \epsilon n^{1/2}-2-2\times(\epsilon n^{1/2})/4=0.5\epsilon n^{1/2}-2>1.
\end{equation*} 
Specifically, we have $\tilde{Q}^n_j(\tau^-)<\tilde{Q}^n_i(\tau^-)$ and therefore, by the definition of the JSAQ load balancing algorithm, buffer $i$ cannot receive a job at time $\tau$. This proves the second claim. 

\emph{Claim 3}.  By the first claim, at time $\tau$, we have
\begin{equation}\label{eq:tau rest}
    Q^n_i(\tau) - Q^n_k(\tau)\geq 0 \quad \forall k  \neq i, j, 
\end{equation}
and
\begin{equation}\label{eq:tau one}
    Q^n_i(\tau)-Q^n_k(\tau)>\eps n^{1/2}, \quad k =j. 
\end{equation}
Using \eqref{eq:tau rest} and \eqref{eq:tau one} and summing over all $k \neq i$ yields:  
\begin{equation}\label{eq:tau inequality}
    (K-1)Q^n_i(\tau)-Q^n_S(\tau)>\eps n^{1/2}, 
\end{equation}
where $Q^n_S(t) = \sum_{k \in [K]\setminus\{i\}}Q^n_k(t)$, as defined in \eqref{eq:Q_Sdef}.

Next, by the definition of $\sigma$, the approximation $\tilde{Q}^n_{i}(\sigma)$ must satisfy
\begin{equation}\label{eq:sigma 1}
    \tilde{Q}^n_{i}(\sigma)\leq \tilde{Q}^n_{k}(\sigma), \text{ for all } k \in[K]\setminus \{i\}.
\end{equation}
Subtracting and adding $Q^n_{i}(\sigma)$ from the left hand side and then $Q^n_{k}(\sigma)$ from the right hand side of \eqref{eq:sigma 1} yields 

\begin{equation*}
    \tilde{Q}^n_{i}(\sigma)-Q^n_{i}(\sigma)+Q^n_{i}(\sigma) \leq \tilde{Q}^n_{k}(\sigma)-Q^n_{k}(\sigma)+Q^n_{k}(\sigma), \text{ for all } k \in[K]\setminus \{i\},
\end{equation*}
and, after rearranging, we obtain
\begin{equation}\label{eq:sigma 3}
    Q^n_{i}(\sigma)-Q^n_{k}(\sigma) \leq 2\max_{k\in[K]}\vert Q^n_{k}(\sigma)-\tilde{Q}^n_{k}(\sigma)\vert.
\end{equation}
Since the event we analyze is contained in $B^n_{\eps}$, we have for all $k \in[K]$
\begin{equation}\label{eq:sigma 4}
    \vert Q^n_{k}(\sigma)-\tilde{Q}^n_{k}(\sigma)\vert \leq \frac{\eps}{4(K-1)}n^{1/2}.
\end{equation}
Combining \eqref{eq:sigma 3} and \eqref{eq:sigma 4} yields
\begin{equation}\label{eq:sigma 5}
    Q^n_{i}(\sigma)-Q^n_{k}(\sigma) \leq \frac{\eps}{2(K-1)}n^{1/2}.
\end{equation}
Summing the inequalities in \eqref{eq:sigma 5} over $k \in [K]\setminus \{i\}$, we obtain
\begin{equation}\label{eq:sigma 6}
    (K-1)Q^n_{i}(\sigma)-Q^n_{S}(\sigma) \leq \frac{\eps}{2}n^{1/2}.
\end{equation}

Set $J = (\sigma, \tau]$.  Combining \eqref{eq:tau inequality} and \eqref{eq:sigma 6}, we have 
\begin{align}\label{eq:sigma 7}
(K-1)Q^n_i(J]-Q^n_S(J]\geq  \frac{\eps}{2}n^{1/2}.    
\end{align}
This completes the third and last claim in Lemma \ref{lem:superior_i_prop}. 
\qed 

Using Lemma \ref{lem:superior_i_prop} and Equation \eqref{eq:sigma 7}, we now complete the proof by showing that the event $C^n\cap B^n_{\eps}$ is very unlikely to occur under the JASQ policy for large $n$. 

The intuition is as follows. The argument considers two cases, namely, if the interval $J$ is short, or long. If $|J|$ is small, then the difference between queue lengths at any server must also be small, so the left hand side of \eqref{eq:sigma 7} cannot exceed the right hand side which goes to infinity with $n$. If $|J|$ is not small, consider the terms $Q^n_i(J]$ and $Q^n_S(J]$ on the left hand side of \eqref{eq:sigma 7}. On the one hand, by Claim 2 of Lemma \ref{lem:superior_i_prop}, no jobs are routed to server $i$ during the interval $J$. Thus the queue length at server $i$ can only decrease during $J$, and the difference $Q^n_i(J]$ is roughly determined by $-\mu_i^n|J|$. On the other hand, all jobs are routed to the other servers during $J$, thus the difference $Q^n_S(J]$ is roughly determined by $\int_{J}\lambda^n(u)du
-\sum_{j \in[K]}\mu_j^n|J|$. As a result, after dividing by $n^{1/2}$ and some rearrangement, the left hand side of \eqref{eq:sigma 7} roughly behaves as 
$$n^{1/2}\Big\{-K\bar\mu_i|J|-
    \int_{J}\bar\lambda^n(u)du
+\sum_{j \in[K]}\bar\mu_j|J|\Big\}.$$
If $|J|$ is not small, the left hand side of \eqref{eq:sigma 7} (after diving by $n^{1/2}$) goes to zero with $n$ while the right hand side is constant, thus indicating that \eqref{eq:sigma 7} is unlikely to hold. 

We now proceed with the proof. Define the events: 
\begin{align*}
\Om^n(i_0)=\Big\{\exists s,t\in[0,T], s<t,\ \text{s.t.~} &
(K-1)Q^n_{i_0}(s,t]-Q^n_S(s,t]\geq  \frac{\eps}{2}n^{1/2}, \ \text{and} \\
&\text{there is no arrival to $i_0$ during $(s,t]$ } \Big\}, 
\end{align*}
where $Q^n_S=\sum_{k \in [K]\setminus\{i_0\}}Q^n_k$. Then, by Lemma \ref{lem:superior_i_prop} considering all possible realizations of the superior queue index $i$, we have that $C^n\cap B^n_{\eps} \subset \cup_{i_0 \in [K]} \Om^n(i_0)$. Using the union bound, this further implies that: 
\begin{align}\label{eq:union bound}
\PP(C^n\cap B^n_{\eps})\le\sum_{i_0=1}^{K}
\PP(\Om^n(i_0)). 
\end{align}
To prove the result,
it thus suffices to show that each summand $\PP(\Om^n(i_0))$
converges to zero with $n$.
We focus on fixed, deterministic $i_0$, and, with a slight abuse
of notation, refer to it again as $i$. We also write $J$ for the corresponding
time interval $(s,t]$, and $\Om^n$ for $\Om^n(i)$.

On the event $\Om^n$, using \eqref{eq: balance}, we have
\be \label{eq: decompose Q}
(K-1)Q^n_i(J]-Q^n_S(J]=(K-1)A^n_i(J]-(K-1)D^n_i(J]-A^n_S(J]+ D^n_S(J],
\ee
where $X_S=\sum_{j\in [K]\setminus\{i\}}X_j$ for $X=Q^n, A^n$ and $D^n$.
On this event, no jobs are routed to server $i$ by the load balancer during $(s,t]$. Thus $A^n_i(J]=0$ and $A^n_S(J]=A^n(J]$. Also, server $i$ is non-idling during $(s,t]$.
Therefore, if we denote for $j\in [K]$,
$J_j=(s-I^n_j(s),t-I^n_j(s)]$, we have $D^n_i(J]=S^n_i(J_i]$.
Servers $j\in [K]\setminus \{i\}$ may idle during this time period,
but since $S^n_j$ has non-decreasing sample paths, we have, by \eqref{eq:D vs S},
that $D^n_j(J] \leq S^n_j(J_j]$ for $j\in [K]\setminus \{i\}$.
Hence by \eqref{eq:final inequality} in Claim 3 of Lemma \ref{lem:superior_i_prop} and \eqref{eq: decompose Q},
\be \label{eq: a and s}
-(K-1)S^n_i(J_i]-A^n(J]+ \sum_{j \in[K]\setminus\{i\}}S^n_j(J_j] \geq \frac{\eps}{2} n^{1/2}.
\ee
By \eqref{eq: sclaed primitives} and \eqref{eq: scaled A} we have
\begin{align*}
    A^n(J]=A^n(t)-A^n(s)&=n^{1/2}(\hat A^n\circ\II\bar\la^n)(t)+\II\la^n(t)-n^{1/2}(\hat A^n\circ\II\bar\la^n)(s)-\II\la^n(s)\cr
    &=n^{1/2}(\hat A^n\circ\II\bar\la^n)(J]+\II\la^n(J]
\end{align*}
and for all $j \in [K]$,
\begin{align*}
    S^n_j(J_j]&=S^n_j(t-I^n_j(s))-S^n_j(s-I^n_j(s))\cr
    &=n^{1/2}\hat{S}^n_j(t-I^n_j(s))+\mu^n_j(t-I^n_j(s))-\big(n^{1/2}\hat{S}^n_j(s-I^n_j(s))+\mu^n_j(s-I^n_j(s))\big)\cr
    &=n^{1/2}\hat S^n_j(J_j]+\mu^n_j(t-s).
\end{align*}
Thus, after dividing by $n^{1/2}$, we can write \eqref{eq: a and s} as
\begin{equation*}
-(K-1)\hat{S}^n_i(J_i]-(\hat A^n\circ\II\bar\la^n)(J]+  \sum_{j \in[K]\setminus\{i\}}\hat S^n_j(J_j]+Y^n \geq \frac{\eps}{2},
\end{equation*}
where
\begin{align*}
    Y^n&=n^{1/2}\Big\{-(K-1)\bar\mu_i(t-s)-
    \int_{J}\bar\lambda^n(u)du
+\sum_{j \in[K]\setminus\{i\}}\bar\mu_j(t-s)\Big\} \cr
&=n^{1/2}\Big\{-K\bar\mu_i(t-s)-
    \int_{J}\bar\lambda^n(u)du
+\sum_{j \in[K]}\bar\mu_j(t-s)\Big\}.
\end{align*}
To bound $Y^n$, we obtain
\begin{equation*}
    Y^n\leq n^{1/2}\Big\{
-K\bar\mu_{min}-
\inf_{u\in[0,T]}\bar\la^n(u)
+\sum_{j \in[K]}\bar\mu_j
\Big\}(t-s)\leq  n^{1/2}\Big\{
-K\bar\mu_{min}-
\bar\la_{min}(T)
+\sum_{j \in[K]}\bar\mu_j
\Big\}(t-s) 
\end{equation*}

Hence by condition (ii) in the Theorem \ref{thm:SDDP}, for some $\eps_0>0$, $Y^n\le-\eps_0n^{1/2}(t-s)$ on the event $\Om^n$.
As a result, 
on the same event we have
\begin{equation}\label{eq:modulus}
w_T(\hat A^n\circ\II\bar\la^n,t-s)
+K\sum_{j=1}^Kw_T(\hat S^n_j,t-s)-\eps_0n^{1/2}(t-s)\geq \frac{\eps}{2},
\end{equation}
where we recall that $w_T$ is the modulus of continuity, which is defined, for $f:\mathbbm{R}_+ \rightarrow \mathbbm{R}^k$ and $\theta>0$, by:
\begin{equation*}
    \omega_T(f,\theta)=\sup_{0\leq s<u\leq s+\theta\leq T}\norm{f(u)-f(s)}.
\end{equation*}

Fix a sequence $\{r_n\}$ such that $r_n \rightarrow 0$ and $n^{1/2}r^n \rightarrow \infty$.
Considering the two cases $t-s\le r_n$ and $t-s>r_n$,
it follows from \eqref{eq:modulus} that $\PP(\Om^n)\le p^n_1+p^n_2$, where
\[
p^n_1=\PP\Big(w_T(\hat A^n\circ\II\bar\la^n,r_n)
+K\sum_{j=1}^Kw_T(\hat S^n_j,r_n)\ge\frac{\eps}{2}\Big),
\]
\[
p^n_2=\PP\Big(2\|\hat A^n\circ\II\bar\la^n\|_T
+2K\sum_{j=1}^K\|\hat S^n_j\|_T
\ge \eps_0n^{1/2}r_n\Big).
\]
The processes $\hat{A}^n$ and $\hat S^n_j$ are $\calC$-tight, as processes that converge
to Brownian Motions. Moreover, by the assumed uniform bound on $\bar\la^n(t)$,
there exists a deterministic constant $c_1$ (independent of $n$)
such that $\II\bar\la^n$ is bounded by $c_1$
on the time interval $[0,T]$, as well as $c_1$-Lipschitz on it.
It follows that $\hat A^n\circ \II\bar\la^n$ is also $\calC$-tight. As a result,
both $p^n_1$ and $p^n_2$ converge to zero. This shows that 
\begin{equation}
    \PP(\Om^n(i))\to 0, \quad \forall i\in[K]. 
\end{equation}
By \eqref{eq:union bound}, this completes the proof of Theorem \ref{thm:SDDP}. 
\qed \\

\end{proof}


\subsubsection{SDDP and Asymptotic Optimality}


The main result of this section is that \textit{any} load balancing policy under which SDDP occurs is optimal in the sense that the \textit{nominal workload} in the system is minimal, up to a factor of order $o(n^{1/2})$, at any point in time. The nominal workload refers to the time it will take a server with a processing rate of 1 to serve all current jobs in the system. Moreover, the $K$ server system, where each server $i$ processes jobs at rate $\mu^n_i$, acts as a single server queue with the server processing rate equalling $\sum_{i=1}^K\mu^n_i$. In particular, this applies to all the approaches we considered provided that the conditions of Theorem \ref{thm:SDDP} hold.

We prove these results in two stages. First, we couple the $K$ server system with a corresponding single server system such that they have the same input of jobs. We prove that the nominal workload in the single server system is a lower bound for the nominal workload in the $K$ server system \textit{for any load balancing algorithm}, including, but not limited to, approximation-based.
Second, we use the SDDP result to prove that the (scaled) difference between the nominal workload in the $K$ server system and the single server system converges to zero, thus proving asymptotic optimality.\\

\noindent \textbf{The single server system.} 
Let a system with a single server, which we refer to as the $S$-system, with an infinite buffer in which a queue can form, be given. The server is work conserving and provides first-come-first-serve service. Assume that the system starts empty. There is a single stream of incoming jobs, denoted by $\mathcal{A}^n$. Denote by ${\mathcal{T}}(k)$ the nominal service requirement of the $k$th job to enter the $S$-system. This is the time it takes a rate-1 server to complete this job. We assume that $\{{\mathcal{T}}(k)\}$ is a sequence of i.i.d.~random variables with mean 1 and variance $\sigma^2$. The actual processing time of job $k$ is defined as
\begin{equation}\label{eq:S system ser req}
    {\mathcal{T}}^n(k):={\mathcal{T}}(k)/ \mu^n,
\end{equation}
where $\mu^n$ is the deterministic processing rate of the server. Denote by ${\mathcal{I}}^n,{\mathcal{S}}^n,{\mathcal{D}}^n,{\mathcal{Q}}^n$ the idle-time, potential service, departure and queue length processes, which are defined analogously to \eqref{eq:nominal pot serv proc}, \eqref{eq: pot serv proc}, \eqref{eq:D vs S} and \eqref{eq: balance}.

We will be interested in the nominal workload in the system. To this end, define the input nominal workload process by
\begin{equation}\label{eq:def of input nom wo S}
    {\mathcal{W}}^n_{\mathcal{A}}(t):=\sum_{k=1}^{{\mathcal{A}^n}(t)}{\mathcal{T}}(k).
\end{equation}
By \eqref{eq:S system ser req}, the nominal workload is drained at a rate of $\mu^n$, whenever the system is not empty. Thus the nominal workload in the $S$-system, which we denote by ${\mathcal{W}}^n$, is given by
\begin{equation}\label{eq:nom w S}
    {\mathcal{W}}^n(t)={\mathcal{W}}^n_{\mathcal{A}}(t)-\mu^n(t-{\mathcal{I}}^n(t)).\\
\end{equation}

\noindent \textbf{The nominal workloads in the $K$ server system.} 
In what follows, we refer to the $K$ server system as the $K$-system. Analogously to \eqref{eq:def of input nom wo S} and \eqref{eq:nom w S}, define the input nominal workload processes by
\begin{equation*}
    W^n_{A,i}(t):=\sum_{k=1}^{A^n_i(t)}T_i(k), \qquad i \in [K].
\end{equation*}
The nominal workload processes, denoted by $\{W^n_i\}$, satisfy
\begin{equation*}
    W^n_i(t)=W^n_{A,i}(t)-\mu^n_i(t-{I}^n_i(t)).
\end{equation*}

\noindent \textbf{The coupling.} 
 By the definition of the $S$-system, its dynamics is completely determined once we specify the primitive processes $\mathcal{A}^n$ and $\{{\mathcal{T}}(k)\}$, as well as the service rate $\mu^n$. Afterwards, we couple it with the $K$-system. To this end, fix $\mu^n$ and the sequence of i.i.d.~random variables $\{{\mathcal{T}}(k)\}$ with mean 1 and variance $\sigma^2$. Set
 \begin{equation*} 
{\mathcal{A}}^n(t)=A\Big(\int_0^t\la^n(s)ds\Big), \quad t\ge0,
\end{equation*}
where $A$ and $\la^n(t)$ are defined the same as in \eqref{eq:arrival}. Now, couple the $K$-system to the $S$-system as follows: 
 
 \begin{enumerate}
     \item Set $\sum_{i=1}^K\mu^n_i=\mu^n$, i.e., the single server works at the combined rate of the servers in the $K$-system.
     \item Set $A^n(t)={\mathcal{A}}^n(t)$, i.e., both systems have the same job arrival processes.
     \item For job $k$ that enters the $K$-system and is sent to server $i_k\in[K]$, such that it is the $l$th job to be sent to server $i_k$, set $T_{i_k}(l)={\mathcal{T}}(k)$. Namely, each job incurs the same nominal service requirement in both systems. Note that its actual service time depends on which server it is routed to.
 \end{enumerate}

\begin{remark}
The last item in the coupling means that for the $K$-system, instead of sampling the primitive processes $\{T_i(k)\}$ for the nominal service requirements, we sample the single process $\{{\mathcal{T}}(k)\}$. Since all of these sequences are assumed to be i.i.d.~(within, as well as jointly), these two approaches are equivalent and result in the exact same stochastic behaviour. 
\end{remark}
We are now ready to prove our results from Section \ref{subsec:main results}.

\begin{reptheorem}{thm:workload lower bound}[General Lower Bound]
For any $K$-server load balancing system we have:

\begin{equation}\label{eq: general lb}
    \sum_{i=1}^K W_i^n(t) \geq {\mathcal{W}}^n(t), \qquad t \geq 0.
\end{equation}
\end{reptheorem}

\begin{proof}
The proof is by induction. Let $\tau_k$ denote the (random) arrival time of the $k$th job. Equation \eqref{eq: general lb} holds trivially for $t\in[0,\tau_1)$ (since the systems are empty). Assume that it holds for $t\in [0,\tau_k)$. We want to prove that this implies that \eqref{eq: general lb} also holds for $t \in [0,\tau_{k+1})$. At time $\tau_k$ a job arrives to the two systems with the same nominal workload ${\mathcal{T}}(k)$. By the induction assumption, we must have 

\begin{equation*}
    \sum_{i=1}^K W_i^n(\tau_k^-) \geq {\mathcal{W}}^n(\tau_k^-), 
\end{equation*}
and therefore
\begin{equation}\label{eq:for gen lb}
    \sum_{i=1}^K W_i^n(\tau_k)=\sum_{i=1}^K W_i^n(\tau_k^-)+ {\mathcal{T}}(k)\geq {\mathcal{W}}^n(\tau_k^-)+{\mathcal{T}}(k)={\mathcal{W}}^n(\tau_k).
\end{equation}

Now, during $(\tau_k,\tau_{k+1})$ no jobs arrive. Hence, the nominal workloads can only decrease, possibly reaching (and staying) at zero. The rate at which the total nominal workload in the $K$-system decreases is at most $\mu^n$ (when all servers are non-idle). Thus the nominal workload in the $S$-system starts lower at $\tau_k$ (by \eqref{eq:for gen lb}) and decreases at least as fast. This concludes the proof.
\qed\\
\end{proof}

Thus, there is no load balancing policy for the $K$-server system that can do better than the corresponding single server system. Our next result shows, under a mild boundedness assumption on the nominal service requirements, that any policy under which SDDP occurs is asymptotically optimal. Namely, the scaled difference between the nominal workload compared to the $S$-system is of order $o(n^{1/2})$. 


\begin{reptheorem}{thm:as optimal}[Asymptotic Optimality]
Fix $T>0$. Assume that the service time requirements of all arriving jobs are uniformly upper bounded by some constant $c>0$. If the queue lengths undergo SDDP (property \eqref{property:ssc}), then:
\begin{equation*}
       \frac{1}{n^{1/2}}\left\|\sum_{i=1}^K W^n_i-{\mathcal{W}}^n\right\|_T\xrightarrow{n \rightarrow \infty}0 \text{ in probability.}
\end{equation*}
\end{reptheorem}

\begin{proof}
Fix $T$ and $\epsilon$. Define
\begin{align*}
    \tau^n:=\inf\{t\geq 0: \sum_{i=1}^K W^n_i(t)-{\mathcal{W}}^n(t)>\epsilon n^{1/2}\}.
\end{align*}
By \eqref{eq: general lb}, the result will follow once we prove that $\mathbbm{P}(\tau^n<T)\rightarrow 0$ as $n \rightarrow \infty$.
To this end, define
\begin{equation*}
    \sigma^n=\sup{\{t<\tau^n:\sum_{i=1}^K W^n_i(t)-{\mathcal{W}}^n(t)\leq\frac{\epsilon}{2}n^{1/2}\}}.
\end{equation*}
The existence of $\sigma^n$ is guaranteed by the fact that both systems start empty. In addition, since the upward jumps are equal in both systems, by the definitions of $\tau^n$ and $\sigma^n$, we must have $\sigma^n<\tau^n$.
Now, 
\begin{align*}
    &\sum_{i=1}^K W^n_i(\tau^n)-{\mathcal{W}}^n(\tau^n)>\epsilon n^{1/2}\cr
    &\sum_{i=1}^K W^n_i(\sigma^n)-{\mathcal{W}}^n(\sigma^n)\leq\frac{\epsilon}{2}n^{1/2}.
\end{align*}
Denoting $I^n:=[\sigma^n,\tau^n]$, we obtain
\begin{align}\label{eq: combined}
    \sum_{i=1}^K W^n_i(I^n]-{\mathcal{W}}^n(I^n]>\frac{\epsilon}{2}n^{1/2}.
\end{align}
We now argue that for all large enough $n$, all of the servers in the $K$-system must be non-idle during the interval $I^n$. By the definition of $\sigma^n$, we have $\forall t \in I^n$
\begin{equation}\label{eq:bounded 1}
    \sum_{i=1}^K W^n_i(t) \geq \frac{\epsilon}{2}n^{1/2}.
\end{equation} 
Since we assumed that the nominal service requirements of jobs are uniformly bounded by a constant $c>0$, each job in the queue (including the job in processing, if there is any) contributes at most $c$ to the total nominal workload. Therefore $\forall t \geq 0$ we have
\begin{equation}\label{eq:bounded 2}
    c\sum_{i=1}^K Q^n_i(t)\geq \sum_{i=1}^K W^n_i(t).
\end{equation} 
Combining \eqref{eq:bounded 1} and \eqref{eq:bounded 2} yields
\begin{equation}\label{eq:bounded 3}
    \sum_{i=1}^K Q^n_i(t)\geq \frac{\epsilon}{2c}n^{1/2}, \qquad t \in I^n.
\end{equation} 

Now, for any non decreasing sequence $a_1,\ldots,a_K$ we have:

$$\sum_{i=1}^K a_i=\sum_{i=1}^K(a_1+a_i-a_1)\leq Ka_1+(K-1)(a_K-a_1).$$

Thus, we obtain:
\begin{equation}\label{eq:bounded new}
    \sum_{i=1}^K Q^n_i(t)\leq K\min_i\{Q^n_i(t)\}+(K-1)\max_{i,j}|Q^n_i(t)-Q^n_j(t)|
\end{equation}

By the SDDP assumption, for large enough $n$, we have $\forall t \geq 0$
\begin{equation}\label{eq:bounded 4}
    \max_{i,j}\vert Q^n_i(t)-Q^n_j(t) \vert \leq \frac{1}{K-1} \frac{\epsilon}{4c}n^{1/2}.
\end{equation}
Combining Equations \eqref{eq:bounded 3} \eqref{eq:bounded new} and \eqref{eq:bounded 4} imply that $\min_i\{Q_i(t)\}>0$, for all $t \in I^n$, as desired.

Since the servers in the $K$-system are non-idling during $I^n$, the total nominal workload is drained at a maximal rate of $\sum_{i=1}^K \mu^n_i=\mu^n$. Namely
\begin{align}\label{eq:final step 1}
    \sum_{i=1}^K W^n_i(I^n]=\sum_{i=1}^K (W^n_{A,i}(I^n]-\mu^n_i|I^n|)={\mathcal{W}}^n_{\mathcal{A}}(I^n]-\mu^n |I^n|.
\end{align}
On the other hand, using \eqref{eq:nom w S}, we have
\begin{align}\label{eq:final step 2}
    {\mathcal{W}}^n(I^n]&={\mathcal{W}}^n(\tau^n)-{\mathcal{W}}^n(\sigma^n)={\mathcal{W}}^n_{\mathcal{A}}(\tau^n)-\mu^n(\tau^n-{\mathcal{I}}^n(\tau^n))-({\mathcal{W}}^n_{\mathcal{A}}(\sigma^n)-\mu^n(\sigma^n-{\mathcal{I}}^n(\sigma^n)))\cr
    &={\mathcal{W}}^n_{\mathcal{A}}(I^n]-\mu^n|I^n|+\mu^n{\mathcal{I}}^n(I^n]\geq {\mathcal{W}}^n_{\mathcal{A}}(I^n]-\mu^n|I^n|,
\end{align}
where in the last inequality we used the fact that the idle process has non decreasing sample paths.

Combining Equations \eqref{eq:final step 1} and \eqref{eq:final step 2} yields
\begin{align*}
    \sum_{i=1}^N W^n_i(I^n]-{\mathcal{W}}^n(I^n]\leq {\mathcal{W}}^n_A(I^n]-\mu^n|I^n|-( {\mathcal{W}}^n_{\mathcal{A}}(I^n]-\mu^n|I^n|)=0,
\end{align*}
which contradicts Equation \eqref{eq: combined}. We conclude that for all $n$ large enough we have $\mathbbm{P}(\tau^n<T)=0$, concluding the proof.
\qed \\
\end{proof}

\section{Connecting Communication to Performance}\label{sec: connecting comm to per}
Finally, we are ready to connect the communication-approximation and approximation-performance links for studying the communication-performance link. We will prove that all of the approaches we developed above, namely, RT-$r$, DT-$x$ or ET-$x$, combined with the basic, MSR or MSR-$x$ approximation algorithms, result in admissible approximations provided that their communication rate is of order $\omega(n^{1/2})$. By Theorem \ref{thm:SDDP}, this implies that the queue lengths undergo SDDP, which in turn, by Theorem \ref{thm:as optimal}, implies that asymptotically optimal performance is achieved. The main take-away is that asymptotically optimal performance can be achieved in the model under consideration with a communication rate that can be significantly slower than the rate events occur in the system, i.e., of order $\Theta(n)$.

As a reminder, the basic, MSR, and MSR-$x$ algorithms (see Definitions \ref{def:basic}, \ref{def:MSR} and \ref{def:MSR-x} respectively) are special cases of the queue length emulation approach (see Definition \ref{def:qls}) for the approximation algorithm. RT-$r$, DT-$x$ and ET-$x$ (see Definitions \ref{def:rt}, \ref{def:dt} and \ref{def:et} respectively) are the communication patterns of the servers. 
For simplicity, we assume that under RT all servers send messages at the same rate $r$, and with a slight abuse of notation write RT-$r$ for this communication pattern.  

\begin{theorem}\label{thm:connecting}
Fix sequences $r^n \in \mathbbm{R}_+$ and $x^n,y^n\in\mathbbm{N}$. If:
\begin{enumerate}
    \item RT-$r^n$ and the basic, MSR or MSR-$x^n$ approximation algorithms are used where $r^n \in \omega(n^{1/2})$, or
    \item DT-$x^n$ and either the basic or MSR-$y^n$ approximation algorithms are used where $x^n,y^n \in o(n^{1/2})$, or
    \item ET-$x^n$ and \textbf{any} queue length emulation based approximation algorithm (which includes, but not limited to, the basic, MSR, and MSR-$y^n$ algorithms) are used, where $x^n \in o(n^{1/2})$,
\end{enumerate}
Then: the resulting approximations are admissible and \eqref{property:aoa} holds.

\end{theorem}

The implications of Theorem \ref{thm:connecting} for DT-$x$ and ET-$x$ are the following. Consider the long term rate of communication, given by $M^n(T)/T$. 
Suppose for simplicity that $x_n=y_n=n^{1/2-\epsilon}$, where $0<\epsilon<0.01$. By Propositions \ref{prop:dt comm guarnatee}, \ref{prop:et gurantees} and \ref{prop:et msr gurantees}, 
for any combination of DT-$x_n$ or ET-$x_n$ with the basic, MSR or MSR-$x_n$ approximation algorithms, we have:
$$M^n(T)\leq \frac{1}{x_n}(S^n(T)+Cn)=\frac{1}{x_n}\big(\sum_i S_i(\bar{\mu}_i nT)+Cn\big)=n^{1/2+\epsilon}\big(\sum_i S_i(\bar{\mu}_i nT)/n+C\big),$$
where $C>0$ is a constant and the first equality is due to \eqref{eq: pot serv proc}. Since $S_i$ is a renewal process, $S_i(\bar{\mu}_i nT)/n$ converges to a constant almost surely. Therefore, $M^n(T)/T$ is of order, for example, $o(n^{1/2+2\epsilon})$, much smaller than than the communication rate needed for full state information, which is of order $\Theta(n)$.\\

\begin{proof}

\paragraph{Proof for RT-$r^n$.}   Fix $T>0$ and $\epsilon>0$. Define 

\begin{equation}\label{eq:tau rt}
    \tau^n=\inf\{t \geq 0: \exists i\in [K] \text{ s.t. } |\tilde{Q}^n_{i}(t)-Q^n_{i}(t) |> \epsilon n^{1/2}\}
\end{equation}
The result will follow once we prove that $\mathbbm{P}(\tau^n\leq T)\rightarrow 0$ as $n \rightarrow \infty$. 

The queue $i$ is random. Appealing to the union bound, we have
\begin{align}\label{eq:const freq union bound}
&\PP(\{\tau^n\leq T\})\le\sum_{i_0=1}^{K}
\PP(\Om^n(i_0)),\quad\text{where}\\
\notag
&\Om^n(i_0)=\Big\{\text{there exists } t\in[0,T] \  \text{such that }
|\tilde{Q}^n_{i_0}(t)-Q^n_{i_0}(t) |> \epsilon n^{1/2}\Big\}.
\end{align}
To prove the result,
it thus suffices to show that each summand $\PP(\Om^n(i_0))$
converges to zero with $n$.
We focus on a fixed, deterministic $i_0$, and, with a slight abuse
of notation, refer to it again as $i$. We also write $\Om^n$ for $\Om^n(i_0)$.

Under $\Om^n$, by the definition of the approximations in \eqref{eq:approximations diffusive section}, $t$ cannot be a messaging time. Define $s$ as the last time before $t$ such that the load balancer received an update on the queue length in server $i$: 
\begin{equation}\label{eq:sigma const appr}
    s=\sup_m\{t^m_i: t^m_i<t\}.
\end{equation}

Denote the corresponding time interval by $J=(s,t]$, and note that $|J|=t-s\leq 1/r^n$. 

By \eqref{eq:approximations diffusive section} and \eqref{eq:actual ql diffusion section} we have 
\begin{equation}\label{eq:final 1}
    |\tilde{Q}^n_{i}(t)-Q^n_{i}(t) |=|D^n_i(s,t]-\tilde{D}^n_i(s,t]|\leq D^n_i(s,t]+\tilde{D}^n_i(s,t].
\end{equation}
Now, under any of the approximation algorithms we consider, the quantity $\tilde{D}^n_i(s,t]$ is at most the maximum number of emulated departures during $(s,t]$, corresponding to jobs with constant service requirements of $1/\mu^n_i$. Thus 
\begin{equation}\label{eq:final 2}
    \tilde{D}^n_i(s,t]\leq \mu_i^n(t-s)\leq \mu_i^n/r^n.
\end{equation}
As for $D^n_i(s,t]$, we have
\begin{align}\label{eq:final 3}
    D^n_i(s,t]=D^n_i(t)-D^n_i(s)=S^n_i(t-I^n_i(t))-S^n_i(s-I^n_i(s))&\leq S^n_i(t-I^n_i(s))-S^n_i(s-I^n_i(s)),
\end{align}
where we have used \eqref{eq:D vs S} and the fact that the idle time process $I^n_i$ and the potential service process $S_i^n$ have non-decreasing sample paths.

Combining \eqref{eq:const freq union bound}, \eqref{eq:final 1}, \eqref{eq:final 2} and \eqref{eq:final 3}, we conclude that under $\Om^n$ we have
\begin{align}\label{eq:final 4}
    S^n_i(t-I^n_i(s))-S^n_i(s-I^n_i(s))+\mu^n_i/r^n\geq D^n_i(s,t]+\tilde{D}^n_i(s,t]\geq \vert \tilde{Q}^n_{i}(t)-Q^n_i(t)\vert >\epsilon n^{1/2}.
\end{align}
Now, by the scaling in \eqref{eq: sclaed primitives}, we have:

$$S^n_i(t-I^n_i(s))=n^{1/2}\hat{S}^n_i(t-I^n_i(s))+\mu^n_i(t-I^n_i(s))$$
and 
$$S^n_i(s-I^n_i(s))=n^{1/2}\hat{S}^n_i(s-I^n_i(s))+\mu^n_i(s-I^n_i(s)).$$
Therefore
\begin{align}\label{eq:final 5}
    S^n_i(t-I^n_i(s))-S^n_i(s-I^n_i(s))&=n^{1/2}(\hat{S}^n_i(t-I^n_i(s))-\hat{S}^n_i(s-I^n_i(s)))+\mu^n_i(t-s)\cr
    &\leq n^{1/2} w_T(\hat{S}^n_i,1/r^n)+\mu^n_i/r^n,
\end{align}
where $w_T$ is the modulus of continuity and we have used the fact that 
$$t-I^n_i(s)-(s-I^n_i(s))=t-s\leq 1/r^n.$$

Combining \eqref{eq:final 4} and \eqref{eq:final 5}, dividing by $n^{1/2}$ and using the fact that $\mu^n_i=\bar{\mu}_in$, we obtain that under $\Om^n$ we have
\begin{align}\label{eq:final 6}
    w_T(\hat{S}^n_i,1/r^n)+2\bar{\mu}_i n^{1/2}/r^n >\epsilon .
\end{align}

Due to the $\mathcal{C}$-tightness of $\hat{S}^n_i$, and the fact that $r^n \in \omega(n^{1/2})$ and therefore $n^{1/2}/r^n$ and $1/r^n$ both converge to zero, we have that the left hand side of \eqref{eq:final 6} converges to zero as $n$ tends to infinity. Therefore $\mathbbm{P}(\Om^n)\rightarrow 0$ as $n \rightarrow \infty$, which, by \eqref{eq:const freq union bound}, concludes the proof. 

\qed 


\paragraph{Proof for DT-$x^n$.} 
Let $\tau^n$, $\Om^n$ and $t$ be defined as in \eqref{eq:tau rt} and \eqref{eq:const freq union bound}. Again, $t$ cannot be a messaging time by the definition of $\Om^n$. Let $s$ be the last messaging time before $t$, as defined in \eqref{eq:sigma const appr}.

Using Equation \eqref{eq:final 1}, we have
$$|\tilde{Q}^n_{i}(t)-Q^n_{i}(t) |=|D^n_i(s,t]-\tilde{D}^n_i(s,t]|\leq D^n_i(s,t] \vee \tilde{D}^n_i(s,t]\leq x^n \vee \tilde{D}^n_i(s,t],$$
where the last inequality is by the definition of the DT-$x^n$ communication pattern. Now, under the basic and MSR-${y^n}$ algorithms,  $\tilde{D}^n_i(s,t]=0$ and $\tilde{D}^n_i(s,t] \leq y^n-1$, respectively. Therefore, for these algorithms, we have 
$$\mathbbm{P}(\Om^n)\leq \mathbbm{P}(x^n \vee y^n>\epsilon n^{1/2}) \xrightarrow{n \rightarrow \infty}0,$$
where the convergence is due to the fact that $x^n, y^n\in o(n^{1/2})$.

\qed 
\paragraph{Proof for ET-$x^n$.} 
By the definition of the ET-$x^n$ communication pattern, the approximation error never exceeds $x^n$. Therefore
$$\mathbbm{P}(\Om^n)\leq \mathbbm{P}(x^n >\epsilon n^{1/2}) \xrightarrow{n \rightarrow \infty}0,$$
where the convergence is due to the fact that $x^n\in o(n^{1/2})$.

\qed 
\end{proof}
\section{Simulations}\label{sec:simulations}

In this section we present simulation results that support the theoretical findings of this paper. They also provide insights into the relationship between communication, approximation and performance under the different algorithms we proposed. Our main goal is to demonstrate that it is possible to achieve good performance using very sparse communication. We begin by describing the simulation setting and the metrics that we use. We proceed with studying the communication-approximation trade-off, followed by studying the performance given a communication constraint, with an emphasis on the performance of ET-$x$ combined with MSR.

\subsection{Simulation Setting}

We simulate a discrete time system with $K$ servers and a single load balancer. At each time slot $t \in \mathbbm{N}$, a job arrives to the load balancer with probability $0<\lambda<1$. The service time requirements of jobs are i.i.d.~Geometric($1/K$) random variables. Each server processes jobs using the FIFO service discipline and completes a unit of work at each time slot (if it has a job to process). Thus, at each time slot, the maximum amount of work that can be completed by the servers is $K$, and the average amount of work that enters the system equals $\lambda K$. We therefore refer to the parameter $\lambda$ as the \textit{load} on the system. \\

\noindent \textbf{Compared algorithms.} We simulate all of our proposed algorithms, namely, combinations of the approximation algorithms (basic, MSR-x and MSR) and communication patterns (RT-$r$, DT-$x$ and ET-$x$). 

Our benchmark for comparison are five conventional algorithms. At one extreme, we consider the optimal JSQ algorithm, which requires exact state information and at least one message per departure in the system. At the other extreme, we consider the state-independent Round Robin (RR) algorithm, where jobs are assigned to servers in a deterministic cyclic fashion. The third algorithm is Power of Choice (SQ(2)), where at each arrival time, two serves are chosen uniformly at random and the job is sent to the server with the shortest queue of the two. As described in the introduction, SQ(2) also requires at least one message per departure, but we consider it because of its popularity and the interest in comparing its performance to that of our proposed algorithms.

The forth and fifth algorithms are Join-the-Idle-Queue (JIQ) and Persistent-Idle (PI), which we will, for ease of exposition, discuss separately from RR and SQ(2). JIQ sends an incoming job to one of the idle servers, and if there are none, sends it to a server chosen uniformly at random. PI sends an incoming job to one of the idle servers, and if there are none,  sends it to the server to which it had sent the previous job. JIQ and PI are perceived as using a small amount of communication, but we will demonstrate that this is not the case in our setting, except for extremely high loads.

\begin{remark}
Note that the Random algorithm, where each job is assigned to a server uniformly at random, is also a state-independent candidate for comparison. However, our simulations indicate that it performs significantly worse than $RR$ across all loads, and therefore it is not shown.
\end{remark}

\noindent \textbf{Comparison method.} For fixed $K$ and $\lambda$, we simulate the different algorithms using \textit{the same input}, namely, the same job arrival times and service times requirements. Thus, the simulation of the different algorithms differs only in the load balancer routing decisions, which allows us to make a fair comparison.\\

\noindent \textbf{Comparison metrics.} We are interested in two metrics: communication amount and job completion times. 

\vspace{5pt}

\noindent \textit{Communication.} For the communication, we run long simulations of our proposed algorithms along with JSQ, and count how many messages pass in the system and compare it to how many messages JSQ uses (i.e., the number of departures). Since over a long period of time the number of departures roughly equals the number of arrivals, we would expect the long term communication rate of JSQ to equal $\lambda$. This allows us to determine the long term use of communication of our algorithms compared to the communication needed for exact state information. 

\vspace{5pt}

\noindent \textit{Job completion times.} For each job that enters the system, we measure the number of time slots that pass from its arrival to the time slot it finishes being serviced and leaves the system. We then look at the distribution of the job completion times (JCTs) over a long time period.\\

\noindent \textbf{Simulation parameters.} 
The parameters for the simulation are: $K,\lambda,x$, and $r$. For exposition purposes we focus on the following parameter choices. The number of servers is chosen to be $K=30$. The loads we consider are $\lambda\in \{0.5,0.8,0.95\}$, corresponding to moderate to high loads. The values for $x$, the maximal approximation error under DT-$x$ and ET-$x$, as well as the truncation value for MSR-x, are
$x\in\{2,3,4,5,6,7,8\}$. Finally, for fixed $K,\lambda,x$, whenever RT-$r$ is compared to DT-$x$ and/or ET-$x$, we choose $r=\frac{\lambda}{K x}$ for the communication amount to be comparable: Under DT-$x$, a message is sent every $x$ departures, and over a long time period the number of departures approximately equals the number of arrivals. Thus we expect the total message rate to be close to $\lambda/x$, which implies that each server sends messages at a rate of $\frac{\lambda}{K x}$.

\subsection{Communication vs. Maximal Approximation Error}
In this section we verify our findings from Section \ref{sec: The Communication-Approximation Link} regarding how much communication our proposed algorithms use for a given maximal approximation error guarantee. Most importantly, we wish to demonstrate the validity of Theorems \ref{thm:summary 2} and \ref{thm:summary 3}, namely, that ET-$x$ with MSR achieves a maximal approximation error of $y\in\{1,2,3,4,5,6,7\}$ ($y=x-1$), while using a fraction of the communication of JSQ that decays quadratically with $y$.

In the following experiments, we run long simulations of the different load balancing algorithms with different loads, using the same inputs, and measure the number of messages that were sent by the servers. We only show results concerning ET-$x$ combined with MSR or MSR-$x$ for the following reasons: (1) variations using RT-$r$ are not shown because their relative communication requirements are determined apriori using the value of $r$, (2) for DT-$x$, as Proposition \ref{prop:dt comm guarnatee} suggests, for a guaranteed maximal approximation error of $y=x-1$, the relative communication requirement of DT-$x$ with any approximation algorithm and any load is extremely close to $1/(y+1)$ and hence is not shown, and (3) ET-$x$ combined with the basic approximation algorithm is identical to DT-$x$ combined with the basic approximation algorithm, and so is not shown. 

Figure \ref{fig:et_msr_comm} depicts the results for ET-$x$ with MSR. We see that for all loads, the relative communication decreases quadratically with the approximation error guarantee. Moreover, we see that the upper bound we derived in Theorem \ref{thm:summary 3} is not tight and that the communication is reduced substantially more. Remarkably, for example, ET combined with MSR guarantees an error that does not exceed 2 while using less than 11$\%$ of the communication of JSQ, for all loads.

\begin{figure}[H]
\centering
\includegraphics[width=0.6\linewidth]{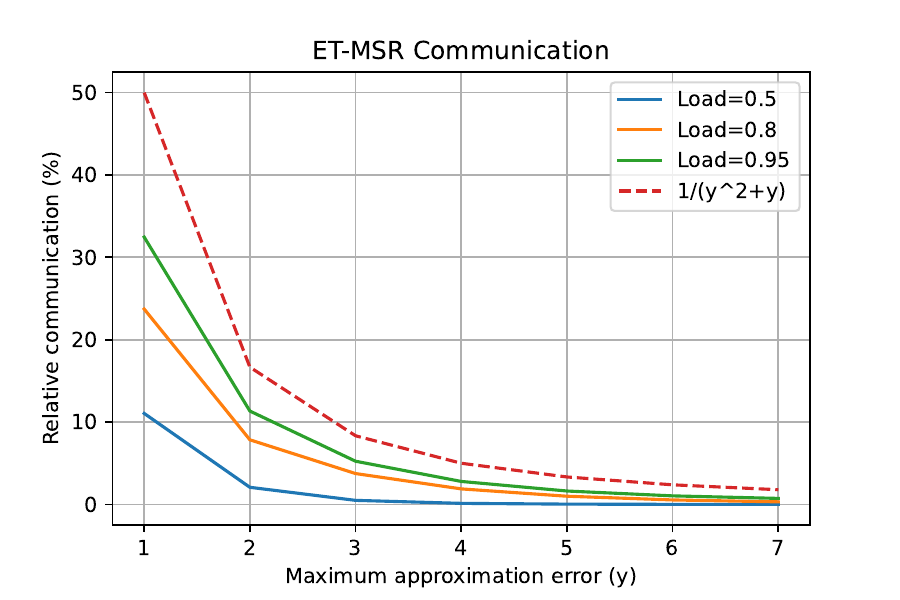}
\caption{The communication requirement of ET-$x$ with MSR for different loads and maximal approximation error $y$ ($=x-1$), relative to the communication required for full state information. The required communication decreases quadratically with $y$, and is lower than the upper bound derived in Theorem \ref{thm:summary 3}.  }
\label{fig:et_msr_comm}
\end{figure}

Figure \ref{fig:et_msrx_comm} depicts the results for ET-$x$ combined with MSR-$x$. We see that the relative communication requirements are far below the theoretical upper bound given in Proposition \ref{prop:et gurantees} for all loads, but are larger compared to ET-$x$ with MSR. Specifically, for a guaranteed maximal approximation error of $y=x-1$, the relative communication requirement of ET-$x$ with MSR-$x$ is much less than $1/(y+1)$. 

\begin{figure}[H]
\centering
\includegraphics[width=0.6\linewidth]{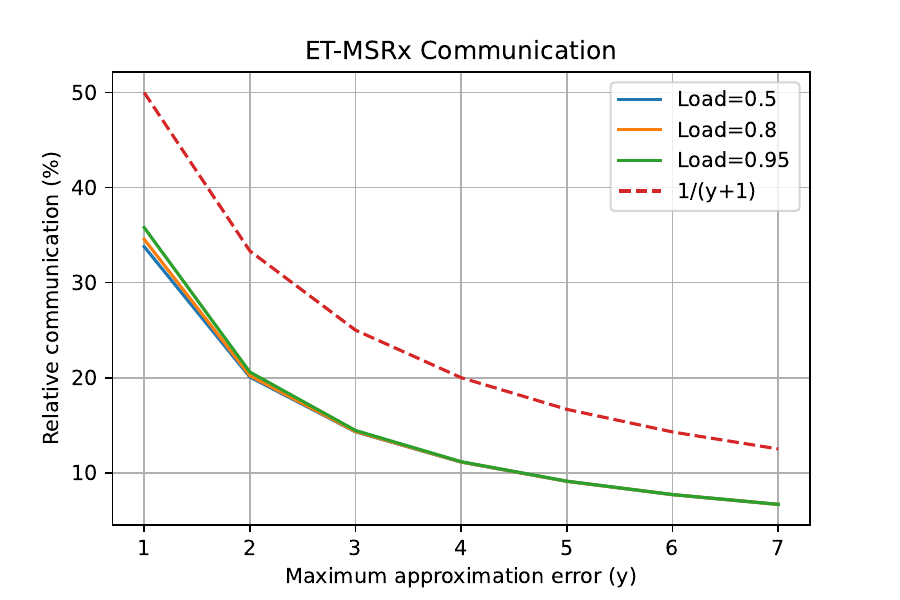}
\caption{The communication requirement of ET-$x$ with MSR-$x$ for different loads and maximal approximation error $y$ ($=x-1$), relative to the communication required for full state information. The required communication decreases is lower than the upper bound derived in Theorem \ref{thm:summary 1}, but higher than that of ET-$x$ with MSR. }
\label{fig:et_msrx_comm}
\end{figure}

\subsection{Communication vs. Performance}
In this section we focus on the best performance our algorithms achieve for different communication constraints and loads. Our main goal is to show that good performance can be achieved using sparse communication and that the ET-$x$ MSR combination is superior to the rest of the algorithms when the communication constraint is substantial. 

In the following experiments, we run long time simulations of all algorithms and all loads. For each load and algorithm, we measure the job completion times and the communication requirement compared to that of JSQ. We then consider different values for the relative communication constraint (e.g., the relative communication cannot exceed 50$\%$ of the communication JSQ uses). For each such constraint, we compare all the different algorithms that respect it using their Complement Cumulative Distribution Functions (CCDFs) of the job completion times. Luckily, in the vast majority of cases, the CCDFs mostly dominate one another (i.e., one has equal or lower values than the other for most percentiles), which allows for straightforward comparisons. The result is that for every load, there is a clear winning algorithm. Note that for all loads RT-$r$ with any approximation algorithm  is never the best performing algorithm.

Figure \ref{fig:best_0.5} depicts the results for a load of 0.5. We see that the ET-$x$ combined with MSR dominates when the communication constraint is 25$\%$ or less. We also see that it offers a significant improvement over Round Robin and even SQ(2) for constraints larger or equal to 15$\%$. For a constraint of 10$\%$ the performance is worse than that of Round Robin which uses no communication at all. 
\begin{figure}[H]
\centering
\includegraphics[width=0.6\linewidth]{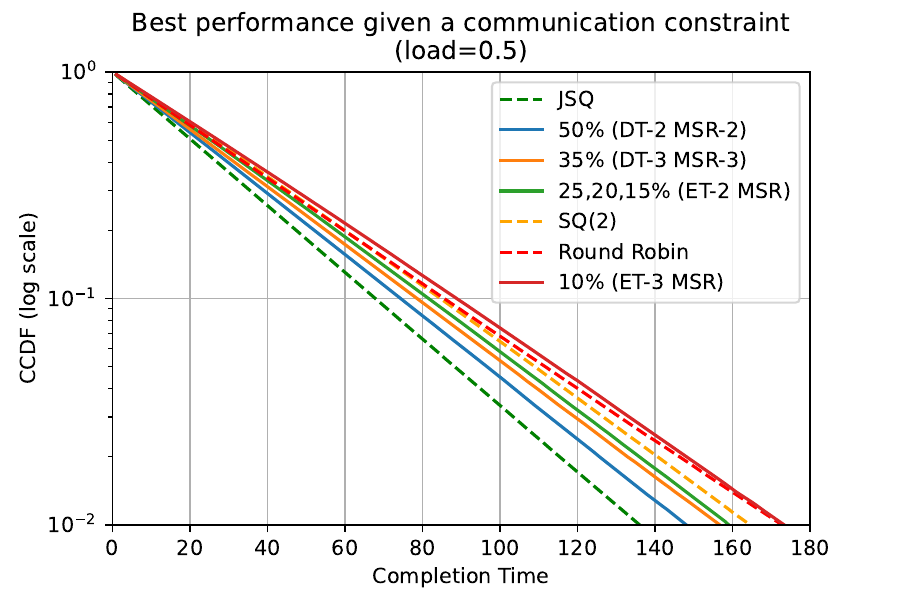}
\caption{The job completion time CCDF of the best performing algorithms for communication constraints of 10, 15, 20, 25, 35 and 50 percent of the communication required for full state information, compared to JSQ, SQ(2) and RR. ET-$x$ MSR performs best for sparse communication. DT-$x$ with MSR-$x$ is useful when more communication is allowed.  }
\label{fig:best_0.5}
\end{figure}
Figure \ref{fig:best_0.8} depicts the results for a load of 0.8. We see much more substantial gaps between the performance of JSQ and SQ(2) and SQ(2) and Round Robin than for 0.5 load. ET-$x$ combined with MSR dominates when the communication constraint is 25$\%$ and 10$\%$. Remarkably, it behaves almost identically to SQ(2) for $x=3$, while using 10$\%$ of the communication.

\begin{figure}[H]
\centering
\includegraphics[width=0.6\linewidth]{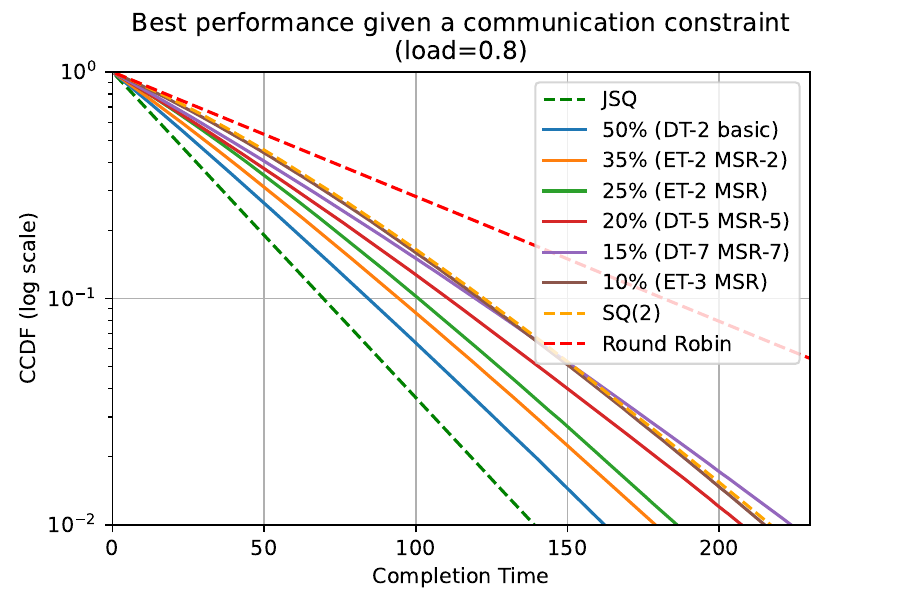}
\caption{The job completion time CCDF of the best performing algorithms for communication constraints of 10, 15, 20, 25, 35 and 50 percent of the communication required for full state information, compared to JSQ, SQ(2) and RR. ET-$x$ MSR performs best for very sparse communication. DT-$x$ with MSR-$x$ is useful when more communication is allowed. The performance of our algorithms is better than SQ(2) and RR for all constraints above 10 percent.}
\label{fig:best_0.8}
\end{figure}

We also see that there is room for improvement over Round Robin for even lower percentages of communication. We demonstrate this in Figure \ref{fig:et_msr_0.8}, which shows the performance only of ET-$x$ with MSR, which dominates this sparse communication regime. Exact values of the relative communication requirements for different $x$ values are displayed (these coincide with the values depicted previously in Figure \ref{fig:et_msr_comm}). We see that even when the communication is reduced by more than 99$\%$, we obtain a sizable improvement over Round Robin.

\begin{figure}[H]
\centering
\includegraphics[width=0.6\linewidth]{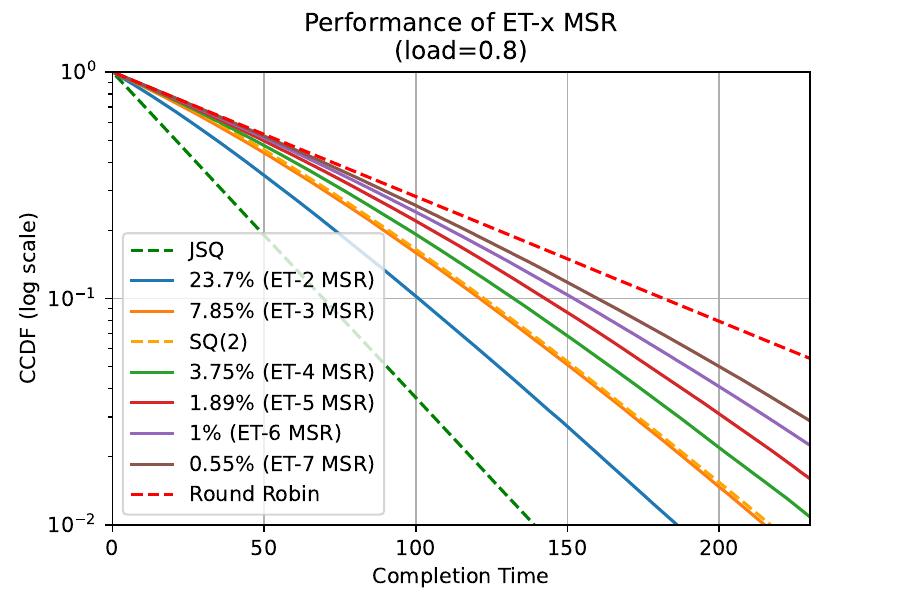}
\caption{The job completion time CCDF of ET-$x$ with MSR for $x\in\{2,\ldots,7\}$ compared to JSQ, SQ(2) and RR. The labels indicate the required relative communication. ET-$3$ with MSR is comparable to SQ(2) using only $7.85\%$ relative communication. Remarkably, ET-$x$ with MSR performs significantly better than RR even when using less than 2$\%$ relative communication.}
\label{fig:et_msr_0.8}
\end{figure}

Finally, Figure \ref{fig:best_0.95} depicts the results for a load of 0.95. We again see that ET-$x$ with MSR is the best for stricter communication constraints. Interestingly, for constraints larger or equal to 20$\%$, DT-$x$ with the basic approximation algorithm performs best, always better than SQ(2) and Round Robin. 

\begin{figure}[H]
\centering
\includegraphics[width=0.6\linewidth]{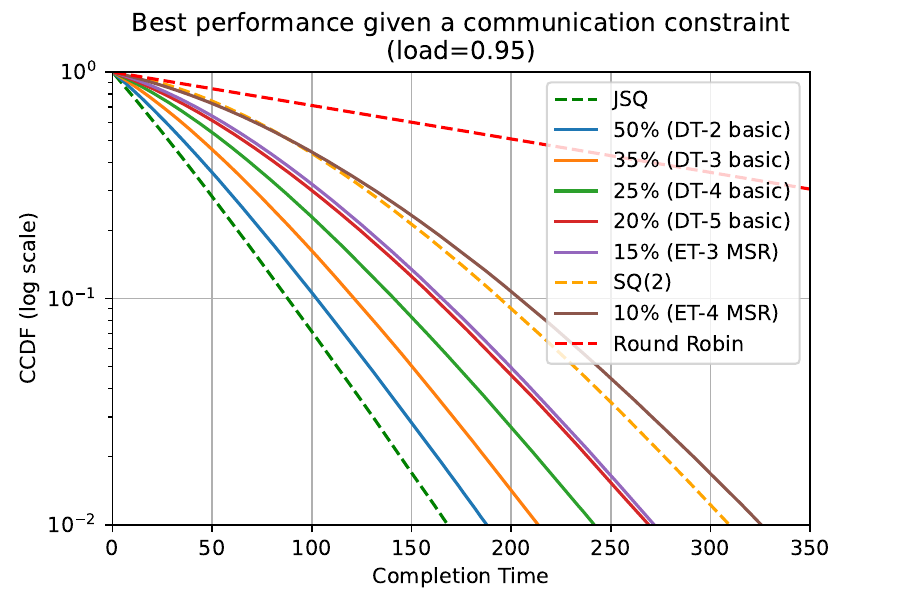}
\caption{The job completion time CCDF of the best performing algorithms for communication constraints of 10, 15, 20, 25, 35 and 50 percent of the communication required for full state information, compared to JSQ, SQ(2) and RR. ET-$x$ MSR performs best for very sparse communication. DT-$x$ with the basic approximation algorithm is useful when more communication is allowed. The performance of our algorithms is better than SQ(2) and RR for all constraints above 10 percent. }
\label{fig:best_0.95}
\end{figure}

As in the case of 0.8 load, we see a lot of room for improvement over Round Robin with even smaller communication constraints. Figure \ref{fig:et_msr_0.95} shows the performance only of ET-$x$ with MSR. We see that even when the communication is reduced by more than 99$\%$, we obtain a substantial improvement over Round Robin, more so than the improvement we observed for a load of 0.8.
\begin{figure}[H]
\centering
\includegraphics[width=0.6\linewidth]{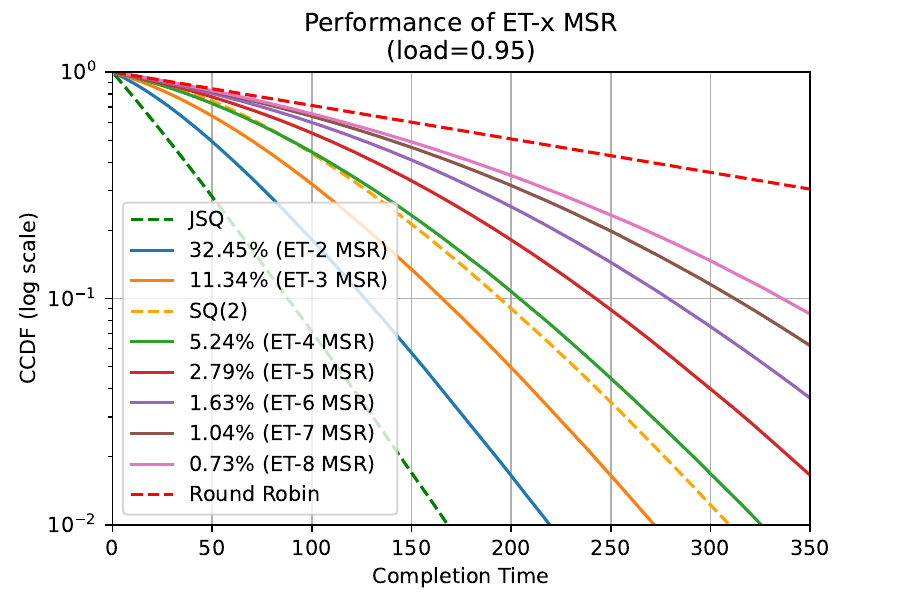}
\caption{The job completion time CCDF of ET-$x$ with MSR for $x\in\{2,\ldots,8\}$ compared to JSQ, SQ(2) and RR. The labels indicate the required relative communication. ET-$x$ with MSR performs significantly better than RR even when using less than 2$\%$ relative communication. }
\label{fig:et_msr_0.95}
\end{figure}

\subsection{Communication vs. Performance: JIQ and PI}
In this section we consider the idle-based algorithms JIQ and PI. We begin by demonstrating their required communication, which depends on the load. We then compare their performance to the algorithms we proposed, as well as to JSQ. 

\begin{figure}[H]
\centering
\includegraphics[width=0.6\linewidth]{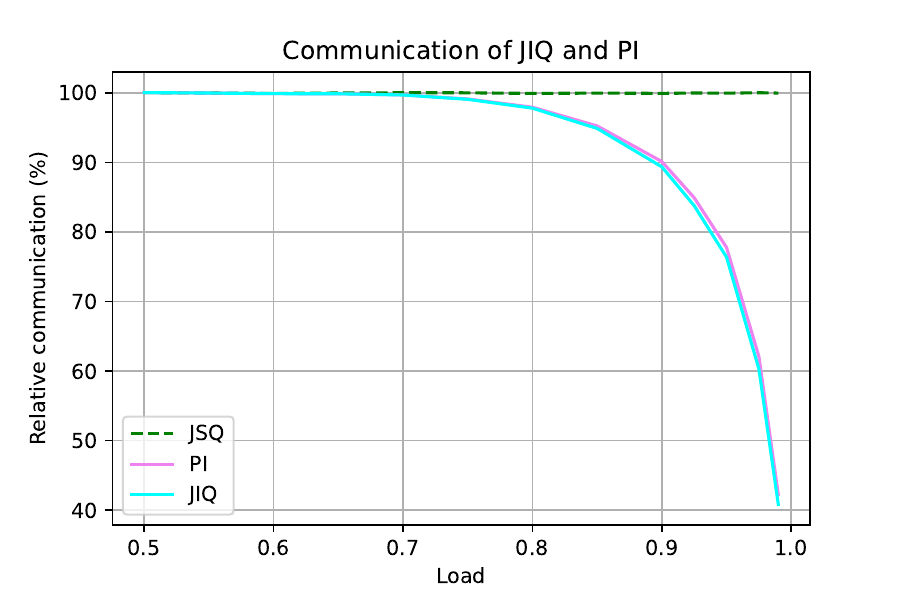}
\caption{The amount of required communication of JIQ and PI relative to that of JSQ as a function of the load in a setting with 30 servers. JIQ and PI can be regraded as using sparse communication only for extremely high loads. }
\label{fig:rev_jiq_pi_comm}
\end{figure}

Figure \ref{fig:rev_jiq_pi_comm} depicts the amount of required communication of JIQ and PI relative to that of JSQ as a function of the load in a setting with 30 servers. We can see that JIQ and PI require almost the same amount of communication for all loads and that their requirements are almost the same as that of JSQ for loads below $70\%$. In addition, their requirements drop below $50\%$ only for loads higher than approximately $98\%$. Based on these findings, as opposed to the sparse communication algorithms we propose, JIQ and PI are not suitable algorithm choices for most loads when the communication constraint is substantial. 

Given the above discussion, it would be more informative to compare the performance of JIQ and PI against our proposals only under higher loads, where their respective communication intensities are meaningfully lower than that of JSQ.  Figure \ref{fig:rev_jiq_pi_per_95} depicts the performance in terms of job completion times of JIQ and PI, as well as JSQ, Round Robin and the relevant approximation based algorithms we proposed, namely DT-2 basic and ET-2 MSR, for a load of $0.95$. JIQ and PI use around $70\%$ communication, DT-2 basic $50\%$ and ET-2 roughly $35\%$. We can see that while DT-2 uses less communication than JIQ and PI, it is competitive in terms of performance. In addition, the CCDF of JIQ and PI have larger tails.

\begin{figure}[H]
\centering
\includegraphics[width=0.6\linewidth]{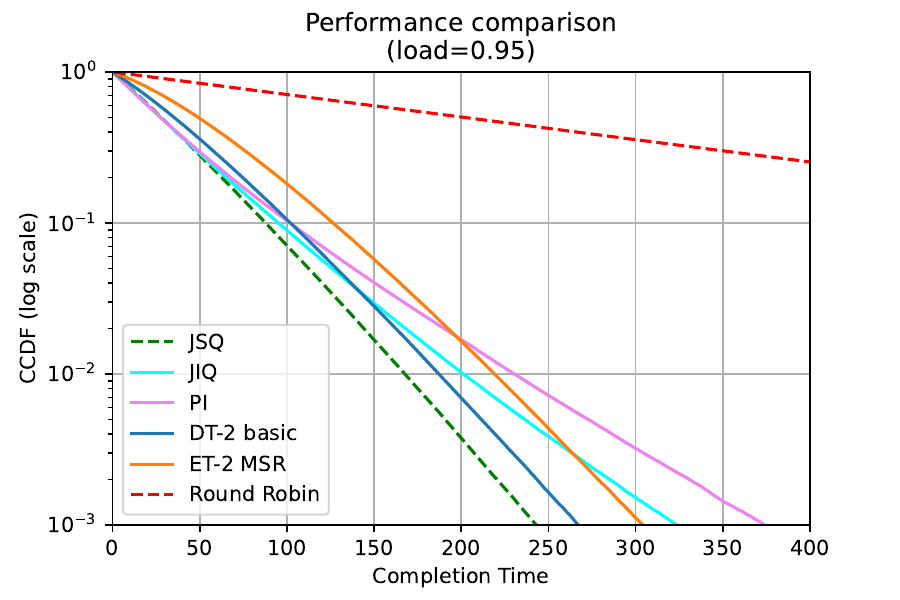}
\caption{The job completion time CCDF of JSQ, JIQ, PI, DT-2 basic, ET-2 MSR and Round Robin for a load of 0.95. DT-2 uses less communication than JIQ and PI and is competitive in terms of performance. In addition, the CCDF of JIQ and PI have larger tails. }
\label{fig:rev_jiq_pi_per_95}
\end{figure}

Figure \ref{fig:rev_jiq_pi_per_95} depicts the performance for a load of $0.95$. JIQ and PI use around $48\%$ communication, DT-2 basic $50\%$ and ET-2 roughly $35\%$. We can see that while DT-2 uses roughly the same amount of communication as JIQ and PI, it is better in terms of performance. ET-2 MSR uses less communication but performs well. In addition, the CCDF of JIQ and PI have much larger tails.

\begin{figure}[H]
\centering
\includegraphics[width=0.6\linewidth]{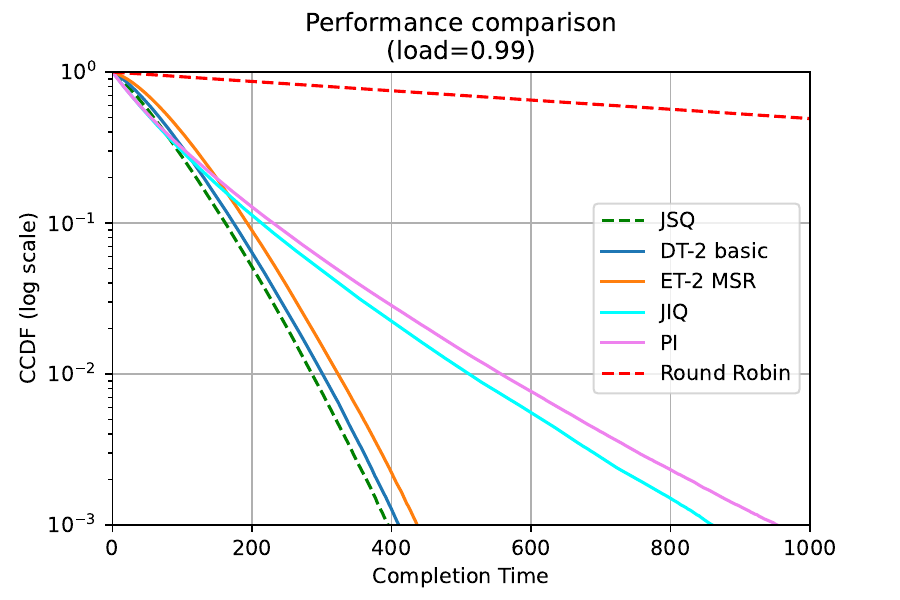}
\caption{The job completion time CCDF of JSQ, DT-2 basic, ET-2 MSR, JIQ, PI,  and Round Robin for a load of 0.95. DT-2 performs better than JIQ and PI. ET-2 MSR performs well, and both have much smaller tails than JIQ and PI.  }
\label{fig:rev_jiq_pi_per_99}
\end{figure}


\bibliography{mybib}
\bibliographystyle{apalike}

\end{document}